\documentclass[11pt]{article}
\usepackage[paper=a4paper, margin=1in]{geometry}

\usepackage{ifthen}

\ifthenelse{\isundefined{\GenerateShortVersion}}
{
\def\LongVersion{}
\def\LongVersionEnd{}
\long\def\ShortVersion#1\ShortVersionEnd{}
}
{
\def\ShortVersion{}
\def\ShortVersionEnd{}
\long\def\LongVersion#1\LongVersionEnd{}
}

\newcommand{\Ignore}[1]{\ignorespaces}

\usepackage{amsmath, amssymb}

\usepackage{amsthm}

\usepackage{ifpdf}

\usepackage{authblk}

\ifpdf
\usepackage[pdftex]{graphicx}
\else
\usepackage[dvips]{graphicx}
\fi

\usepackage[%
  ocgcolorlinks,%
  linkcolor=blue,%
  filecolor=blue,%
  citecolor=blue,%
  urlcolor=blue]{hyperref}

\usepackage{enumitem}

\usepackage{xcolor}

\usepackage{mdframed}

\LongVersion 
\LongVersionEnd 

\ShortVersion 
\usepackage{times}
\ShortVersionEnd 

\ShortVersion 
\renewcommand{\paragraph}[1]{\par\noindent\textbf{#1}}
\ShortVersionEnd 

\newtheorem{theorem}{Theorem}[section]
\newtheorem{lemma}[theorem]{Lemma}
\newtheorem{observation}[theorem]{Observation}
\newtheorem{corollary}[theorem]{Corollary}
\newtheorem{proposition}[theorem]{Proposition}

\theoremstyle{definition}
\newtheorem*{definition*}{Definition}
\newtheorem*{example*}{Example}
\theoremstyle{plain}

\theoremstyle{definition}

\theoremstyle{plain}

\usepackage{caption}
\usepackage{subcaption}

\usepackage{tikz}
\usetikzlibrary{shapes.misc, positioning, patterns}
\tikzstyle{cycrec} = [draw=black!40, rounded rectangle, minimum width=#1 cm,minimum height=1.2cm]
\tikzstyle{blackvertex}=[circle, draw, fill, inner sep=0pt, minimum size=6pt]
\tikzstyle{whitevertex}=[circle, draw, inner sep=0pt, minimum size=6pt]
\tikzstyle{inarrow}=[<-,shorten <=1pt,>=stealth]
\tikzstyle{outarrow}=[->,shorten >=1pt,>=stealth]
\tikzstyle{vertex}=[circle, draw, inner sep=2pt, minimum size=12pt]
\newcommand{\black}{\node[blackvertex]}
\newcommand{\white}{\node[whitevertex]}

\newcommand{\Integers}{\mathbb{Z}}
\newcommand{\Reals}{\mathbb{R}}
\newcommand{\Ex}{\mathbb{E}}
\newcommand{\Prtk}{\mathcal{P}}
\newcommand{\length}{\operatorname{len}}
\newcommand{\lenUtil}{\lambda}
\newcommand{\Wish}{\mathit{W}}
\newcommand{\WishVecSpace}{\mathcal{W}}
\newcommand{\utility}{\mathit{u}}

\newcommand{\ExcSet}{\Pi}
\newcommand{\SW}{\mathrm{SW}}
\newcommand{\Cycles}{\mathcal{C}}

\newcommand{\Weight}{\mathit{w}}

\newcommand{\GreedyAlg}{\ensuremath{\mathtt{Greedy}}}
\newcommand{\LocalSearchAlg}{\ensuremath{\mathtt{LS}}}
\newcommand{\NonUniformAlg}{\ensuremath{\mathtt{NU}}}
\newcommand{\IterativeOpt}{\ensuremath{\mathtt{IO}}}
\newcommand{\criticalRatio}{\rho}

\newcommand{\Neighbors}{\mathit{N}}
\newcommand{\Degree}{\operatorname{deg}}
\newcommand{\Alg}{\ensuremath{\mathtt{Alg}}}
\newcommand{\Opt}{\ensuremath{\mathtt{OPT}}}
\newcommand{\Agents}{\alpha}
\newcommand{\CycleGraphSet}{\mathcal{CG}}

\newcommand{\Sect}{Sec.}
\newcommand{\Thm}{Thm.}
\newcommand{\Lem}{Lem.}
\newcommand{\Obs}{Obs.}
\newcommand{\Prop}{Prop.}

\begin{document}

\title{Barter Exchange with Bounded Trading Cycles}

\author{Yuval Emek}
\author{Matan-El Shpiro}
\affil{Technion – Israel Institute of Technology}

\date{}

\begin{titlepage}

\maketitle

\begin{abstract}
Consider a \emph{barter exchange} problem over a finite set of agents, where
each agent owns an item and is also associated with a (privately known)
\emph{wish list} of items belonging to the other agents.
An outcome of the problem is a (re)allocation of the items to the agents such
that each agent either keeps her own item or receives an item from her
(reported) wish list, subject to the constraint that the length of the trading
cycles induced by the allocation is up-bounded by a prespecified \emph{length
bound} $k$.
The utility of an agent from an allocation is $1$ if she receives an item from
her (true) wish list and $0$ if she keeps her own item (the agent incurs a
large dis-utility if she receives an item that is neither hers nor belongs to
her wish list).

In this paper, we investigate the aforementioned barter exchange problem from
the perspective of mechanism design without money, aiming for truthful (and
individually rational) mechanisms whose objective is to maximize the
\emph{social welfare}.
As the construction of a social welfare maximizing allocation is
computationally intractable for length bounds
$k \geq 3$,
this paper focuses on (computationally efficient) truthful mechanisms that
approximate the (combinatorially) optimal social welfare.
We also study a more general version of the barter exchange problem, where the
utility of an agent from participating in a trading cycle of length
$2 \leq \ell \leq k$
is $\lenUtil(\ell)$, where $\lenUtil$ is a general (monotonically
non-increasing) \emph{length function}.
Our results include upper and lower bounds on the guaranteed approximation
ratio, expressed in terms of the length bound $k$ and the length function
$\lenUtil$.
On the technical side, our main contribution is an algorithmic tool that can
be viewed as a truthful version of the \emph{local search} paradigm.
As it turns out, this tool can be applied to more general (bounded size)
coalition formation problems.
\end{abstract}

\thispagestyle{empty}

\end{titlepage}

\setcounter{page}{1}

\section{Introduction}
\label{section:introduction}
Barter markets, where players exchange goods without involving any payments,
are a fundamental form of trade that dates back to the origin of mankind.
Modern forms of barter markets include the extensively studied kidney exchange
market
\cite{roth2004kidney,
roth2005pairwise,
hatfield2005pairwise,
roth2007efficient,
unver2010dynamic,
ashlagi2014free,
ashlagi2015mix,
dickerson2019failure}
as well as the emerging electronic markets for book exchange
\cite{PaperBackSwap}
and for recreational home exchange
\cite{HomeExchange, LoveHomeSwap}
(cf.\ \cite{abbassi2015exchange}).

In the current paper, we study an abstract game theoretic model for barter
exchange that was introduced by Roth, S\"{o}nmez, and \"{U}nver
\cite{roth2005pairwise} and studied further by Hatfield
\cite{hatfield2005pairwise} and in a follow up paper of Roth et al.\
\cite{roth2007efficient}, as well as serving as the cornerstone for various
other game theoretic models studied subsequently (refer to
\Sect{}~\ref{section:related-work} for a detailed discussion).
This model is defined over a set of agents, each holding an item that she
wishes to exchange for one item from a certain subset of items held by the
other agents.\footnote{%
Roth et al.\ use kidney exchange terminology in their model description.
We prefer the more ``neutral'' agent-item terminology since we do not regard
kidney exchange as this paper's main application.}
The agents report their item wish lists to a mechanism that determines the
outcome in the form of a collection of trading cycles, effectively
reallocating the items among the agents.

A fundamental characteristic of the model of \cite{roth2005pairwise}
is the
$\{ 0, 1 \}$-preferences
assumption, stating that the agents are indifferent about which item they
receive as long as this item belongs to their wish list.
Specifically, an agent that ends up with an item from her wish list enjoys a
unit utility, whereas her utility is zero if she retains her original item
(Roth et al.\ assume that the agent cannot receive a non-wished item).
Among other aspects of this barter exchange model, Roth et al.\
\cite{roth2005pairwise} investigate the design of truthful (strategy proof)
mechanisms with the objective of maximizing the social welfare.

Motivated by real-life difficulties in implementing long trading cycles, the
authors of
\cite{roth2005pairwise, hatfield2005pairwise, roth2007efficient}
advocate for mechanisms that generate trading cycles of bounded length.
The extreme case of trading cycles that consist of (exactly) two agents
reduces the mechanism designer's task to that of constructing a maximum
matching in an undirected graph \cite{roth2005pairwise}, a classic
optimization problem that admits computationally efficient algorithms.
As demonstrated in \cite{roth2007efficient} though, there is a significant
social welfare gain in allowing slightly longer trading cycles, thus raising
the challenge of designing truthful mechanisms for the \emph{$k$-barter
exchange} problem, namely, the problem of constructing a collection of trading
cycles that maximizes the number of participating agents (and thus also the
social welfare) under the constraint that each trading cycle is of length at
most $k$, for small constants
$k \geq 3$.

\begin{example*}
For a concrete application, consider a home swap platform where users
register their home for specific vacation dates and report on their wish list
over other homes registered for the same dates.
The platform then determines an exchange that speciﬁes an assignment of the
users to vacation homes.\footnote{%
Although a more general abstraction of the home swap application may include
ordinal utilities or heterogeneous cardinal utilities, we feel that the
$\{ 0, 1 \}$-preferences
assumption is reasonable also for this application, where users will probably
refrain from including a home in their wish list unless it perfectly fits
their needs.}
A downside of this application is that a single user that defects and stays
in her own home leads to a chain of cancellations affecting all participants
of the corresponding trading cycle.
Since users are likely to be highly upset by such forced cancellations,
especially if they happen at the last minute, it makes sense for the
platform to bound the maximum length $k$ of the trading cycles, thus hedging
the reputation loss resulting from unsatisﬁed users.

Of particular interest is the
$k = 3$
case that comes with the desirable feature that each user is acquainted with
all other users in her trading cycle (they are assigned to her home or she is
assigned to theirs).
In this case, a last minute cancellation is attributed to a person the user is
actually aware of (rather than some obscure defector up the cancellation
chain), hopefully preventing resentment towards the platform.
\end{example*}

Unfortunately, as proved by Abraham, Blum, and Sandholm
\cite{abraham2007clearing}, solving the aforementioned $k$-barter exchange
problem optimally is computationally intractable for any
$k \geq 3$,
irrespective of truthfulness considerations.
Consequently, the goal of this paper is to design efficient truthful
mechanisms for the $k$-barter exchange problem,
$k \geq 3$,
that approximate the (combinatorially) optimal solution.
Our main technical contribution is a novel framework for the design of
truthful versions of \emph{local search} algorithms for a related graph
theoretic problem.
Employing this framework, we develop the first efficient truthful $k$-barter
exchange mechanism with a non-trivial approximation ratio for every constant
$k \geq 3$.

While the
$\{ 0, 1 \}$-preferences
assumption of \cite{roth2005pairwise} states that an agent has no priorities
over the items in her wish list, it is often natural to assume that she
prefers to receive a wished item through shorter trading cycles over longer
ones (up to the length bound $k$).
Indeed, the longer is the trading cycle, the higher is the risk for the
exchange to break down due to defecting participants (cf.\
\cite{dickerson2019failure}).
Moreover, a trading cycle of length $2$ has the unique
``accountability advantage'' (in the aforementioned home swap application,
this means that you stay in the home of the person who stays in your home)
which is lost with longer trading cycles, so it is reasonable to assume that
agents would prioritize trading cycles of length $2$ even if longer trading
cycles are allowed.

Motivated by these insights, we wish to extend the scope of our research to a
more general barter exchange setting where each agent is associated with a
non-increasing \emph{length function}
$\lambda : \{ 2, 3, \dots, k \} \rightarrow (0, 1]$
so that her utility from receiving a wished item
via a trading cycle of length
$2 \leq \ell \leq k$
is $\lenUtil(\ell)$.
As a first step in this direction, we assume that the same length function
$\lenUtil$ is shared by all agents, introducing the
\emph{$(k, \lenUtil)$-barter
exchange} problem.
It turns out that the combinatorial aspect of this problem remains intractable
for any length function $\lenUtil$, thus we still aim for approximate
solutions.
Relying on the aforementioned local search framework, we develop efficient
truthful
$(k, \lenUtil)$-barter
exchange mechanisms with a non-trivial approximation ratio for any choice of
function $\lambda$.
On the negative side, we establish approximability lower bounds, showing that
the dependency of our mechanisms' approximation ratio on $\lenUtil$ is
optimal.

\subsection{Model and Problem}
\label{section:model}
Consider a collection of
$n \in \Integers_{> 0}$
agents indexed by the integers in $[n]$.
A \emph{wish list} of agent
$i \in [n]$
is a subset
$\Wish_{i} \subseteq [n] - \{ i \}$
of the other agents.
Let
$\WishVecSpace_{n}
=
2^{[n] - \{ 1 \}}
\times \cdots \times
2^{[n] - \{ n \}}$
denote the set of all wish list vectors, indexed by the agents in
$[n]$.

An \emph{exchange} over $[n]$ is a bijection
$\pi : [n] \rightarrow [n]$;
we say that an agent
$i \in [n]$
\emph{partakes} in the exchange $\pi$ if
$\pi(i) \neq i$,
denoting the set of agents that partake in $\pi$ by
$\Prtk_{\pi}$.
Observe that from a graph theoretic perspective, the exchange $\pi$ induces a
collection $\Cycles_{\pi}$ of disjoint directed cycles on $\Prtk_{\pi}$,
referred to as the \emph{trading cycles} of $\pi$, so that an arc
$(i, j) \in \Prtk_{\pi} \times \Prtk_{\pi}$
is included in a trading cycle if and only if
$\pi(i) = j$.\footnote{%
Unless stated otherwise, the graph theoretic term cycle refers in this paper
to a simple cycle, i.e., a cycle in which each vertex appears at most once.}
Given a wish list vector
$\Wish \in \WishVecSpace_{n}$,
we say that the exchange $\pi$, as well as the trading cycles in
$\Cycles_{\pi}$,
\emph{respect} $\Wish$ if
$\pi(i) \in \Wish_{i}$
for each
$i \in \Prtk_{\pi}$.

Given a trading cycle $c$, let
$\Agents(c) \subseteq [n]$
denote the set of agents that form $c$.
The \emph{length} of $c$ is defined as
$\length(c) = | \Agents(c) |$.
The \emph{length} of an exchange $\pi$, denoted by $\length(\pi)$,
is defined to be the maximum length of any of its trading cycles, i.e.,
$\length(\pi) = \max \{ \length(c) \, : \, c \in \Cycles_{\pi} \}$.
For an integer
$k \geq 2$,
let $\ExcSet_{n}^{k}$ denote the set of all exchanges over $[n]$ whose length at
most $k$.
Given a wish list vector
$\Wish \in \WishVecSpace_{n}$,
let
$\ExcSet_{\Wish}^{k} \subseteq \ExcSet_{n}^{k}$
denote the set of all exchanges over $[n]$ that respect $\Wish$ whose length
is at most $k$.
Likewise, let
$\Cycles_{\Wish}^{k} = \bigcup_{\pi \in \ExcSet_{\Wish}^{k}} \Cycles_{\pi}$
denote the set of all trading cycles over $[n]$ that respect $\Wish$ whose
length is at most $k$.

Given an integral \emph{length bound}
$k \geq 2$,
an instance of the \emph{$k$-barter exchange ($k$-BE)} problem over
$n \in \Integers_{> 0}$
agents is characterized by a wish list
$\Wish^{*}_{i} \subseteq [n] - \{ i \}$
for each agent
$i \in [n]$,
regarded as the agent's private information.
The outcome space of the $k$-BE problem is the collection
$\ExcSet_{n}^{k}$
of exchanges over $[n]$ of length at most $k$.
The utility of an agent
$i \in [n]$
from an exchange
$\pi \in \ExcSet_{n}^{k}$
is determined by the utility function
$\utility
:
[n] \times \ExcSet_{n}^{k} \rightarrow \{ -\infty \} \cup [0, 1]$
defined so that
\begin{equation}\textstyle
\label{equation:utility-function-uniform}
\utility(i, \pi)
\, = \,
\begin{cases}
0, & i \notin \Prtk_{\pi} \\
1, & i \in \Prtk_{\pi} \text{ and } \pi(i) \in \Wish^{*}_{i} \\
-\infty, & i \in \Prtk_{\pi} \text{ and } \pi(i) \notin \Wish^{*}_{i}
\end{cases}
\end{equation}
(cf.\
\cite{roth2005pairwise,
abbassi2015exchange}).
This allows us to define the \emph{social welfare} of $\pi$, denoted by
$\SW(\pi)$, as the sum of all agent utilities, observing that if
$\pi \in \ExcSet_{\Wish^{*}}^{k}$,
then
$\SW(\pi) = |\Prtk_{\pi}|$.

A (deterministic) \emph{mechanism} for the $k$-BE problem is a function $f$
that for any
$n \in \Integers_{> 0}$,
maps a given wish list vector
$\Wish \in \WishVecSpace_{n}$
to an exchange
$f(\Wish) \in \ExcSet_{n}^{k}$,
regarding $\Wish_{i}$ as the wish list reported to the mechanism by agent
$i \in [n]$.
Following \cite{roth2005pairwise} (see also \cite{abbassi2015exchange}), we
make the simplifying \emph{subset wish list assumption} stating that an
agent's reported wish list is always a subset of her true wish list, i.e.,
$\Wish_{i} \subseteq \Wish^{*}_{i}$
for every
$i \in [n]$.\footnote{%
In Appendix~\ref{appendix:subset-assumption}, we present a simple method that
lifts this assumption using an arbitrarily small amount of randomness.
}
The mechanism $f$ is said to be \emph{truthful} (a.k.a.\ \emph{strategy
proof}) if for every 
wish list vector
$\Wish \in \WishVecSpace_{n}$
(that satisfies
$\Wish_{j} \subseteq \Wish^{*}_{j}$
for all
$j \in [n]$)
and for every agent
$i \in [n]$,
it holds that
\[\textstyle
\utility(i, f(\Wish))
\, \leq \,
\utility \left( i, f \left( \Wish_{-i}, \Wish^{*}_{i} \right) \right)
\, ;
\]
in other words, it is a dominant strategy for agent $i$ to report her (full)
true wish list $\Wish^{*}_{i}$ to the mechanism.\footnote{%
Given a wish list vector
$\Wish \in \WishVecSpace_{n}$
and an agent
$i \in [n]$,
we adhere to the convention that $\Wish_{-i}$ denotes the vector
obtained from $\Wish$ by omitting its $i$th coordinate, that is,
$\Wish_{-i}
=
(\Wish_{1}, \dots, \Wish_{i - 1}, \Wish_{i + 1}, \dots, \Wish_{n})$.
This notation is naturally extended to agent subsets
$S \subseteq [n]$,
using $W_{-S}$ to denote the vector obtained from $\Wish$ by omitting its
$i$th coordinate for every
$i \in S$.}

To ensure that mechanism $f$ is \emph{individually rational}, we further
require that $f(\Wish)$ respects $\Wish$
(i.e.,
$f(\Wish) \in \ExcSet_{\Wish}^{k}$),
which means that the utility of the agents from $f$ is never negative.
Consequently, we can express the social welfare of the mechanism's outcome as
\begin{equation}\textstyle
\label{equation:social-welfare-uniform}
\SW(f(\Wish))
\, = \,
\sum_{i \in [n]} \utility(i, f(\Wish))
\, = \,
\sum_{c \in \Cycles_{f(\Wish)}} \length(c)
\, = \,
|\Prtk_{f(\Wish)}|
\, .
\end{equation}

We aim for truthful (and individually rational) mechanisms that maximize the
social welfare of the exchanges they produce.
To this end, we say that a truthful $k$-BE mechanism $f$ admits
\emph{approximation ratio}
$r \geq 1$
if for every
$n \in \Integers_{> 0}$
and wish list vector
$\Wish^{*} \in \WishVecSpace_{n}$,
it is guaranteed that
\[\textstyle
\SW(f(\Wish^{*}))
\, \geq \,
\frac{1}{r}
\cdot
\max \left\{
\SW(\pi) \, : \, \pi \in \ExcSet_{\Wish^{*}}^{k}
\right\}
\, ;
\]
that is, the benchmark of the approximation is the optimal exchange of length
at most $k$ that respects $\Wish^{*}$.
We typically regard the length bound $k$ as a small constant (recall our
particular interest in
$k = 3$),
whereas the number $n$ of agents grows asymptotically.
Our ultimate goal is to design truthful mechanisms whose approximation ratio
is expressed as a function of $k$ (i.e., constant approximation) and whose
running time is polynomial in $n$ (i.e., they are \emph{computationally
efficient}).

\paragraph{General Length Functions.}
As discussed earlier, we also consider a generalization of the $k$-BE problem,
referred to as 
\emph{$(k, \lenUtil)$-BE},
where
$\lenUtil : \{ 2, 3, \dots, k \} \rightarrow (0, 1]$
is a non-increasing \emph{length function}.
This generalization is obtained by redefining the utility function from
\eqref{equation:utility-function-uniform} as
\begin{equation}\textstyle
\label{equation:utility-function-general}
\utility(i, \pi)
\, = \,
\begin{cases}
0, & i \notin \Prtk_{\pi} \\
\lenUtil(\length(c)), &
i \in \Prtk_{\pi} \, ,
\pi(i) \in \Wish^{*}_{i} \, ,
\text{ and $c$ is the trading cycle in $\Cycles_{\pi}$ s.t.\
$i \in \Agents(c)$} \\
-\infty, & i \in \Prtk_{\pi} \text{ and } \pi(i) \notin \Wish^{*}_{i}
\end{cases} \, .
\end{equation}
Consequently, if
$\pi \in \ExcSet_{\Wish^{*}}^{k}$,
then the social welfare $\SW(\pi)$ of $\pi$ becomes
$\sum_{c \in \Cycles_{\pi}} \length(c) \cdot \lenUtil(\length(c))$
and the social welfare expression in
\eqref{equation:social-welfare-uniform} is now replaced by
\begin{equation}\textstyle
\label{equation:social-welfare-general}
\SW(f(\Wish))
\, = \,
\sum_{i \in [n]} \utility(i, f(\Wish))
\, = \,
\sum_{c \in \Cycles_{f(\Wish)}} \length(c) \cdot \lenUtil(\length(c))
\, .
\end{equation}
Similarly to the length bound $k$, the images
$\lenUtil(2), \lenUtil(3), \dots, \lenUtil(k)$
of the length function are also regarded as constants (part of the problem's
specification).

To see that the
$(k, \lenUtil)$-BE
problem indeed generalizes the $k$-BE problem, simply take the \emph{uniform}
length function $\lenUtil^{u}$ defined by setting
$\lenUtil^{u}(\ell) = 1$
for every
$2 \leq \ell \leq k$.
In this regard, a length function $\lenUtil$ is \emph{non-uniform} if
$\lenUtil(\ell) > \lenUtil(\ell')$
for some
$2 \leq \ell < \ell' \leq k$.\footnote{%
Note that scaling the images of $\lambda$ by a constant factor has no effect
on the approximation ratio.}
Unless stated otherwise, we subsequently address the
$(k, \lenUtil)$-BE
problem in its general form that includes both uniform and non-uniform length
functions $\lenUtil$.
Although the $k$-BE problem is a special case of this more general problem, it
is undoubtedly the most important special case and therefore, receives its own
focus.

\subsection{Our Contribution}
\label{section:contribution}
As observed by Roth et al.\ \cite{roth2005pairwise}, the
$2$-BE problem admits an efficient optimal truthful mechanism based on
constructing a maximum size matching (breaking ties consistently) in the
undirected graph
$G = ([n], E_{\leftrightarrow})$
obtained from the reported wish list vector
$\Wish \in \WishVecSpace_{n}$
by setting
$E_{\leftrightarrow}
=
\{ \{ i, j \} \mid i \in \Wish_{j} \text{ and } j \in \Wish_{i} \}$.\footnote{%
Notice that when
$k = 2$,
the length functions $\lenUtil$ have no role as all trading cycles are of the
same length.}
Consequently, our attention in the remainder of this paper is devoted to
length bounds
$k \geq 3$.

Arguably the simplest efficient
$(k, \lenUtil)$-BE
mechanism is the one we denote by \GreedyAlg{}.
This naive mechanism iterates over the possible trading cycles $c$ in
increasing order of length
$\length(c) = 2, 3, \dots, k$,
greedily adding trading cycle $c$ to $\Cycles_{\pi}$ if $c$ does not intersect
with any trading cycle already included in $\Cycles_{\pi}$.

\begin{theorem} \label{theorem:greedy-mechanism}
\GreedyAlg{} is truthful and approximates the
$(k, \lenUtil)$-BE
problem within ratio $k$ for every length bound
$k \geq 3$
and (uniform or non-uniform) length function $\lenUtil$.
\end{theorem}

Using the approximation bound of \Thm{}~\ref{theorem:greedy-mechanism} as the
starting point of our journey, this work is guided by the
following two research questions:
Does there exist an efficient truthful mechanism that approximates the
$(k, \lenUtil)$-BE
problem within ratio
$r < k$?
If so, how does the approximation ratio $r$ depend on the length function
$\lenUtil$?
To answer these questions, we first investigate the uniform length function
$\lenUtil = \lenUtil^{u}$.

\begin{theorem} \label{theorem:uniform-upper-bound}
The
$k$-BE
problem admits an efficient truthful mechanism with approximation ratio
$k - 1 + \epsilon$
for every length bound
$k \geq 3$
and constant
$\epsilon > 0$.\footnote{\label{footnote:agent-maximal-exchanges}%
The readers familiar with the work of Ashlagi and Roth \cite{ashlagi2014free}
may wonder whether our \Thm{}~\ref{theorem:uniform-upper-bound} is implied by
\cite[Theorem~1]{ashlagi2014free},
stating that any solution to the $k$-BE problem, which is maximal in terms of
the set of partaking agents, is guaranteed to be a
$(k - 1)$-approximation.
It is certainly not!
The exchanges constructed by the (truthful) algorithm that lies at the heart
of \Thm{}~\ref{theorem:uniform-upper-bound}'s proof are not necessarily
maximal (see \Sect{}~\ref{section:main-algorithm-bad-example} for an example
that demonstrates this fact), which forces us to use a considerably more
involved argument when bounding the approximation ratio (see
\Sect{}~\ref{section:uniform-approximation-ratio}) in comparison to the
$1$-paragraph proof of \cite[Theorem~1]{ashlagi2014free}'s upper bound.
The existence of efficient constructions of (agent-)maximal exchanges, let
alone truthful ones, remains an open question.%
}%
$^{\,}$%
\footnote{%
The runtime of the mechanism promised in this theorem, as well as that of the
one promised in \Thm{}~\ref{theorem:non-uniform-upper-bound}, is up-bounded
by
$n^{O ( k^{2} / \epsilon )}$
which is polynomial in $n$ (only) as long as $k$ and $\epsilon$ are
constants.
In the current paper, we do not focus on improving the dependency of the
runtime bound on $k$ and $\epsilon$.%
}
\end{theorem}

To the best of our knowledge, Theorem~\ref{theorem:uniform-upper-bound}
provides the first non-trivial upper bound on the approximation ratio of
efficient truthful $k$-BE mechanisms for any
$k \geq 3$.
Next, we turn our attention to non-uniform length functions
$\lenUtil : \{ 2, 3, \dots, k \} \rightarrow (0, 1]$
and define the parameter
\[\textstyle
\criticalRatio(\lenUtil)
\, = \,
\textstyle\max_{%
2 \leq \ell < \ell' \leq k \, : \, \lenUtil(\ell) > \lenUtil(\ell')%
}
\left\{
\textstyle\max \left\{
\tfrac{\ell' \cdot \lenUtil(\ell')}{\lenUtil(\ell)}, \,
\tfrac{\ell - 1}{\ell}
\cdot
\tfrac{\ell' \cdot \lenUtil(\ell')}{\lenUtil(\ell)}
+ 1
\right\}
\right\}
\, .
\]

\begin{theorem} \label{theorem:non-uniform-lower-bound}
The
$(k, \lenUtil)$-BE
problem does not admit a truthful mechanism with approximation ratio strictly
smaller than $\criticalRatio(\lenUtil)$ for every length bound
$k \geq 3$
and non-uniform length function $\lenUtil$.
\end{theorem}

Refer to Figure~\ref{figure:comb-instance} for an intuition regarding (a
slightly weaker version of) the $\criticalRatio(\lenUtil)$-lower bound of
\Thm{}~\ref{theorem:non-uniform-lower-bound};
an elaborate discussion of the properties of $\criticalRatio(\lenUtil)$ is
provided in Appendix~\ref{appendix:critical-ratio}.
We emphasize that this lower bound is universal in the sense that it applies
to \emph{any} non-uniform length function $\lenUtil$.
Moreover, it does not contradict \Thm{}~\ref{theorem:greedy-mechanism} as
$\criticalRatio(\lenUtil) < k$
for every non-uniform length function $\lenUtil$ (see
Appendix~\ref{appendix:critical-ratio}).
Surprisingly, $\criticalRatio(\lenUtil)$ plays a central role in
our approximation ratio upper bound as well.

\begin{theorem} \label{theorem:non-uniform-upper-bound}
The
$(k, \lenUtil)$-BE
problem admits an efficient truthful mechanism with approximation ratio
$\max \{ k - 1 + \epsilon, \criticalRatio(\lenUtil) \}$
for every length bound
$k \geq 3$,
non-uniform length function $\lenUtil$,
and constant
$\epsilon > 0$.
\end{theorem}

Refer to Appendix~\ref{appendix:critical-ratio} for a discussion of the
conditions under which the lower bound of
\Thm{}~\ref{theorem:non-uniform-lower-bound} matches the upper bound of
\Thm{}~\ref{theorem:non-uniform-upper-bound}.
The following theorem shows that the lower bound of
\Thm{}~\ref{theorem:non-uniform-lower-bound} is tight if we disregard
computational efficiency considerations.

\begin{theorem} \label{theorem:no-time-limit-upper-bound-non-uniform}
The
$(k, \lenUtil)$-BE
problem admits a computationally inefficient truthful mechanism with
approximation ratio $\criticalRatio(\lenUtil)$ for every length bound
$k \geq 3$
and non-uniform length function $\lenUtil$.
\end{theorem}

Back to uniform length functions, as proved by Hatfield
\cite{hatfield2005pairwise}, the $k$-BE problem admits a computationally
inefficient optimal truthful mechanism for every length bound
$k \geq 3$,
which means that \Thm{}~\ref{theorem:non-uniform-lower-bound} separates
between the
$k$-BE problem and the
$(k, \lenUtil)$-BE
problem with non-uniform $\lenUtil$.
This leaves us with the following interesting open question:
What is the best possible approximation ratio achievable by efficient truthful
mechanisms for the $k$-BE problem?
The existence of an optimal computationally inefficient truthful mechanism for
this problem implies that any attempt to devise a non-trivial lower bound on
the approximation ratio will have to rely on computational complexity hardness
assumptions.
Approximability lower bounds that employ such hardness assumptions are
discussed in \Sect{}~\ref{section:related-work} for the combinatorial
optimization aspect of the $k$-BE problem, irrespective of truthfulness
considerations.
Whether they can be improved by exploiting both the truthfulness requirement
and the efficiency requirement remains an intriguing, and notoriously
challenging, open question (we are unaware of any lower bound result that
exploits both requirements concurrently in the context of mechanism design
without money).

The mechanisms that lie at the heart of \Thm{}\
\ref{theorem:uniform-upper-bound} and \ref{theorem:non-uniform-upper-bound}
belong to a family that we refer to as \emph{standard local search
mechanisms} (see \Sect{}~\ref{section:local-search-algorithm} for a
definition).
It turns out that among the mechanisms in this family, the ones promised in
\Thm{}\ \ref{theorem:uniform-upper-bound} and
\ref{theorem:non-uniform-upper-bound} are optimal in terms of the dependency
of their approximation ratio on $k$.

\begin{theorem} \label{theorem:local-search-lower-bound}
The 
$(k, \lenUtil)$-BE
problem does not admit a truthful standard local search mechanism with
approximation ratio strictly smaller than
$k - 1$
for every length bound
$k \geq 3$
and (uniform or non-uniform) length function
$\lenUtil$.
\end{theorem}

\paragraph{Forming Small Coalitions.}
The
$(k, \lenUtil)$-BE
problem is related to a family of problems whose outcome is a partition of the
agent set $[n]$ into clusters, a.k.a.\ \emph{coalitions}, of size at most
$k$.
Each agent
$i \in [n]$
is associated with a wish list over the possible coalitions she is willing to
partake in;
her utility is $0$, if she is assigned to a singleton coalition, and it is
$\lenUtil(\ell)$ if she is assigned to a coalition of size
$2 \leq \ell \leq k$
from her wish list.
It turns out that the techniques used to establish our positive results apply
also to this coalition formation setup (see the maxCGIS problem formulated in
\Sect{}~\ref{section:preliminaries}).
Having said that, we subsequently restrict our attention to the BE problem
family as it seems more natural, with more convincing applications, and given
that our negative results are specific for this problem family.
Refer to \Sect{}~\ref{section:related-work} for a further discussion.

\subsection{Paper's Organization}
The remainder of this paper is organized as follows.
In \Sect{}~\ref{section:preliminaries}, we present additional notation and
terminology as well as some fundamental observations.
In \Sect{}~\ref{section:local-search-algorithm}, we introduce the class of
local search mechanisms and establish a sufficient condition for their
truthfulness under the uniform length function.
The (efficient) mechanisms promised in \Thm{}\
\ref{theorem:uniform-upper-bound} and \ref{theorem:non-uniform-upper-bound} are
developed in \Sect{}\ \ref{section:uniform-algorithm} and
\ref{section:non-uniform-algorithm}, respectively;
the latter mechanism relies on a technique for mechanism concatenation
developed in \Sect{}~\ref{section:concatenation-truthful-algorithms}.
In between, we analyze \GreedyAlg{} in \Sect{}~\ref{section:greedy} and prove
\Thm{}~\ref{theorem:greedy-mechanism} using, again, the mechanism
concatenation technique of
\Sect{}~\ref{section:concatenation-truthful-algorithms}.
While we are unaware of any general approximability lower bounds for the
uniform case (as discussed in \Sect{}~\ref{section:contribution}, this is an
important open question), in \Sect{}~\ref{section:main-algorithm-bad-example},
we show that the approximation ratio analysis conducted in
\Sect{}~\ref{section:uniform-algorithm} is (nearly) tight for the specific
mechanism of that section.
\Sect{}~\ref{section:no-running-time-limit} is then dedicated to the design
of the computationally inefficient mechanism promised in
\Thm{}~\ref{theorem:no-time-limit-upper-bound-non-uniform}, whereas the lower
bounds promised in \Thm{}\ \ref{theorem:non-uniform-lower-bound} and
\ref{theorem:local-search-lower-bound} are established in \Sect{}\
\ref{section:lower-bound-non-uniform} and
\ref{section:lower-bound-local-search}, respectively.
We conclude in \Sect{}~\ref{section:related-work} with an extended discussion
of the current work's place within the existing literature.

\section{Preliminaries}
\label{section:preliminaries}

\paragraph{Graph Theoretic Definitions.}
Throughout, the term graph is reserved for undirected graphs, whereas directed
graphs are referred to as digraphs.
To avoid confusion, the terms vertex and arc refer to the basic elements of a
digraph, whereas the terms node and edge refer to the basic elements of an
(undirected) graph.

Consider an (undirected) graph
$G = (V, E)$.
For a node
$v \in V$,
let
$\Neighbors_{G}(v) = \{ u \in V \mid \{ u, v \} \in E \}$
denote the set of $v$'s neighbors in $G$ and let
$\Degree_{G}(v) = |\Neighbors_{G}(v)|$
denote $v$'s degree in $G$.
The former notation is extended to node subsets
$U \subseteq V$,
denoting
$\Neighbors_{G}(U) = \bigcup_{v \in U} \Neighbors_{G}(v)$.

Given  a node subset
$U \subseteq V$,
the subgraph \emph{induced} by $U$ on $G$, denoted by $G[U]$, is the graph
whose node set is $U$ and whose edge set is
$\{ e \in E \mid e \subseteq U \}$.
A \emph{$d$-claw} is a complete bipartite graph with a single node on one side
and $d$ nodes on the other size.
A graph is said to be \emph{$d$-claw free} if it does not admit
a $d$-claw as a node induced subgraph.

A node subset
$I \subseteq V$
is an \emph{independent set (IS)} in $G$ if
$\{ u, v \} \notin E$
for any
$u, v \in I$.
An IS
$I \subseteq V$
is \emph{maximal} if $U$ is not an IS for any
$I \subset U \subseteq V$.
Alternatively, $I$ is a maximal IS if and only if it is both an IS and a
dominating set in $G$.

\paragraph{Cycle Graphs.}
Consider a wish list vector
$\Wish \in \WishVecSpace_{n}$,
$n \in \Integers_{> 0}$,
and recall that
$\Cycles_{\Wish}^{k} = \bigcup_{\pi \in \ExcSet_{\Wish}^{k}} \Cycles_{\pi}$
is the set of trading cycles over $[n]$ of length at most $k$ that respect
$\Wish$.
The
\emph{$k$-cycle graph} of $W$ is an (undirected) graph
$G = (V, E)$
defined over the trading cycle set
$V = \Cycles_{\Wish}^{k}$
so that nodes
$u, v \in V$
are adjacent under $E$ if and only if their corresponding trading cycles share
some (at least one) agents;
that is,
$\{ u, v \} \in E$
if and only if
$\Agents(u) \cap \Agents(v) \neq \emptyset$,
recalling that $\Agents(u)$ and $\Agents(v)$ are the sets of agents that form
the trading cycles $u$ and $v$, respectively (see
Figure~\ref{figure:example} for an illustration).

\begin{figure}
\centering
\begin{subfigure}{0.5\textwidth}%
\centering%
\includegraphics[scale=0.8]{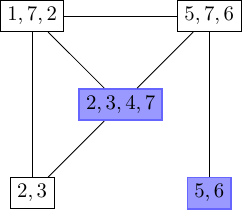}%
\caption{\label{figure:example:cycle-graph-marked}}%
\end{subfigure}%
\begin{subfigure}{0.5\textwidth}%
\centering%
\includegraphics[scale=0.8]{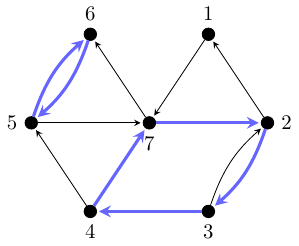}%
\caption{\label{figure:example:cycle-digraph-marked}}%
\end{subfigure}%
\caption{\label{figure:example}%
The $k$-cycle graph depicted in Figure~\ref{figure:example:cycle-graph-marked}
(where the agents in each node are listed according to their cyclic order)
corresponds to the wish list vector represented by the digraph depicted in
Figure~\ref{figure:example:cycle-digraph-marked} under the length bound
$k = 4$.
The nodes marked in Figure~\ref{figure:example:cycle-graph-marked} form an
IS that corresponds to the exchange (i.e., set of disjoint trading cycles)
marked in Figure~\ref{figure:example:cycle-digraph-marked}.
This IS turns out to be an optimal solution for the $k$-maxCGIS problem, hence
the corresponding exchange is an optimal solution for the $k$-BE problem.
}
\end{figure}

We extend the operator $\Agents$ from single nodes to node subsets
$U \subseteq V$,
defining
$\Agents(U) = \bigcup_{v \in U} \Agents(v)$.
In the converse direction, for an agent
$i \in [n]$,
let
$\Agents^{-1}(i) = \{ v \in V \mid i \in \Agents(v) \}$
be the set of nodes in which $i$ \emph{partakes};
notice that by definition, the nodes in $\Agents^{-1}(i)$ form a clique in
$G$ and every edge
$e \in E$
is internal to the clique corresponding to (at least) one of the
agents.
We define the \emph{length} of a node
$v \in V$
to be the length of the corresponding trading cycle, i.e.,
$|\Agents(v)|$.
Given a (uniform or non-uniform) length function
$\lenUtil : \{ 2, 3, \dots, k \} \rightarrow (0, 1]$,
we define the \emph{weight}
$\Weight(v) \in \Reals_{> 0}$
of $v$ to be
$\Weight(v) = |\Agents(v)| \cdot \lenUtil(|\Agents(v)|)$;
the weight operator is also extended to node subsets
$U \subseteq V$,
defining
$\Weight(U) = \sum_{v \in U} \Weight(v)$.

\begin{observation} \label{observation:cycle-graph}
Consider a wish list vector
$\Wish \in \WishVecSpace_{n}$
and let
$G = (V, E)$
be its corresponding $k$-cycle graph.
There is a one-to-one correspondence between the exchanges in
$\ExcSet_{\Wish}^{k}$ and the ISs in $G$.
Moreover, if
$I \subseteq V$
is the IS in $G$ that corresponds to an exchange
$\pi \in \ExcSet_{\Wish}^{k}$,
then
$\Weight(I) = \SW(\pi)$.
\end{observation}

Owing to \Obs{}~\ref{observation:cycle-graph}, we shall implement our
$(k, \lenUtil)$-BE
mechanisms by means of algorithms that construct an IS for a given $k$-cycle
graph with the objective of maximizing its weight.
It will be convenient to extend the input domain of these algorithms, allowing
them to be invoked on $k$-cycle graphs that do not necessarily correspond to
any particular wish list vector.
Formally, we say that a graph
$G = (V, E)$,
equipped with an operator
$\Agents : V \rightarrow 2^{[n]}$,
is a \emph{$k$-cycle graph} if
(1)
each node
$v \in V$
is identified with a unique cyclic order over the agents in $\Agents(v)$ of
length
$|\Agents(v)| \leq k$;\footnote{%
The cyclic orders are introduced only for the sake of distinguishing between
the nodes (and ultimately, bounding their total number).}
and
(2)
$E = \{ \{ u, v \} \subseteq V
\mid
\Agents(u) \cap \Agents(v) \neq \emptyset \}$.
Let $\CycleGraphSet_{n}^{k}$ denote the family of $k$-cycle graphs over $n$
agents.
Notice that in general, the cyclic orders identified with the nodes of a
$k$-cycle graph may not be compatible with the cycles in
$\Cycles_{\Wish}^{k}$ for any wish list vector
$\Wish \in \WishVecSpace_{n}$;
refer to Appendix~\ref{appendix:cycle-graph-with-no-wish-list-vector} for an
example of such a $k$-cycle graph.

\paragraph{The $(k, \lenUtil)$-maxCGIS Problem.}
We say that an algorithm \Alg{} admits \emph{approximation ratio}
$r \geq 1$
for the maximum weight IS problem in $k$-cycle graphs under the length
function $\lenUtil$, abbreviated as the
\emph{$(k, \lenUtil)$-maxCGIS}
problem, if for every
$n \in \Integers_{> 0}$
and for every $k$-cycle graph
$G = (V, E) \in \CycleGraphSet_{n}^{k}$,
it is guaranteed that $\Alg(G)$ is an IS in $G$ and that
$\Weight(\Alg(G)) \geq \frac{1}{r} \Weight(I)$
for every IS $I$ in $G$.
The \emph{$k$-maxCGIS} problem refers to the special case of
$(k, \lenUtil)$-maxCGIS
where
$\lenUtil = \lenUtil^{u}$
is the uniform length function.

To account for the game theoretic aspects of the
$(k, \lenUtil)$-BE
problem, we introduce the notions of agent utility and truthfulness also in
the context of the
$(k, \lenUtil)$-maxCGIS
problem.
Given a $k$-cycle graph
$G = (V, E) \in \CycleGraphSet_{n}^{k}$,
we define the \emph{utility} of an agent
$i \in [n]$
from an IS
$I \subseteq V$
in $G$, denoted by
$\utility(i, I)$,
to be
\[\textstyle
\utility(i, I)
\, = \,
\begin{cases}
0, & I \cap \Agents^{-1}(i) = \emptyset
\\
\lenUtil(|\Agents(v)|), & I \cap \Agents^{-1}(i) = \{ v \}
\end{cases} \, ;
\]
this is well defined as the intersection of $\Agents^{-1}(i)$ with any IS in
$G$ contains at most one node.
A
$(k, \lenUtil)$-maxCGIS
algorithm \Alg{} is said to be \emph{truthful} if for every
$n \in \Integers_{> 0}$,
$k$-cycle graph
$G = (V, E) \in \CycleGraphSet_{n}^{k}$,
agent
$i \in [n]$,
and node subset
$S \subseteq \Agents^{-1}(i)$,
it is guaranteed that
$\utility(i, \Alg(G[V - S]))
\leq
\utility(i, \Alg(G))$.

\begin{observation} \label{observation:algorithm-implies-mechanism}
An (efficient) truthful
$(k, \lenUtil)$-maxCGIS
algorithm implies an (efficient) truthful
$(k, \lenUtil)$-BE
mechanism with the same approximation ratio.
\end{observation}

It is important to point out that the truthfulness requirement from a
$(k, \lenUtil)$-BE
mechanism invoked on a wish list vector $\Wish$ is weaker than the
truthfulness requirement from a
$(k, \lenUtil)$-maxCGIS
algorithm invoked on the corresponding $k$-cycle graph
$G = (V, E)$
as the strategy space of the agents in the latter is more extensive
than that of the former.
Indeed, in the context of the
$(k, \lenUtil)$-maxCGIS
problem, an agent
$i \in [n]$
may choose to ``hide'' from the algorithm a node subset
$S \subseteq \Agents^{-1}(i)$
that does not correspond to omitting any agent subset from $i$'s wish list
$\Wish_{i}$.
We emphasize that while the stronger requirement from
$(k, \lenUtil)$-maxCGIS
algorithms makes our upper bounds ``harder'', we do not impose this stronger
requirement in our impossibility results that merely exploit the requirements
from a truthful mechanism as presented in \Sect{}~\ref{section:model}.

The following structural observation regarding the family of $k$-cycle graphs
plays a key role in the design of our approximation algorithms.

\begin{observation} \label{observation:cycle-graph-structure}
Consider a $k$-cycle graph
$G = (V, E) \in \CycleGraphSet_{n}^{k}$.
The number of nodes in $G$ satisfies
$|V| \leq \sum_{h = 2}^{k} \binom{n}{h} (h - 1)! \leq \operatorname{poly}(n)$.
Moreover, the number of nodes in any IS $I$ in $G$ satisfies
$|I| \leq n / 2$
and
$|I \cap \Neighbors_{G}(v)| \leq |\Agents(v)|$
for every node
$v \in V$.
In particular, the graph $G$ is
$(k + 1)$-claw
free.
\end{observation}

We note that the reduction from the (combinatorial aspect of the) BE problem
with length bound
$k \geq 3$
to the problem of constructing a maximum weight IS in
$(k + 1)$-claw
free graphs, particularly line graphs of hypergraphs, is used also by Biro et
al.\ \cite{biro2009maximum};
refer to \Sect{}~\ref{section:related-work} for further discussion.

\paragraph{INPA Algorithms.}
An algorithm \Alg{} for the
$(k, \lenUtil)$-maxCGIS
problem is said to be \emph{independent of non-partaking agents (INPA)} if
for every
$n \in \Integers_{> 0}$,
$k$-cycle graph
$G = (V, E) \in \CycleGraphSet_{n}^{k}$,
agent
$i \in [n] - \Agents(\Alg(G))$,
and node subset
$S \subseteq \Agents^{-1}(i)$,
it is guaranteed that
$\Alg(G[V - S]) = \Alg(G)$.
(Refer to \Sect{}~\ref{section:related-work} for a discussion of related game
theoretic notions.)
The importance of the INPA property of maxCGIS algorithms, specifically in the
context of the uniform length function, is demonstrated by the following
observation.

\begin{observation} \label{observation:uniform-inpa-implies-truthful}
For every length bound
$k \geq 3$,
any INPA algorithm for the $k$-maxCGIS problem is truthful.
\end{observation}
\begin{proof}%
Let \Alg{} be an INPA algorithm for the
$k$-maxCGIS
problem.
Fix some
$n \in \Integers_{> 0}$
and consider a $k$-cycle graph
$G = (V, E) \in \CycleGraphSet_{n}^{k}$
and an agent
$i \in [n]$.
If
$i \in \Agents(\Alg(G))$,
then
$\utility(i, \Alg(G)) = 1$
and agent $i$ cannot increase her utility by omitting any nodes from
$\Agents^{-1}(i)$.
So, assume that
$i \notin \Agents(\Alg(G))$
and consider a node subset
$S \subseteq \Agents^{-1}(i)$.
Since $\Alg$ is INPA, it follows that
$\Alg(G[V - S]) = \Alg(G)$,
hence
$i \notin \Agents(\Alg(G[V - S]))$.
Therefore,
$\utility(i, \Alg(G[V - S]))
=
0
=
\utility(i, \Alg(G))$
which means that agent $i$ does not improve her utility by omitting $S$ from
$\Agents^{-1}(i)$.
\end{proof}

\section{Local Search Algorithms}
\label{section:local-search-algorithm}
In this section, we introduce the class of local search
$(k, \lenUtil)$-maxCGIS
algorithms that is pivotal to our work.
We then establish a sufficient condition for the truthfulness of
algorithms in this class in the context of the $k$-maxCGIS problem (i.e., the
special case of
$\lenUtil = \lenUtil^{u}$).

\begin{definition*}[improvement rule, loyal, INPA, efficient]
An \emph{improvement rule} $r$ for the
$(k, \lenUtil)$-maxCGIS
problem is a function that maps a pair
$(G, I)$,
consisting of a $k$-cycle graph
$G = (V, E) \in \CycleGraphSet_{n}^{k}$,
$n \in \Integers_{> 0}$,
and an IS $I$ in $G$,
either to an IS $I'$ in $G$ satisfying
$\Weight(I') > \Weight(I)$
or to $\bot$.
The improvement rule $r$ is said to be \emph{loyal} if
$r(G, I) = I' \neq \bot$
implies that
$\Agents(I) \subseteq \Agents(I')$.
The improvement rule $r$ is said to be \emph{independent of non-partaking
agents (INPA)} if
(1)
$r(G, I) = I' \neq \bot$
implies that
$r(G[V - S], I) = I'$
for every agent
$i \in [n] - (\Agents(I) \cup \Agents(I'))$
and node subset
$S \subseteq \Agents^{-1}(i)$;
and
(2)
$r(G, I) = \bot$
implies that
$r(G[V - S], I) = \bot$
for every agent
$i \in [n] - \Agents(I)$
and node subset
$S \subseteq \Agents^{-1}(i)$.
Finally, the improvement rule $r$ is said to be \emph{efficient} if it admits
a computationally efficient implementation.
\end{definition*}

In words, the loyal property means that any agent that partakes in the
original IS $I$ must also partake in the new IS $I'$.
The INPA property means that agents that do not partake in $I$ nor in $I'$
have no impact on the improvement rule's outcome.
We can now define a local search algorithm over a list of improvement rules.

\begin{definition*}[local search algorithm]
A \emph{local search} algorithm \Alg{} for the
$(k, \lenUtil)$-maxCGIS
problem is characterized by an (ordered) list
$R = (r_{1}, \dots, r_{|R|})$
of improvement rules.
Given a $k$-cycle graph
$G = (V, E) \in \CycleGraphSet_{n}^{k}$,
$n \in \Integers_{> 0}$,
the (deterministic) algorithm returns the IS
$I = \Alg(G)$
constructed via the following iterative process:
\\
(1)
set
$t \gets 0$
and
$I^{t} \gets \emptyset$;
\\
(2)
if
$r_{j}(G, I^{t}) = \bot$
for all
$1 \leq j \leq |R|$,
then return the set
$I \gets I^{t}$
and halt;
\\
(3)
let $j^{t}$ be the smallest
$1 \leq j \leq |R|$
such that
$r_{j}(G, I^{t}) \neq \bot$;
\\
(4)
set
$I^{t + 1} \gets r_{j^{t}}(G, I^{t})$,
set
$t \gets t + 1$,
and go to step (2).
\\
Let $\Alg^{t}(G)$ denote the IS $I^{t}$ associated with iteration $t$ of this
process.
\end{definition*}

It turns out that improvement rules which are loyal and INPA lead to an INPA
(local search) algorithm.

\begin{lemma}
\label{lemma:local-search-algorithm-inpa-efficient}
Consider a local search algorithm \Alg{} for the
$(k, \lenUtil)$-maxCGIS
problem characterized by a list
$R = (r_{1}, \dots, r_{|R|})$
of improvement rules.
If all improvement rules in $R$ are loyal and INPA, then \Alg{} is INPA.
Moreover, if $|R|$ is fixed and all improvement rules in $R$ are efficient,
then \Alg{} is efficient.
\end{lemma}
\begin{proof}
Suppose that all improvement rules in $R$ are loyal and INPA and assume
towards contradiction that \Alg{} is not INPA which means that there exist
$n \in \Integers_{> 0}$,
a $k$-cycle graph
$G = (V, E) \in \CycleGraphSet_{n}^{k}$,
an agent
$i \in [n] - \Agents(\Alg(G))$,
and a node subset
$S \subseteq \Agents^{-1}(i)$
such that
$\Alg(G[V - S]) \neq \Alg(G)$.
Let $t$ be the smallest integer such that
$\Alg^{t}(G[V - S]) \neq \Alg^{t}(G)$
($t$ is well defined since
$\Alg(G[V - S]) \neq \Alg(G)$).
Notice that
$t > 0$
as
$\Alg^{0}(G[V - S]) = \emptyset = \Alg^{0}(G)$;
let $r_{j}$,
$1 \leq j \leq |R|$,
be the improvement rule whose application to
$(G, \Alg^{t - 1}(G))$
results in $\Alg^{t}(G)$ during the application of \Alg{} to $G$.

Recalling that all improvement rules in $R$ are loyal, we know that
$\Agents(\Alg^{s}(G)) \subseteq \Agents(\Alg^{s + 1}(G))$
for all $s$.
Since
$i \notin \Agents(\Alg(G))$,
it follows that
$i \notin \Agents(\Alg^{s}(G))$
for all $s$.
In particular, we deduce that
$i \notin \Agents(\Alg^{t - 1}(G)) = \Agents(\Alg^{t - 1}(G[V - S]))$.
By the definition of a local search algorithm, we know that
$r_{h}(G, \Alg^{t - 1}(G)) = \bot$
for every
$1 \leq h < j$.
Since the improvement rules $r_{h}$,
$1 \leq h < j$,
are INPA, it follows that
\[\textstyle
r_{h}(G[V - S], \Alg^{t - 1}(G[V - S]))
\, = \,
r_{h}(G[V - S], \Alg^{t - 1}(G))
\, = \,
\bot
\]
for every
$1 \leq h < j$.
As $r_{j}$ is also INPA and recalling that
$i \notin \Agents(\Alg^{t - 1}(G)) \cup \Agents(\Alg^{t}(G))$,
we conclude that
\[\textstyle
r_{j}(G[V - S], \Alg^{t - 1}(G[V - S]))
\, = \,
r_{j}(G[V - S], \Alg^{t - 1}(G))
\, = \,
r_{j}(G, \Alg^{t - 1}(G))
\, = \,
\Alg^{t}(G) \, .
\]
But this means that
$\Alg^{t}(G[V - S]) = \Alg^{t}(G)$,
deriving a contradiction.
Therefore, \Alg{} is INPA.

Now, suppose that $|R|$ is fixed and all improvement rules in $R$ are
efficient.
To deduce that \Alg{} is efficient, we argue that the application
of \Alg{} to any $k$-cycle graph
$G = (V, E) \in \CycleGraphSet_{n}^{k}$
includes $\operatorname{poly}(n)$ iterations.
To this end, recall that by the definition of improvement rules, we know that
$\Weight(\Alg^{t + 1}(G)) > \Weight(\Alg^{t}(G))$
for all $t$.
Since the weight of any IS $I$ in $G$ can be expressed as
$\Weight(I) = x \cdot c$,
where $x$ is a positive integer and $c$ is a constant in
$(0, 1]$,
it follows that
$\Weight(\Alg^{t + 1}(G)) - \Weight(\Alg^{t}(G)) \geq \Omega (1)$
for all $t$.
The argument is completed since
$\Weight(\Alg^{0}(G)) = 0$
and since \Obs{}~\ref{observation:cycle-graph-structure} implies that
$\Weight(\Alg(G)) \leq n k$.
\end{proof}

Note that if we omit the requirement that all improvement rules in $R$ are
loyal, then there may exist an agent that partakes in one of the intermediate
ISs maintained by \Alg{}, but does not partake in \Alg{}'s output.
By deleting some of the nodes this agent partakes in, one can steer the
execution of \Alg{}, leading to a different output.
Recalling that INPA implies truthfulness in the context of the $k$-maxCGIS
problem, we obtain the following corollary from
\Obs{}~\ref{observation:uniform-inpa-implies-truthful} combined with
\Lem{}~\ref{lemma:local-search-algorithm-inpa-efficient}.

\begin{corollary}
\label{corollary:local-search-uniform-algorithm-truthful-efficient}
Consider a local search algorithm \Alg{} for the
$k$-maxCGIS
problem characterized by a fixed length list $R$ of improvement rules.
If all improvement rules in $R$ are loyal, INPA, and efficient, then \Alg{} is
truthful and efficient.
\end{corollary}

\section{Upper Bound for the Uniform Length Function}
\label{section:uniform-algorithm}
We now turn to develop an efficient truthful algorithm for the $k$-maxCGIS
problem,
$k \geq 3$.
Our algorithm is presented as a local search algorithm and its truthfulness is
established based on the infrastructure developed in
\Sect{}~\ref{section:local-search-algorithm}.
Following that, in \Sect{}~\ref{section:uniform-approximation-ratio}, we
analyze the guaranteed approximation ratio of this algorithm and prove
\Thm{}~\ref{theorem:uniform-upper-bound} (by virtue of
\Obs{}~\ref{observation:algorithm-implies-mechanism}).
We start with presenting the improvement rules based on which our algorithm is
designed.

\paragraph{Expansion Improvement Rule.}
Given a $k$-cycle graph
$G = (V, E) \in \CycleGraphSet_{n}^{k}$
and an IS $I$ in $G$, the \emph{expansion} improvement rule $r_{E}$ returns
the lexicographically first IS
$I' = r_{E}(G, I)$
that satisfies
(1)
$I' \supset I$;
and
(2)
$|I'| = |I| + 1$.
If no such IS $I'$ exists, then $r_{E}$ returns
$r_{E}(G, I) = \bot$.
We refer to any IS that satisfies conditions (1) and (2) as a \emph{candidate}
of $r_{E}$.
Refer to Figures~\ref{figure:expansion1}--\ref{figure:expansion3} for an
illustration of $r_{E}$.

\begin{figure}
\centering
\begin{subfigure}{.33\textwidth}%
\centering%
\includegraphics[scale=0.35]{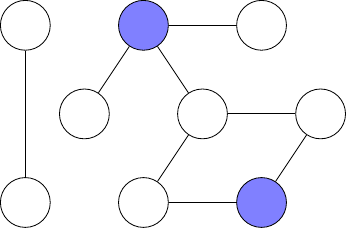}%
\caption{\label{figure:expansion1}%
Initial IS $I$}%
\end{subfigure}%
\begin{subfigure}{.33\textwidth}%
\centering%
\includegraphics[scale=0.35]{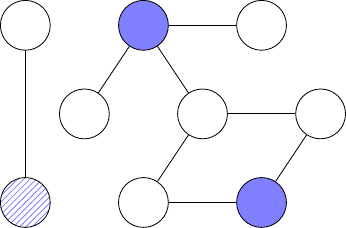}%
\caption{\label{figure:expansion2}%
Improvement attempt}%
\end{subfigure}%
\begin{subfigure}{.33\textwidth}%
\centering%
\includegraphics[scale=0.35]{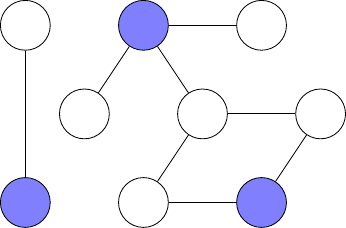}%
\caption{\label{figure:expansion3}%
New IS $I'$}%
\end{subfigure}%
\vskip 5mm
\begin{subfigure}{.33\textwidth}%
\centering%
\includegraphics[scale=0.35]{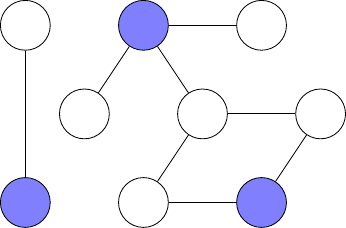}%
\caption{\label{figure:all-for-q1}%
Initial IS $I$}%
\end{subfigure}%
\begin{subfigure}{.33\textwidth}%
\centering%
\includegraphics[scale=0.35]{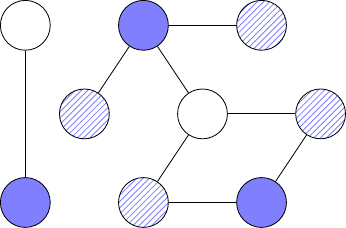}%
\caption{\label{figure:all-for-q2}%
Improvement attempt}%
\end{subfigure}%
\begin{subfigure}{.33\textwidth}%
\centering%
\includegraphics[scale=0.35]{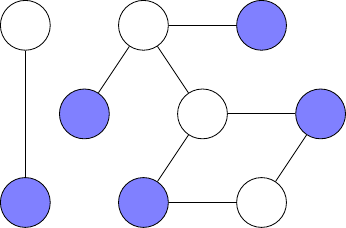}%
\caption{\label{figure:all-for-q3}%
New IS $I'$}%
\end{subfigure}%
\caption{\label{figure:improvement-rules}%
The expansion improvement rule $r_{E}$
(\ref{figure:expansion1}--\ref{figure:expansion3}) and the all-for-$q$
improvement rule $r_{q}$
(\ref{figure:all-for-q1}--\ref{figure:all-for-q3}).}%
\end{figure}

\begin{observation} \label{observation:uniform-expansion-rule}
The expansion improvement rule $r_{E}$ is loyal, INPA, and efficient.
\end{observation}
\begin{proof}
The fact that $r_{E}$ is loyal follows from the requirement that
$I' \supset I$.
To see that $r_{E}$ is INPA, observe that if
$i \notin \Agents(I)$,
then reducing $\Agents^{-1}(i)$ to
$\Agents^{-1}(i) - S$
for some node subset
$S \subseteq \Agents^{-1}(i)$
only decrease the collection of candidates from which
$r_{E}(G, I)$
is selected.
Therefore,
(1)
if
$r_{E}(G, I) = I' \neq \bot$
and
$i \notin \Agents(I')$,
then
$r_{E}(G[V - S], I) = I'$;
and
(2)
if
$r_{E}(G, I) = \bot$,
then
$r_{E}(G[V - S], I) = \bot$.
Finally, to see that $r_{E}$ is efficient, notice that it can be
implemented simply by inspecting all nodes in
$V - I$.
The assertion follows from the $\operatorname{poly}(n)$ size of $G$ as
promised in \Obs{}~\ref{observation:cycle-graph-structure}.
\end{proof}

\paragraph{All-For-$q$ Improvement Rule.}
Given a $k$-cycle graph
$G = (V, E) \in \CycleGraphSet_{n}^{k}$
and an IS $I$ in $G$, the \emph{all-for-$q$} improvement rule $r_{q}$
for a constant
$q \in \Integers_{> 0}$
returns the lexicographically first IS
$I' = r_{q}(G, I)$
that satisfies
(1)
$\Agents(I') \supset \Agents(I)$;
and
(2)
$I' = (I \cup X) - (\Neighbors_{G}(X) \cap I)$,
where
$X \subseteq \Neighbors_{G}(I)$
is an IS in $G$ and
$|\Neighbors_{G}(X) \cap I| \leq q$.
If no such IS $I'$ exists, then $r_{q}$ returns
$r_{q}(G, I) = \bot$.
We refer to any IS $I'$ that satisfies conditions (1) and (2) as a
\emph{candidate} of $r_{q}$.
Refer to Figures~\ref{figure:all-for-q1}--\ref{figure:all-for-q3} for an
illustration of $r_{q}$.

\begin{observation} \label{observation:uniform-all-for-q-rule}
The all-for-$q$ improvement rule is loyal, INPA, and efficient for any
constant
$q \in \Integers_{> 0}$.
\end{observation}
\begin{proof}
The fact that $r_{q}$ is loyal follows from requirement (1).
To see that $r_{q}$ is INPA, observe that if
$i \notin \Agents(I)$,
then reducing $\Agents^{-1}(i)$ to
$\Agents^{-1}(i) - S$
for some node subset
$S \subseteq \Agents^{-1}(i)$
can only decrease the collection of candidates from which
$r_{q}(G, I)$
is selected.
Therefore,
(1)
if
$r_{q}(G, I) = I' \neq \bot$
and
$i \notin \Agents(I')$,
then
$r_{q}(G[V - S], I) = I'$;
and
(2)
if
$r_{q}(G, I) = \bot$,
then
$r_{q}(G[V - S], I) = \bot$.
Finally, to see that $r_{q}$ is efficient, recall that
\Obs{}~\ref{observation:cycle-graph-structure} guarantees that the
$k$-cycle graph $G$ is
$(k + 1)$-claw
free.
Therefore, if
$X \subseteq \Neighbors_{G}(I)$
is an IS in $G$ and
$|\Neighbors_{G}(X) \cap I| \leq q$,
then
$|X| \leq q k$.
This means that one can implement the all-for-$q$ improvement rule by
inspecting all subsets of
$V - I$
of size at most
$q k$
which is a constant by the assumption that $k$ and $q$ are constants.
The assertion follows from the $\operatorname{poly}(n)$ size of $G$ as
promised in \Obs{}~\ref{observation:cycle-graph-structure}.
\end{proof}

\paragraph{The $\LocalSearchAlg_{q}$ Algorithm.}
Given a constant parameter
$q \in \Integers_{> 0}$,
our main algorithm for the $k$-maxCGIS problem, denoted by
$\LocalSearchAlg_{q}$, is defined to be the local search algorithm
characterized by the improvement rule list
$(r_{E}, r_{q})$.
By combining \Obs{}\ \ref{observation:uniform-expansion-rule} and
\ref{observation:uniform-all-for-q-rule} with
Corollary~\ref{corollary:local-search-uniform-algorithm-truthful-efficient},
we obtain the following corollary.

\begin{corollary} \label{corollary:uniform-main-algorithm-truthful}
$\LocalSearchAlg_{q}$ is truthful and efficient.
\end{corollary}

\subsection{Bounding the Approximation Ratio}
\label{section:uniform-approximation-ratio} 
Our goal in this section is to analyze the approximation ratio of
$\LocalSearchAlg_{q}$ and establish
\Thm{}~\ref{theorem:uniform-main-algorithm-approximation} that, combined with
\Obs{}~\ref{observation:algorithm-implies-mechanism} and
Corollary~\ref{corollary:uniform-main-algorithm-truthful}, yields
\Thm{}~\ref{theorem:uniform-upper-bound} by choosing
$q \geq 1 / \epsilon$.
In \Sect{}~\ref{section:main-algorithm-bad-example}, we show that the
upper bound established here on the approximation ratio of
$\LocalSearchAlg_{q}$ is nearly tight.

\begin{theorem} \label{theorem:uniform-main-algorithm-approximation}
$\LocalSearchAlg_{q}$ approximates the $k$-maxCGIS problem within ratio
$k - 1 + 1 / q$.
\end{theorem}

Consider a $k$-cycle graph
$G = (V, E) \in \CycleGraphSet_{n}^{k}$
input to $\LocalSearchAlg_{q}$.
Let
$I = \LocalSearchAlg_{q}(G)$
be the IS output by $\LocalSearchAlg_{q}$ and let $O$ be an optimal
(maximum weight) IS in $G$.
To establish \Thm{}~\ref{theorem:uniform-main-algorithm-approximation}, it
suffices to show that if
$H = (V_{H}, E_{H})$
is a connected component of
$G[I \cup O]$
with
$I_{H} = V_{H} \cap I$
and
$O_{H} = V_{H} \cap O$,
then
\begin{equation}\textstyle
\label{equation:connected-component-ratio}
\Weight(O_{H})
\, \leq \,
\left( k - 1 + \frac{1}{q} \right) \Weight(I_{H}) \, ;
\end{equation}
indeed, summing this inequality over all connected components $H$ of
$G[I \cup O]$
yields the desired bound
$\Weight(O)
\leq
\left( k - 1 + \frac{1}{q} \right) \Weight(I)$.

Consider a connected component
$H = (V_{H}, E_{H})$
of
$G[I \cup O]$
with
$I_{H} = V_{H} \cap I$
and
$O_{H} = V_{H} \cap O$.
If
$V_{H} = \{ v \}$
is a singleton set, then the inclusion of the expansion improvement rule
$r_{E}$ in $\LocalSearchAlg_{q}$ ensures that
$I_{H} = \{ v \}$,
hence inequality~(\ref{equation:connected-component-ratio}) clearly holds.
Assume hereafter that
$|V_{H}| > 1$
and notice that in this case, $H$ must be a bipartite graph with $I_{H}$ on
one side and $O_{H}$ on the other.

The construction of $H$ implies that
$(O \cup I_{H}) - O_{H}$
is an independent set in $G$, thus the optimality of $O$ ensures that
$\Weight(I_{H}) \leq \Weight(O_{H})$.
If
$\Weight(I_{H}) = \Weight(O_{H})$,
then inequality~(\ref{equation:connected-component-ratio}) holds trivially, so
assume hereafter that
$\Weight(I_{H}) < \Weight(O_{H})$.
By \Obs{}~\ref{observation:cycle-graph-structure}, we know that
$\Degree_{H}(v) \leq |\Agents(v)| = \Weight(v)$
for every
$v \in I_{H}$.
We shall establish inequality~(\ref{equation:connected-component-ratio}) by
proving \Lem{}\ \ref{lemma:uniform-deg=weight} and
\ref{lemma:uniform-deg<weight} for which we need the following combinatorial
proposition.

\begin{proposition} \label{proposition:combinatorial-property}
Let
$P = (A \uplus B, E)$
be a connected bipartite graph and let
$w : A \cup B \rightarrow \{ 2, 3, \dots, k \}$
be a weight function that satisfies
$w(a) \geq \Degree_{P}(a)$ 
for each node
$a \in A$.
Then,
$w(B) \leq w(A) \left( k - 1 + \frac{1}{|A|} \right)$.
Moreover, if there exists a node
$a \in A$
with
$w(a) > \Degree_{P}(a)$,
then
$w(B) \leq w(A) (k - 1)$.
\end{proposition}
\begin{proof}
Assume that the bipartite graph
$P = (A \uplus B , E)$ 
is cycle free;
this assumption is without loss of generality since the removal of an edge
from a cycle does not break the graph's connectivity nor does it violate the
condition
$\Degree_{P}(a) \leq w(a)$
for any node
$a \in A$.
This means that $P$ is a tree, hence
\[
|E|
\, = \,
|A| + |B| - 1 \, .
\]
On the other hand, as
$\Degree_{P}(a) \leq w(a)$
for all nodes
$a \in A$, we deduce that
\[
|E|
\, = \,
\sum_{a \in A} \Degree_{P}(a)
\, \leq \,
\sum_{a \in A} w(a)
\, = \,
w(A) \,.
\]
Recalling that the images of $w$ are up-bounded by $k$, we know that
$w(A) \leq k \cdot |A|$
and
$w(B) \leq k \cdot |B|$.
Put together, we conclude that
\[
\frac{w(A)}{k} \left( 1 - \frac{1}{|A|} \right) + \frac{w(B)}{k}
\, \leq \,
|A| \left( 1 - \frac{1}{|A|} \right) + |B|
\, = \,
|E|
\, \leq \,
w(A) \, ,
\]
thus
\[
w(B)
\, \leq \,
w(A) \left( k - 1 +  \frac{1}{|A|}\right) \, .
\]

Now, assume that there exists a node
$a \in A$
such that
$\Degree_{P}(a) < w(a)$.
Since $\Degree_{P}(a)$ and $w(a)$ are integers, it follows that
$\Degree_{P}(a) \leq w(a) - 1$,
hence
\[
|E|
\, = \,
\sum_{a' \in A} \Degree_{P}(a')
\, \leq \,
\left( \sum_{a' \in A} w(a') \right) - 1
\, = \,
w(A) - 1 \, .
\]
Therefore,
\[
\frac{w(A)}{k} + \frac{w(B)}{k} - 1 \leq
|A| + |B| - 1 = 
|E| \leq
w(A) - 1 \, ,
\]
which implies that
\[
w(B) \leq (k - 1) \cdot w(A) \, ,
\]
thus completing the proof.
\end{proof}

\begin{lemma} \label{lemma:uniform-deg=weight}
If
$\Degree_{H}(v) = \Weight(v)$
for every
$v \in I_{H}$, 
then
$\Weight(O_{H}) < \left( k - 1 + \frac{1}{q} \right) \Weight(I_{H})$.
\end{lemma}
\begin{proof}
If
$|I_{H}| > q$,
then by plugging
$A = I_{H}$
and
$B = O_{H}$
into \Prop{}~\ref{proposition:combinatorial-property},
we conclude that
$\Weight(O_{H})
\leq
\left( k - 1 + \frac{1}{|I_{H}|} \right) \Weight(I_{H})
<
\left( k - 1 + \frac{1}{q} \right) \Weight(I_{H})$,
thus proving the assertion.
So, assume towards contradiction that
$|I_{H}| \leq q$,
and recall that
$I = \LocalSearchAlg_{q}(G)$
is the IS returned by $\LocalSearchAlg_{q}$ on $G$.
By the construction of $H$, we know that
$I' = (I \cup O_{H}) - I_{H}$
is an IS in $G$.
We shall reach a contradiction by showing that $I'$ is a candidate of
$r_{q}$ when applied to $G$ and $I$, which means that the application
of $\LocalSearchAlg_{q}$ to $G$ should not have returned $I$.

To this end, consider a node
$v \in I_{H}$.
The assumption that
$\Degree_{H}(v) = \Weight(v)$
implies that
$|\Neighbors_{H}(v)| = |\Agents(v)|$.
As $\Neighbors_{H}(v)$ is an IS
(recall that
$\Neighbors_{H}(v) \subseteq O_{H}$),
we know that
$\Agents(u) \cap \Agents(u') = \emptyset$
for every
$u, u' \in \Neighbors_{H}(v)$.
This means that for every agent
$i \in \Agents(v)$,
there exists a unique node
$u = u(i) \in \Neighbors_{H}(v)$
such that
$i \in \Agents(u)$.
In particular, we conclude that
$\Agents(v) \subseteq \Agents(\Neighbors_{H}(v)) \subseteq \Agents(O_{H})$, 
hence
$\Agents(I_{H}) \subseteq \Agents(O_{H})$.
Therefore,
$\Agents(I) \subseteq \Agents(I')$
and the assumption that
$\Weight(I_{H}) < \Weight(O_{H})$
guarantees that
$\Agents(I) \subset \Agents(I')$.
Since
$|I_{H}| \leq q$
and since
$O_{H} \subseteq \Neighbors_{G}(I)$,
it follows that $I'$ is indeed a candidate of $r_{q}$, realized by taking
$X = O_{H}$.
\end{proof}

\begin{lemma} \label{lemma:uniform-deg<weight}
If there exists a node
$v \in I_{H}$
such that
$\Degree_{H}(v) < \Weight(v)$,
then
$\Weight(O_{H}) \leq (k - 1) \Weight(I_{H})$.
\end{lemma}
\begin{proof} 
Follows directly from \Prop{}~\ref{proposition:combinatorial-property} by
plugging
$A = I_{H}$
and
$B = O_{H}$.
\end{proof}

\section{Concatenation of Truthful Algorithms} 
\label{section:concatenation-truthful-algorithms}
In this section, we introduce the notion of algorithm concatenation that
facilitates the construction of our main algorithm for non-uniform length
functions, presented in \Sect{}~\ref{section:non-uniform-algorithm}.
We also introduce a sufficient condition for the truthfulness of concatenated
algorithms;
the usefulness of this condition is demonstrated already in
\Sect{}~\ref{section:greedy}, where it allows us to establish the truthfulness
of \GreedyAlg{}.
All algorithms addressed in this section are
$(k, \lenUtil)$-maxCGIS
algorithms for a length bound
$k \geq 2$
and a (uniform or non-uniform) length function $\lenUtil$.

\begin{definition*}[algorithm concatenation]
The \emph{concatenation} of algorithms $\Alg_{1}$ and $\Alg_{2}$, denoted by
$\Alg_{1} \cdot \Alg_{2}$,
is the algorithm that given an input $k$-cycle graph
$G = (V, E) \in \CycleGraphSet_{n}^{k}$,
outputs the IS
\[\textstyle
\left( \Alg_{1} \cdot \Alg_{2} \right) (G)
\, = \,
\Alg_{1}(G)
\cup
\Alg_{2} \left(
G \left[
V - \left( \Alg_{1}(G) \cup \Neighbors_{G} \left( \Alg_{1}(G) \right) \right)
\right]
\right) \, .
\]
\end{definition*}

Intuitively, the concatenated algorithm
$\Alg_{1} \cdot \Alg_{2}$
works as follows:
first, run $\Alg_{1}$ on $G$ to get the IS $\Alg_{1}(G)$;
then, expand $\Alg_{1}(G)$ with the IS constructed by running $\Alg_{2}$ on
the graph obtained from $G$ by removing the nodes in $\Alg_{1}(G)$ and their
neighbors.
The following definition provides a sufficient condition for the agents to
favor partaking in the output of $\Alg_{1}$ over that of $\Alg_{2}$.

\begin{definition*}[the $\succeq$ relation]
Algorithms $\Alg_{1}$ and $\Alg_{2}$ satisfy the relation
$\Alg_{1} \succeq \Alg_{2}$
if the maximum length of a node included in any IS output by $\Alg_{1}$ is
up-bounded by the minimum length of a node included in any IS output by
$\Alg_{2}$, namely,
\[
\max_{n \in \Integers_{> 0}} \,
\max_{G = (V, E) \in \CycleGraphSet_{n}^{k}} \,
\max_{v \in \Alg_{1}(G)} \,
|\Agents(v)|
\, \leq \,
\min_{n \in \Integers_{> 0}} \,
\min_{G = (V, E) \in \CycleGraphSet_{n}^{k}} \,
\min_{v \in \Alg_{2}(G)} \,
|\Agents(v)| \, .
\]
\end{definition*}

Based on the definitions of algorithm concatenation and the $\succeq$
relation, we can now state the following important lemma.

\begin{lemma} \label{lemma:concatenation-truthful}
Let $\Alg_{1}$ be a truthful INPA algorithm and let $\Alg_{2}$ be a truthful
algorithm such that
$\Alg_{1} \succeq \Alg_{2}$.
Then, the concatenated algorithm
$\Alg_{1} \cdot \Alg_{2}$
is truthful.
\end{lemma}
\begin{proof}
Consider a $k$-cycle graph
$G = (V, E) \in \CycleGraphSet_{n}^{k}$
and an agent
$i \in [n]$
and fix a node subset
$S \subseteq \Agents^{-1}(i)$.
We establish the assertion by proving that
\[\textstyle
\utility(i, (\Alg_{1} \cdot \Alg_{2})(G[V - S]))
\, \leq \,
\utility(i, (\Alg_{1} \cdot \Alg_{2})(G)) \, .
\]
To avoid cumbersome expressions, we denote
$G' = G[V - S]$
and
\[\textstyle
I_{1} = \Alg_{1}(G) \, , \quad
I_{2} = \Alg_{2}(G) \, , \quad
I'_{1} = \Alg_{1}(G') \, , \quad \text{and} \quad
I'_{2} = \Alg_{2}(G') \, .
\]
The proof is divided into two cases depending on whether agent $i$ partakes in
$I_{1}$.

Suppose first that
$i \in \Agents(I_{1})$
which means, by the definition of algorithm concatenation, that
$\utility(i, (\Alg_{1} \cdot \Alg_{2})(G))
=
\utility(i, I_{1})$.
If
$i \in \Agents(I'_{1})$,
then $i$'s utility from
$(\Alg_{1} \cdot \Alg_{2})(G')$
is equal to her utility from
$I'_{1}$.
Since $\Alg_{1}$ is truthful, it follows that 
$\utility(i, I'_{1}) \leq \utility(i, I_{1})$,
hence
\[\textstyle
\utility(i, (\Alg_{1} \cdot \Alg_{2})(G'))
\, = \,
\utility(i, I'_{1})
\, \leq \,
\utility(i, I_{1})
\, = \,
\utility(i, (\Alg_{1} \cdot \Alg_{2})(G)) \, .
\]
If
$i \notin \Agents(I'_{1})$,
then $i$'s utility from
$(\Alg_{1} \cdot \Alg_{2})(G')$
is equal to her utility from
$\Alg_{2}(G'[(V - S) - (I'_{1} \cup \Neighbors_{G'}(I'_{1}))])$.
Recalling that  
$\Alg_{1} \succeq \Alg_{2}$,
we conclude that the utility of $i$ from $I_{1}$ is at least as high as her
utility from any IS output by $\Alg_{2}$, thus
\begin{align*}
\utility(i, (\Alg_{1} \cdot \Alg_{2})(G'))
\, = \, &
\utility(i, \Alg_{2}(G'[(V - S) - (I'_{1} \cup \Neighbors_{G'}(I'_{1}))]))
\\
\leq \, &
\utility(i, I_{1})
\, = \,
\utility(i, (\Alg_{1} \cdot \Alg_{2})(G)) \, .
\end{align*}
Therefore, the assertion holds when
$i \in \Agents(I_{1})$.

So, suppose that
$i \notin \Agents(I_{1})$
which means, by the definition of algorithm concatenation, that 
$\utility(i, (\Alg_{1} \cdot \Alg_{2})(G))
=
\utility(i, \Alg_{2}(G[V - (I_{1} \cup \Neighbors_{G}(I_{1}))]))$.
By the truthfulness of $\Alg_{1}$, we deduce that
$i \notin \Agents(I'_{1})$,
thus
\[\textstyle
\utility(i, (\Alg_{1} \cdot \Alg_{2})(G'))
\, = \,
\utility(i, \Alg_{2}(G'[(V - S) - (I'_{1} \cup \Neighbors_{G'}(I'_{1}))])) \, .
\]
Since $\Alg_{1}$ is INPA, it follows that 
$I'_{1} = I_{1}$,
therefore
\[\textstyle
G'[(V - S) - (I'_{1} \cup \Neighbors_{G'}(I'_{1}))]
\, = \,
G'[(V - S) - (I_{1} \cup \Neighbors_{G'}(I_{1}))]
\, = \,
G[(V - (I_{1} \cup \Neighbors_{G}(I_{1}))) - S] \, .
\]
The truthfulness of $\Alg_{2}$ ensures that
\[\textstyle
\utility(i, \Alg_{2}(G[(V - (I_{1} \cup \Neighbors_{G}(I_{1}))) - S]))
\, \leq \,
\utility(i, \Alg_{2}(G[V - (I_{1} \cup \Neighbors_{G}(I_{1}))])) \, ,
\]
hence
\[\textstyle
\utility(i, (\Alg_{1} \cdot \Alg_{2})(G'))
\, \leq \,
\utility(i, \Alg_{2}(G[V - (I_{1} \cup \Neighbors_{G}(I_{1}))]))
\, = \,
\utility(i, (\Alg_{1} \cdot \Alg_{2})(G)) \, ,
\]
where the second transition holds by the definition of algorithm
concatenation.
This means that the assertion holds also when
$i \notin \Agents(I_{1})$
which completes the proof.
\end{proof}

A repeated application of \Lem{}~\ref{lemma:concatenation-truthful} leads to
the following corollary.

\begin{corollary} \label{corollary:concatenation-truthful-algorithms}
Let
$\Alg_{1}, \dots, \Alg_{r - 1}$
be truthful INPA algorithms, let $\Alg_{r}$ be a truthful algorithm, and
assume that
$\Alg_{j} \succeq \Alg_{j + 1}$
for every
$1 \leq j \leq r - 1$.
Then, the concatenated algorithm
$\Alg_{1} \cdots \Alg_{r}$
is truthful.
\end{corollary}

\section{The Greedy Approach}
\label{section:greedy}
Before moving to the more involved 
$(k, \lenUtil)$-maxCGIS
algorithm of \Sect{}~\ref{section:non-uniform-algorithm}, let us analyze what
is perhaps the simplest approach for designing efficient truthful
$(k, \lenUtil)$-maxCGIS
algorithms.
Given a $k$-cycle graph
$G = (V, E)$,
the greedy approach iterates over all nodes in $V$ in a predefined order and
adds a node
$v \in V$
to the constructed IS $I$, setting
$I \gets I \cup \{ v \}$,
if $v$ is not adjacent to any node
$v' \in I$
(i.e., $I$ remains an IS in $G$).
It is not difficult to show that a greedy algorithm that prioritizes the nodes
$v \in V$
according to their weight
$\Weight(v) = |\Agents(v)| \cdot \lenUtil(|\Agents(v)|)$
admits an approximation ratio $k$.
However, this algorithm is not necessarily truthful as the objective of the
individual agent
$i \in [n]$
is to minimize the length $|\Agents(v)|$ of (the unique) node
$v \in \Agents^{-1}(i) \cap I$
(recall that the length function $\lenUtil$ is monotonically non-increasing)
which is not necessarily aligned with maximizing
$|\Agents(v)| \cdot \lenUtil(|\Agents(v)|)$.

To overcome this difficulty, we propose a greedy algorithm, denoted by
\GreedyAlg{}, that prioritizes the nodes
$v \in V$
according to their length $|\Agents(v)|$, thus adding the nodes $v$ in
non-decreasing order of their contribution to the utility of the individual
agents included in $\Agents(v)$.
Formally, for
$j = 2, \dots, k$,
let $r_{E}^{j}$ be the improvement rule derived from the expansion
improvement rule $r_{E}$, introduced in \Sect{}~\ref{section:uniform-algorithm},
by augmenting it with the requirement that the nodes $v$ included in any
candidate IS $I$ are of length
$|\Agents(v)| = j$;
let $\GreedyAlg^{j}$ be the local search algorithm characterized by the lone
improvement rule $r_{E}^{j}$.
The greedy algorithm is taken to be the concatenation
$\GreedyAlg = \GreedyAlg^{2} \cdots \GreedyAlg^{k}$.

\begin{observation} \label{observation:greedy-j}
Algorithm $\GreedyAlg^{j}$ is truthful and INPA for each
$2 \leq j \leq k$.
Moreover,
$\GreedyAlg^{j} \succeq \GreedyAlg^{j + 1}$
for each
$2 \leq j \leq k - 1$.
\end{observation}
\begin{proof}
The fact that $\GreedyAlg^{j}$ is truthful and INPA follows from the same line
of arguments as the one used in \Sect~\ref{section:uniform-algorithm}
regarding the expansion improvement rule $r_{E}$.
The fact that
$\GreedyAlg^{j} \succeq \GreedyAlg^{j + 1}$
holds as $\GreedyAlg^{j}$ outputs ISs with nodes of length $j$ and 
$\GreedyAlg^{j + 1}$ 
outputs ISs with nodes of length
$j + 1$.
\end{proof}

The truthfulness of \GreedyAlg{} follows immediately from
Corollary~\ref{corollary:concatenation-truthful-algorithms}.

\begin{corollary} \label{corollary:greedy-truthful}
\GreedyAlg{} is truthful.
\end{corollary}

In the remainder of this section, we show that \GreedyAlg{} approximates the
$(k, \lenUtil)$-maxCGIS
problem within ratio $k$, thus establishing
\Thm{}~\ref{theorem:greedy-mechanism}.
This relies on the following observation that serves us in
\Sect{}~\ref{section:non-uniform-algorithm} as well.

\begin{observation} \label{observation:greedy-mapping}
Consider a $k$-cycle graph
$G = (V, E) \in \CycleGraphSet_{n}^{k}$
and let
$I = \GreedyAlg(G)$.
For every IS $J$ in $G$, there exists a function
$\mu : J \rightarrow I$
such that
(1)
$|\Agents(u)| \geq |\Agents(\mu(u))|$
for every node
$u \in J$;
and
(2)
$|\mu^{-1}(v)| \leq |\Agents(v)|$
for every node
$v \in I$,
where
$\mu^{-1}(v) = \{ u \in J : \mu(u) = v \}$.
\end{observation}
\begin{proof}
We construct the function $\mu$ by setting $\mu(u)$ as follows for each node
$u \in J$:
If
$u \in I$,
then set 
$\mu(u) = u$.
Otherwise
($u \notin I$),
the design of \GreedyAlg{} ensures that there exists a node
$v \in \Neighbors_{G}(u) \cap I$
whose length low-bounds that of $u$, i.e.,
$|\Agents(v)| \leq |\Agents(u)|$,
so we set
$\mu(u) = v$.
Condition (1) follows directly by the construction of $\mu$.
To see that condition (2) also holds, we employ
\Obs{}~\ref{observation:cycle-graph-structure} to deduce that since $J$
is an IS, it follows that
$|J \cap \Neighbors_{G}(v)| \leq |\Agents(v)|$
for every node
$v \in V$
and in particular, for every node $v$ in the image of $\mu$.
Condition (2) follows as $\mu$ maps each node
$u \in J$
either to $u$ itself or to a node in $\Neighbors_{G}(u)$.
\end{proof}

To bound the approximation ratio of \GreedyAlg{}, consider a $k$-cycle graph
$G = (V, E) \in \CycleGraphSet_{n}^{k}$
and let
$I = \GreedyAlg(G)$
be the IS output by \GreedyAlg{}.
Let
$O \subseteq V$
be an optimal (maximum weight) IS in $G$ and let
$\mu : O \rightarrow I$
be the function promised in \Obs{}~\ref{observation:greedy-mapping} when
plugging
$J = O$.
Let
$v_{1}, \dots, v_{r} \in I$
be the images of $\mu$ and partition $O$ into the (pairwise disjoint) clusters
$O_{1}, \dots, O_{r}$
defined so that
$O_{h} = \mu^{-1}(v_{h})$
for each
$h \in [r]$.
We argue that
\begin{equation}\textstyle
\label{equation:greedy-proof-cluster-contribution}
\Weight(O_{h})
\, \leq \,
k \cdot \Weight(v_{h})
\end{equation}
for each
$h \in [r]$.
This implies that
\[\textstyle
\Weight(I)
\, \geq \,
\sum_{h = 1}^{r} \Weight(v_{h})
\, \geq \,
\frac{1}{k} \sum_{h = 1}^{r} \Weight(O_{h})
\, = \,
\frac{1}{k} \Weight(O) \, ,
\]
thus establishing the assertion.

To show that \eqref{equation:greedy-proof-cluster-contribution} holds,
consider some
$h \in [r]$
and recall that the construction of $\mu$ guarantees that
$|\Agents(v_{h})| \leq |\Agents(u)|$
for each
$u \in O_{j}$.
Since $\lenUtil$ is non-increasing, it follows that
\[\textstyle
\lenUtil(|\Agents(u)|)
\, \leq \,
\lenUtil(|\Agents(v_{h})|) \, .
\]
Moreover, \Obs{}~\ref{observation:greedy-mapping} ensures that
\[\textstyle
|O_{h}|
\, \leq \,
|\Agents(v_{h})| \, .
\]
Recalling that the length of every node in $V$ is at most $k$, we deduce that
\begin{align*}
\Weight(O_{h})
\, = \, &
\sum_{u \in O_{h}} \Weight(u)
\, = \,
\sum_{u \in O_{h}} |\Agents(u)| \cdot \lenUtil(|\Agents(u)|)
\\
\leq \, &
\sum_{u \in O_{h}} k \cdot \lenUtil(|\Agents(u)|)
\, \leq \,
|\Agents(v_{h})| \cdot k \cdot \lenUtil(|\Agents(v_{h})|)
\, = \,
k \cdot \Weight(v_{h}) \, ,
\end{align*}
thus establishing (\ref{equation:greedy-proof-cluster-contribution}).

\section{Upper Bound for Non-Uniform Length Functions} 
\label{section:non-uniform-algorithm}
Consider the
$(k, \lenUtil)$-maxCGIS
problem for a length bound
$k \geq 3$
and a non-uniform length function $\lenUtil$.
In the current section, we develop an efficient truthful algorithm for this
problem and establish \Thm{}~\ref{theorem:non-uniform-upper-bound} (by virtue
of \Obs{}~\ref{observation:algorithm-implies-mechanism}).
Intuitively, this algorithm operates by first running \GreedyAlg{} until all
remaining nodes are of of ``sufficiently large length'' and then, running
$\LocalSearchAlg_{q}$ on the remaining subgraph, treating it as in the
uniform case.

\paragraph{The $\NonUniformAlg_{q}$ Algorithm.}
Let $\ell^{*}$ be the largest
$2 \leq \ell \leq k - 1$
that satisfies
$\lenUtil(\ell) > \lenUtil(k)$
(notice that $\ell^{*}$ is well defined as $\lenUtil$ is non-uniform).
Let
$\GreedyAlg^{\leq \ell^{*}}$
be the algorithm defined as the concatenation
$\GreedyAlg^{\leq \ell^{*}} = \GreedyAlg^{2} \cdots \GreedyAlg^{\ell^{*}}$.
Analogously to the definition of the improvement rule $r_{E}^{k}$ (see
\Sect{}~\ref{section:greedy}), let
$r_{E}^{> \ell^{*}}$
(resp.,
$r_{q}^{> \ell^{*}}$)
be the improvement rule derived from the expansion improvement rule $r_{E}$
(resp., all-for-$q$ improvement rule $r_{q}$), introduced in
\Sect{}~\ref{section:uniform-algorithm}, by augmenting it with the requirement
that the nodes $v$ included in any candidate IS $I$ are of length
$|\Agents(v)| > \ell^{*}$.
Let 
$\LocalSearchAlg_{q}^{> \ell^{*}}$ 
be the local search algorithm characterized
by the improvement rule list
$(r_{E}^{> \ell^{*}}, r_{q}^{> \ell^{*}})$.
Given a constant parameter
$q \in \Integers_{> 0}$,
our algorithm for the non-uniform case, denoted by $\NonUniformAlg_{q}$, is
defined to be the concatenation of
$\GreedyAlg^{\leq \ell^{*}}$
and $\LocalSearchAlg_{q}^{> \ell^{*}}$, namely,
$\NonUniformAlg_{q}
=
\GreedyAlg^{\leq \ell^{*}} \cdot \LocalSearchAlg_{q}^{> \ell^{*}}$.

\begin{lemma} \label{lemma:non-uniform-algorithm-truthful}
$\NonUniformAlg_{q}$ is efficient and truthful.
\end{lemma}
\begin{proof}
Notice first that the proofs of \Obs{}\
\ref{observation:uniform-expansion-rule} and
\ref{observation:uniform-all-for-q-rule} can be applied, verbatim, to
$r_{E}^{> \ell^{*}}$
and
$r_{q}^{> \ell^{*}}$,
respectively, concluding that these improvement rules are loyal, INPA, and
efficient (just like their original counterparts).
By \Lem{}~\ref{lemma:local-search-algorithm-inpa-efficient}, we deduce
that
$\LocalSearchAlg_{q}^{> \ell^{*}}$
is INPA and efficient.
The efficiency of $\NonUniformAlg_{q}$ follows as
$\GreedyAlg^{\leq \ell^{*}}$
is clearly efficient as well.

Since algorithm
$\LocalSearchAlg_{q}^{> \ell^{*}}$ 
is designed so that all nodes $v$
included in its output ISs are of length
$|\Agents(v)| > \ell^{*}$,
it follows that
$\LocalSearchAlg_{q}^{> \ell^{*}}$
is oblivious to the length function $\lenUtil$.
In particular, we can employ
\Obs{}~\ref{observation:uniform-inpa-implies-truthful}, whose scope is
restricted (originally) to  the uniform length function, to deduce that the
fact that
$\LocalSearchAlg_{q}^{> \ell^{*}}$ 
is INPA implies its truthfulness also
under non-uniform length functions.

Notice that
$\GreedyAlg^{\leq \ell^{*}} \succeq \LocalSearchAlg_{q}^{> \ell^{*}}$
as 
$\GreedyAlg^{\leq \ell^{*}}$ 
outputs ISs with nodes of length at most
$\ell^{*}$
and $\LocalSearchAlg_{q}^{> \ell^{*}}$ output ISs with nodes of length (strictly) greater that $\ell^{*}$.
Hence, the truthfulness of $\NonUniformAlg_{q}$ follows by
\Obs{}~\ref{observation:greedy-j} and
Corollary~\ref{corollary:concatenation-truthful-algorithms}.
\end{proof}

\subsection{Bounding the Approximation Ratio} 
Our goal in this section is to establish the following theorem that, combined
with \Obs{}~\ref{observation:algorithm-implies-mechanism} and
\Lem{}~\ref{lemma:non-uniform-algorithm-truthful}, yields
\Thm{}~\ref{theorem:non-uniform-upper-bound} by plugging
$q = \lceil 1 / \epsilon \rceil$.

\begin{theorem} \label{theorem:non-uniform-main-algorithm-approximation}
$\NonUniformAlg_{q}$ approximates the
$(k, \lenUtil)$-maxCGIS
problem within ratio
$\max \{ k - 1 + 1 / q, \criticalRatio(\lenUtil) \}$.
\end{theorem}

Consider a $k$-cycle graph
$G = (V, E) \in \CycleGraphSet_{n}^{k}$
input to $\NonUniformAlg_{q}$.
Let
$I = \NonUniformAlg_{q}(G)$
be the IS output by $\NonUniformAlg_{q}$ and let $O$ be an optimal
(maximum weight) IS in $G$.
To avoid cumbersome expressions, we define
\[\textstyle
I_{gr}
\, = \,
\GreedyAlg^{\leq \ell^{*}}(G) , \quad
V_{gr}
\, = \,
I_{gr} \cup \Neighbors_{G}(I_{gr}) , \quad
V_{ls}
\, = \,
V - V_{gr} , \quad \text{and} \quad
I_{ls}
\, = \,
\LocalSearchAlg_{q}^{> \ell^{*}}(G[V_{ls}]) \, ,
\]
observing that
$I = I_{gr} \cup I_{ls}$.
Furthermore, let $O_{gr}$ and $O_{ls}$ be optimal (maximum weight) ISs in
$G[V_{gr}]$ and $G[V_{ls}]$, respectively.
Since
$O \cap V_{gr}$
and
$O \cap V_{ls}$
are ISs in $G[V_{gr}]$ and $G[V_{ls}]$, respectively, it follows that
$\Weight(O)
=
\Weight(O \cap V_{gr})
+
\Weight(O \cap V_{ls})
\leq
\Weight(O_{gr}) + \Weight(O_{ls})$.
We establish \Thm{}~\ref{theorem:non-uniform-main-algorithm-approximation} by
proving \Lem{}\ \ref{lemma:non-uniform-approximation-gr} and
\ref{lemma:non-uniform-approximation-ls}, implying that
\[\textstyle
\Weight(O)
\, \leq \,
\Weight(I_{gr}) \cdot \max \{ k - 1, \criticalRatio(\lenUtil) \}
+
\Weight(I_{ls}) \cdot (k - 1 + 1 / q)
\, \leq \,
\Weight(I) \cdot \max \{ k - 1 + 1 / q, \criticalRatio(\lenUtil) \} \, .
\]

\begin{lemma} \label{lemma:non-uniform-approximation-gr}
$\Weight(O_{gr})
\leq
\Weight(I_{gr}) \cdot \max \{ k - 1, \criticalRatio(\lenUtil) \}$.
\end{lemma}
\begin{proof}
By the design of
$\GreedyAlg^{\leq \ell^{*}}$,
we know that all nodes in
$V_{ls} = V - V_{gr}$
are of length $ > \ell^{*}$.
Therefore, the invocation of
$\GreedyAlg^{\leq \ell^{*}}$
on $G$ is identical to the invocation of \GreedyAlg{} on $G[V_{gr}]$.
Employing \Obs{}~\ref{observation:greedy-mapping}
with
$J = O_{gr}$,
we conclude that there exists a function
$\mu : O_{gr} \rightarrow I_{gr}$
such that
(1)
$|\Agents(u)| \geq |\Agents(\mu(u))|$
for every node
$u \in O_{gr}$;
and
(2)
$|\mu^{-1}(v)| \leq |\Agents(v)|$
for every node
$v \in I_{gr}$.

Let
$v_{1}, \dots, v_{r} \in I_{gr}$
be the images of $\mu$ and partition $O_{gr}$ into the (pairwise disjoint)
clusters
$O_{1}, \dots, O_{r}$
defined so that
$O_{j} = \mu^{-1}(v_{j})$
for each
$j \in [r]$.
We argue that
\begin{equation}\textstyle
\label{equation:non-uniform-gr-cluster-contribution}
\Weight(O_{j})
\, \leq \,
\Weight(v_{j}) \cdot \max \left\{ k - 1, \criticalRatio(\lenUtil) \right\}
\end{equation}
for each
$j \in [r]$.
This implies that
\[\textstyle
\Weight(O_{gr})
\, = \, 
\sum_{j = 1}^{r} \Weight(O_{j})
\, \leq \,
\max \left\{ k - 1, \criticalRatio(\lenUtil) \right\} \cdot 
\sum_{j = 1}^{r} \Weight(v_{j})
\, \leq \,
\max \left\{k - 1, \criticalRatio(\lenUtil) \right\} \cdot \Weight(I_{gr}) \, ,
\]
thus establishing the assertion.

To show that \eqref{equation:non-uniform-gr-cluster-contribution} holds,
consider some 
$j \in [r]$
and recall that 
$|O_{j}| \leq |\Agents(v_{j})|$, 
hence
\begin{align*}
\Weight(O_{j})
\, = \, &
\sum_{u \in O_{j}} |\Agents(u)| \cdot \lenUtil(|\Agents(u)|)
\\
\leq \, &
|O_{j}| \cdot 
\max_{u \in O_{j}} \{ |\Agents(u)| \cdot \lenUtil(|\Agents(u)|) \}
\, \leq \,
|\Agents(v_{j})|
\cdot
\max_{u \in O_{j}} \{ |\Agents(u)| \cdot \lenUtil(|\Agents(u)|) \} \, .
\end{align*}
The design of
$\GreedyAlg^{\leq \ell^{*}}$ 
ensures that 
$2 \leq |\Agents(v_{j})| \leq \ell^{*}$
and since we know that
$|\Agents(v_{j})| \leq |\Agents(u)| \leq k$
for every
$u \in O_{j}$,
we deduce that
\begin{align*}
w(O_{j})
\, \leq \, &
\frac{\Weight(v_{j})}{\lenUtil(|\Agents(v_{j})|)}
\cdot
\max_{|\Agents(v_{j})| \leq \ell' \leq k} \ell' \cdot \lenUtil(\ell')
\, \leq \,
\Weight(v_{j})
\cdot
\max_{2 \leq \ell \leq \ell^{*}} \max_{\ell \leq \ell' \leq k}
\frac{\ell' \cdot \lenUtil(\ell')}{\lenUtil(\ell)} \\
= \, &
\Weight(v_{j})
\cdot
\max \left\{
\max_{2 \leq \ell \leq \ell^{*}} \frac{\ell \cdot \lenUtil(\ell)}{\lenUtil(\ell)}
\, ,
\max_{2 \leq \ell < \ell' \leq k%
\, : \, \lenUtil(\ell) > \lenUtil(\ell')%
}
\frac{\ell' \cdot \lenUtil(\ell')}{\lenUtil(\ell)}
\right\} \, ,
\end{align*}
where the last transition holds by separating the case that
$\lenUtil(\ell) = \lenUtil(\ell')$ 
from its complement.
Inequality~\eqref{equation:non-uniform-gr-cluster-contribution} follows as
$\max_{2 \leq \ell \leq \ell^{*}} \frac{\ell \cdot \lenUtil(\ell)}{\lenUtil(\ell)}
=
\ell^{*}
\leq
k - 1$
and
$\max_{2 \leq \ell < \ell' \leq k \, : \, \lenUtil(\ell) > \lenUtil(\ell')}
\frac{\ell' \cdot \lenUtil(\ell')}{\lenUtil(\ell)}
\leq
\criticalRatio(\lenUtil)$.
\end{proof}

\begin{lemma} \label{lemma:non-uniform-approximation-ls}
$\Weight(O_{ls})
\leq
\Weight(I_{ls}) \cdot (k - 1 + 1 / q)$.
\end{lemma}
\begin{proof}
Recall that the design of
$\GreedyAlg^{\leq \ell^{*}}$
ensures that all nodes in $V_{ls}$ are of length larger than $\ell^{*}$.
Therefore, the invocation of 
$\LocalSearchAlg_{q}^{> \ell^{*}}$ 
on $G[V_{ls}]$ is
identical to the invocation of $\LocalSearchAlg_{q}$ on $G[V_{ls}]$.
Moreover, the length function $\lenUtil$ plays no role in the
$(k, \lenUtil)$-maxCGIS
problem on $G[V_{ls}]$ as
$\lenUtil(|\Agents(u)|) = \lenUtil(|\Agents(v)|)$
for every two nodes
$u, v \in V_{ls}$,
hence the
$(k, \lenUtil)$-maxCGIS
problem is equivalent to the $k$-maxCGIS problem on $G[V_{ls}]$.
The assertion follows by
\Thm{}~\ref{theorem:uniform-main-algorithm-approximation}.
\end{proof}

\section{The Analysis of $\LocalSearchAlg_{q}$ is Nearly Tight}
\label{section:main-algorithm-bad-example}
Fix
$k = 3$
and any parameter
$q \in \Integers_{> 0}$.
In this section, we present an example for a $k$-cycle graph
$G_{\text{bad}} = (V, E) \in \CycleGraphSet_{n}^{k}$
with an optimal IS
$O \subset V$
such that
\[
\frac{\Weight(O)}{\Weight(\LocalSearchAlg_{q}(G_{\text{bad}}))}
\, = \,
k - 1 + \frac{1}{q + 1}
\, ,
\]
thus demonstrating that the analysis in
\Sect{}~\ref{section:uniform-approximation-ratio} is nearly tight.
The $k$-cycle graph $G_{\text{bad}}$ is illustrated in
Figure~\ref{figure:main-algorithm-bad-example}.
(Alternatively, the graph illustrated in
Figure~\ref{figure:main-algorithm-bad-example} can be a connected component in
a larger $k$-cycle graph, allowing $n$ to become arbitrarily large with
respect to $q$.)
The IS
$I \subset V$
that consists of the blue (solid) rectangles admits a weight of
$\Weight(I) = 3 (q + 1)$.
A quick inspection reveals that
$r_{E}(G_{\text{bad}}, I) = r_{q}(G_{\text{bad}}, I) = \bot$,
where
$r_{E}$ and $r_{q}$ are the expansion and all-for-$q$ improvement rules,
respectively.
Therefore, with the right assignment of agent identities, we can ensure that
$I = \LocalSearchAlg_{q}(G_{\text{bad}})$.
On the other hand, the IS
$O \subset V$
that consists of the red (dashed) rectangles admits a weight of
$\Weight(O) = 3 (2 q + 3)$,
hence
$\frac{\Weight(O)}{\Weight(\LocalSearchAlg_{q}(G_{\text{bad}}))}
=
k - 1 + \frac{1}{q + 1}$.

\begin{figure}
\centering
\includegraphics[scale=0.85]{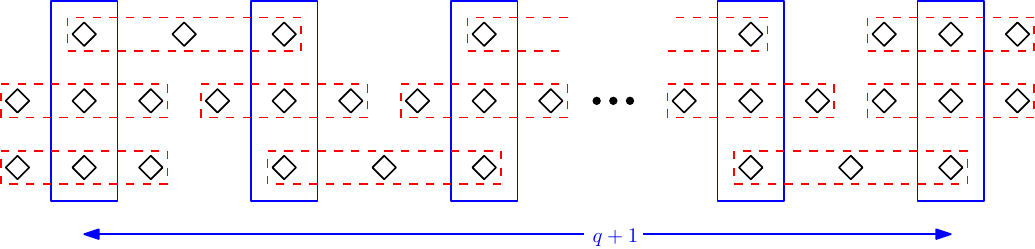}
\caption{\label{figure:main-algorithm-bad-example}%
The $k$-cycle graph
$G_{\text{bad}} = (V, E)$
whose nodes are depicted by the
$q + 1$
blue (solid) rectangles and the
$2 q + 3$
red (dashed) rectangles, defined over a set of
$n = 6 q + 9 $
agents depicted by the black diamond shapes, so that the diamond shape
corresponding to agent
$i \in [n]$
is contained in the rectangle corresponding to node
$v \in V$
if (and only if)
$i \in \Agents(v)$.}
\end{figure}

The $k$-cycle graph $G_{\text{bad}}$ also demonstrates that the IS constructed
by algorithm $\LocalSearchAlg_{q}$ may be far from maximal in terms of the set
of partaking agents (see footnote~\ref{footnote:agent-maximal-exchanges}).

\section{Lifting the Efficiency Requirement}
\label{section:no-running-time-limit} 
Consider the
$(k, \lenUtil)$-maxCGIS
problem for a length bound
$k \geq 3$
and a non-uniform length function $\lenUtil$.
In this section, we abandon the efficiency requirement of 
$(k, \lenUtil)$-maxCGIS 
algorithms as we turn to develop a (non-efficient) truthful algorithm for this problem.
The approximation ratio guaranty of this algorithm, combined with \Obs{}~\ref{observation:algorithm-implies-mechanism}, establishes \Thm{}~\ref{theorem:no-time-limit-upper-bound-non-uniform}.
To this end, 
for each
$2 \leq \ell \leq k$, 
let
$\Opt^{\ell}$
denote the algorithm that given an input $k$-cycle graph
$G = (V, E) \in \CycleGraphSet_{n}^{k}$,
$n \in \Integers_{> 0}$,
constructs the set 
$V^{\ell} = \{v \in V \mid \lenUtil(|\Agents(v)|) = \lenUtil(\ell)\}$,
and outputs the lexicographically first optimal (maximum weight) IS in the subgraph
$G[V^{\ell}]$.
Let
$L = \{ 2 \leq \ell \leq k - 1 \mid \lenUtil(\ell) > \lenUtil(\ell + 1) \} \cup \{ k \}$
be the set of lengths just before a ``tumble" in the length function $\lenUtil$,
and denote it explicitly by
$L = \{\ell_{1}, \dots, \ell_{|L|}\}$.
Our algorithm for the 
$(k, \lenUtil)$-maxCGIS
problem, 
denoted by $\IterativeOpt$ (abbreviation of iterative optimal), is defined to be the concatenation of the mechanisms $\Opt^{\ell}$, $\ell \in L$, 
namely,
$\IterativeOpt
=
\Opt^{\ell_{1}} \cdots \Opt^{\ell_{|L|}}$.

\begin{lemma}
\IterativeOpt{} is truthful.
\end{lemma}

\begin{proof}
Notice first that $\Opt^{\ell}$ is INPA for all 
$\ell \in L$,
as if
$i \notin \Opt^{\ell}(I)$,
then reducing $V$ to
$V - S$
for some node subset
$S \subseteq \Agents^{-1}(i)$
can only decrease the collection of optimal ISs in the subgraph induced by $V_{\ell}$ and cannot conceal
$\Opt^{\ell}(I)$.

Since algorithm $\Opt^{\ell}$ is designed so that all nodes $v$
included in its output ISs admit
$\lenUtil(|\Agents(v)|) = \lenUtil(\ell)$,
it follows that
$\Opt^{\ell}$ is oblivious to the length function $\lenUtil$.
In particular, we can employ
\Obs{}~\ref{observation:uniform-inpa-implies-truthful}, whose scope is
restricted (originally) to uniform length functions, to deduce that the fact
that $\Opt^{\ell}$ is INPA implies its truthfulness also
under non-uniform length functions.

Finally, the truthfulness of \IterativeOpt{} follows by Corollary~\ref{corollary:concatenation-truthful-algorithms}, as each $\Opt^{\ell}$ is INPA and truthful and it is clear that
$\Opt^{\ell_{s}} \succeq \Opt^{\ell_{s + 1}}$
for all
$1 \leq s \leq |L| - 1$.
\end{proof}

\subsection{Bounding the Approximation Ratio} 
Our goal in this section is to establish the following theorem that, combined
with \Obs{}~\ref{observation:algorithm-implies-mechanism},
yields
\Thm{}~\ref{theorem:no-time-limit-upper-bound-non-uniform}.

\begin{theorem} \label{theorem:iterative-opt-approximation}
The approximation ratio of \IterativeOpt{} is at most
$\criticalRatio(\lenUtil)$.
\end{theorem}

Consider a $k$-cycle graph
$G = (V, E) \in \CycleGraphSet_{n}^{k}$,
$n \in \Integers_{> 0}$,
input to \IterativeOpt{}.
Let 
$I = \IterativeOpt(G)$
be the IS output by \IterativeOpt{} and let $O$ be an optimal (maximum weight)
IS in $G$.
To avoid cumbersome expressions, we define 
$H_{\ell_{1}} = G$ 
and
\[
I_{\ell_{s}} = \Opt^{\ell_{s}} (H_{\ell_{s}})
\, , \quad
N_{\ell_{s}} = N_{H_{\ell_{s}}}(I_{\ell_{s}})
\, , \quad
H_{\ell_{s + 1}} = 
G \left[ V - \bigcup_{i = 1}^{s} (I_{\ell_{i}} \cup N_{\ell_{i}}) \right] 
\, ,
\]
for each
$1 \leq s \leq |L|$.
Thinking of \IterativeOpt{} as the concatenation of
$\Opt^{\ell_{1}}, \dots, \Opt^{\ell_{|L|}}$,
the subgraph $H_{\ell_{s}}$ represents the (remaining) input to $\Opt^{\ell_{s}}$,
the IS $I_{\ell_{s}}$ represents the node donation of the application of $\Opt^{\ell_{s}}$ to the returned IS $I$, and $N_{\ell_{s}}$ represents $I_{\ell_{s}}$'s neighbors in $H_{\ell_{s}}$ who are being removed, together with $I_{\ell_{s}}$, from the input to 
$\Opt^{\ell_{s + 1}}$.
Observe that
$I_{\ell_{1}}, \dots, I_{\ell_{|L|}}$
is a partition of the returned IS $I$ into pairwise disjoint clusters.
For each
$\ell \in L$,
let 
$L_{\ell} = \{\ell' \leq \ell \mid \lenUtil(\ell') = \lenUtil(\ell)\}$, and observe that
$\bigcup_{\ell \in L} L_{\ell} = \{2, 3, \dots, k\}$.
By the design of \IterativeOpt{} we deduce the following important observation.

\begin{observation} \label{observation:iterative-opt-component-length}
Since 
$I_{\ell_{s - 1}}$
is a (maximum weight IS and in particular) maximal IS in 
$H_{\ell_{s - 1}}$
for each
$2 \leq s \leq |L|$,
the graph
$H_{\ell_{s}}$ 
admits only nodes of length (strictly) greater that
$\ell_{s - 1}$.
Hence,
for each
$\ell \in L$,
every node 
$v \in I_{\ell}$
admits
$|\Agents(v)| \in L_{\ell}$,
and each node
$u \in N_{\ell}$
admits
$|\Agents(u)| \geq  \min_{\ell' \in L_{\ell}} \ell'$.
\end{observation}

For each
$\ell \in L$,
let $O_{\ell}$ be an optimal IS in 
$G[I_{\ell} \cup N_{\ell}]$.
Since
$V = \bigcup_{\ell \in L} (I_{\ell} \cup N_{\ell})$
and since
$O \cap (I_{\ell} \cup N_{\ell})$
is an IS in 
$G[I_{\ell} \cup N_{\ell}]$
for all
$\ell \in L$,
it follows that
\[
\Weight(O) = 
\sum_{\ell \in L} \Weight(O \cap (I_{\ell} \cup N_{\ell})) \leq
\sum_{\ell \in L} \Weight(O_{\ell}) \, .
\]
We shall prove \Lem{}~\ref{lemma:non-efficient-component-approximation}, implying that
\[
\Weight(O)
\leq
\sum_{\ell \in L}
\Weight(I_{\ell}) \cdot  \beta_{\ell}
\leq
\Weight(I) \cdot \max_{\ell \in L} \{\beta_{\ell}\} \, ,
\]
where
\[
\beta_{\ell} = 
\max \left\{ 
\frac{\max_{\ell < \ell' \leq k} \{\ell' \lenUtil(\ell')\}}{\lenUtil(\ell)} ,
1 + 
\frac{\ell - 1}{\ell} \cdot 
\frac{\max_{\ell < \ell' \leq k} \{\ell' \lenUtil(\ell')\}}{\lenUtil(\ell)}
\right\} \, .
\]

This establishes \Thm{}~\ref{theorem:iterative-opt-approximation} as
\begin{align*}
\max_{\ell \in L} \{\beta_{\ell}\} &= 
\max_{\ell \in L} \left\{
\max \left\{ 
\frac{\max_{\ell < \ell' \leq k} \{\ell' \lenUtil(\ell')\}}{\lenUtil(\ell)} ,
1 + 
\frac{\ell - 1}{\ell} \cdot 
\frac{\max_{\ell < \ell' \leq k} \{\ell' \lenUtil(\ell')\}}{\lenUtil(\ell)}
\right\}
\right\} \\
&=
\max_{2 \leq \ell < \ell' \leq k \, : \, \lenUtil(\ell') < \lenUtil(\ell)} 
\left\{
\max \left\{ \frac{\ell' \lenUtil(\ell')}{\lenUtil(\ell)}, 
1 + \frac{\ell - 1}{\ell} \cdot \frac{\ell' \lenUtil(\ell')}{\lenUtil(\ell)} \right\}
\right\} 
=
\criticalRatio(\lenUtil) \, .
\end{align*}

\begin{lemma} \label{lemma:non-efficient-component-approximation}
$w(O_{\ell}) \leq w(I_{\ell}) \cdot 
\max \left\{ 
\frac{\max_{\ell < \ell' \leq k} \{\ell' \lenUtil(\ell')\}}{\lenUtil(\ell)} ,
1 + 
\frac{\ell - 1}{\ell} \cdot 
\frac{\max_{\ell < \ell' \leq k} \{\ell' \lenUtil(\ell')\}}{\lenUtil(\ell)}
\right\}$
for each 
$\ell \in L$. 
\end{lemma}

\begin{proof}
Consider some
$\ell \in L$, and assume without loss of generality that $I_{\ell}$ and $O_{\ell}$ are disjoint, as nontrivial intersection of them improves the approximation ratio of $I_{\ell}$.
In order to explicitly calculate their weight, we partition the ISs $I_{\ell}$ and $O_{\ell}$ according to node lengths as follows.
\begin{itemize}
\item 
Set
$I_{\ell}^{j} = \{v \in I_{\ell} \mid |\Agents(v)| = j\}$
for each
$j \in L_{\ell}$.
\item 
Set
$O_{\ell}^{j} = \{v \in O_{\ell} \mid |\Agents(v)| = j\}$
for each 
$j \in L_{\ell}$.
\item 
Set
$O_{\ell}^{>\ell} = \{v \in O_{\ell} \mid |\Agents(v)| > \ell\}$.
\end{itemize}

\Obs{}~\ref{observation:iterative-opt-component-length} ensures that
$I_{\ell} = \biguplus_{j \in L_{\ell}} I_{\ell}^{j}$
and
$O_{\ell} = \biguplus_{j \in L_{\ell}} O_{\ell}^{j} \uplus O_{\ell}^{>\ell}$.
Hence, we can write the weights of 
$I_{\ell}, O_{\ell}$
as
\begin{equation}
\label{equation:iterative-opt-1}
\Weight(I_{\ell}) = 
\sum_{j \in L_{\ell}} \Weight(I_{\ell}^{j}) = 
\sum_{j \in L_{\ell}} |I_{\ell}^{j}| \cdot j \lenUtil(j) =
\lenUtil(\ell) \cdot \sum_{j \in L_{\ell}} j \cdot |I_{\ell}^{j}| \, ,
\end{equation}
and
\[
\Weight(O_{\ell})  =
\sum_{j \in L_{\ell}} \Weight(O_{\ell}^{j}) + \Weight(O_{\ell}^{>\ell}) =
\sum_{j \in L_{\ell}} |O_{\ell}^{j}| \cdot j \lenUtil(j) + \Weight(O_{\ell}^{>\ell}) =
\lenUtil(\ell) \cdot \sum_{j \in L_{\ell}} j \cdot |O_{\ell}^{j}| + \Weight(O_{\ell}^{>\ell}) \, .
\]
By the definition of $O_{\ell}^{>\ell}$, we deduce that
\[
\Weight(O_{\ell}^{>\ell}) \leq
|O_{\ell}^{>\ell}| \cdot
\max_{u \in O_{\ell}^{>\ell}} \{|\Agents(u)| \cdot \lenUtil(|\Agents(u)|)\} \leq
|O_{\ell}^{>\ell}| \cdot 
\max_{\ell < \ell' \leq k} \{\ell' \cdot \lenUtil(\ell')\} \, ,
\]
which implies that
\begin{equation} \label{equation:iterative-opt-2}
\Weight(O_{\ell})  \leq
\lenUtil(\ell) \cdot \sum_{j \in L_{\ell}} j \cdot |O_{\ell}^{j}| + |O_{\ell}^{>\ell}| \cdot 
\max_{\ell < \ell' \leq k} \{\ell' \cdot \lenUtil(\ell')\} \, .
\end{equation}
Recall that 
$V^{\ell} = \{v \in V \mid \lenUtil(|\Agents(v)|) = \lenUtil(\ell)\}$.
Since $I_{\ell}$ is a maximum weight IS in the subgraph $H_{\ell}[V_{\ell}]$ and the node set 
$\bigcup_{j \in L_{\ell}} O_{\ell}^{j}$ 
is an IS in $H_{\ell}[V_{\ell}]$, it follows that
\[
\lenUtil(\ell) \cdot \sum_{j \in L_{\ell}} j \cdot |O_{\ell}^{j}| =
\Weight(\bigcup_{j \in L_{\ell}} O_{\ell}^{j}) \leq 
\Weight(I_{\ell}) =
\lenUtil(\ell) \cdot \sum_{j \in L_{\ell}} j \cdot |I_{\ell}^{j}|  
\, ,
\] 
which implies that
\begin{equation} \label{equation:iterative-opt-3}
\sum_{j \in L_{\ell}} j \cdot |O_{\ell}^{j}| \leq 
\sum_{j \in L_{\ell}} j \cdot |I_{\ell}^{j}| \, .
\end{equation}
Recall our beloved \Obs{}~\ref{observation:cycle-graph-structure},
stating that a node of length $j$ in $I_{\ell}$ has at most $j$
neighbors in $O_{\ell}$.
Using the fact that $I_{\ell}$ is a maximal IS in 
$G[I_{\ell} \cup N_{\ell}]$
and that 
$I_{\ell}, O_{\ell}$
are disjoint,
we deduce that
\[
|O_{\ell}|
=
|\Neighbors_{G}(I_{\ell}) \cap O_{\ell}| 
\, \leq \,
\sum_{v \in I_{\ell}} |\Neighbors_{G}(v) \cap O_{\ell}|
\, = \,
\sum_{j \in L_{\ell}} \sum_{v \in I_{\ell}^{j}} |\Neighbors_{G}(v) \cap O_{\ell}|
\, \leq \,
\sum_{j \in L_{\ell}} j \cdot |I_{\ell}^{j}| \, ,
\]
hence
\begin{equation} \label{equation:iterative-opt-4}
|O_{\ell}^{>\ell}| \leq
\sum_{j \in L_{\ell}} j \cdot |I_{\ell}^{j}| -
\sum_{j \in L_{\ell}} |O_{\ell}^{j}| \, .
\end{equation}

The rest of the proof is divided into two cases depending on whether the following inequality holds
\begin{equation} \label{equation:iterative-opt-*}
\ell \lenUtil(\ell) \leq \max_{\ell < \ell' \leq k} \{\ell' \lenUtil(\ell')\} \, .
\end{equation}
We shall prove a different upper bound to $\Weight(O_{\ell})$ in each case and we will get the desired approximation ratio by taking the maximum of those bounds.

Assume first that \eqref{equation:iterative-opt-*} holds.
Applying a few straightforward algebraic transitions, it follows that
\begin{align*}
\Weight(O_{\ell}) 
& \stackrel{\eqref{equation:iterative-opt-2}}{\leq} 
\lenUtil(\ell) \cdot \sum_{j \in L_{\ell}} |O_{\ell}^{j}| \cdot j  +
\max_{\ell < \ell' \leq k} \{\ell' \lenUtil(\ell')\} \cdot |O_{\ell}^{>\ell}| \\
& \stackrel{\eqref{equation:iterative-opt-*}}{\leq}
\max_{\ell < \ell' \leq k} \{\ell' \lenUtil(\ell')\} \cdot
\left(
\frac{1}{\ell} \cdot \sum_{j \in L_{\ell}} j \cdot |O_{\ell}^{j}| + |O_{\ell}^{>\ell}|
\right) \\
& \stackrel{\eqref{equation:iterative-opt-4}}{\leq}
\max_{\ell < \ell' \leq k} \{\ell' \lenUtil(\ell')\} \cdot
\left(
\frac{1}{\ell} \cdot \sum_{j \in L_{\ell}} j \cdot |O_{\ell}^{j}| + 
\sum_{j \in L_{\ell}} j \cdot |I_{\ell}^{j}| -
\sum_{j \in L_{\ell}} |O_{\ell}^{j}|
\right) \\
& \stackrel{\forall j \in L_{\ell}: \, j \leq \ell}{\leq}
\max_{\ell < \ell' \leq k} \{\ell' \lenUtil(\ell')\} \cdot
\left(
\sum_{j \in L_{\ell}} j \cdot |I_{\ell}^{j}| 
\right) \\
& \stackrel{\eqref{equation:iterative-opt-1}}{=}
\Weight(I_{\ell}) \cdot 
\frac{\max_{\ell < \ell' \leq k} \{\ell' \lenUtil(\ell')\}}{\lenUtil(\ell)} \, .
\end{align*}

On the other hand, if inequality~\eqref{equation:iterative-opt-*} is not satisfied,
we can take advantage of our aforementioned equations and use algebraic manipulations to deduce that

\begin{align*}
\Weight(O_{\ell}) 
& \stackrel{\eqref{equation:iterative-opt-2}}{\leq} 
\lenUtil(\ell) \cdot \sum_{j \in L_{\ell}} |O_{\ell}^{j}| \cdot j  +
\max_{\ell < \ell' \leq k} \{\ell' \lenUtil(\ell')\} \cdot |O_{\ell}^{>\ell}| \\
& \stackrel{\eqref{equation:iterative-opt-4}}{\leq}
\sum_{j \in L_{\ell}} j \cdot |O_{\ell}^{j}| \cdot \lenUtil(\ell) +
\max_{\ell < \ell' \leq k} \{\ell' \lenUtil(\ell')\} \cdot \sum_{j \in L_{\ell}} j \cdot |I_{\ell}^{j}| -
\max_{\ell < \ell' \leq k} \{\ell' \lenUtil(\ell')\} \cdot \sum_{j \in L_{\ell}} |O_{\ell}^{j}| \\
& \stackrel{\eqref{equation:iterative-opt-1}}{=}
\Weight(I_{\ell}) \cdot 
\frac{\max_{\ell < \ell' \leq k} \{\ell' \lenUtil(\ell')\}}{\lenUtil(\ell)} +
\sum_{j \in L_{\ell}} j \cdot |O_{\ell}^{j}|
\left(\lenUtil(\ell) - \frac{\max_{\ell < \ell' \leq k} \{\ell' \lenUtil(\ell')\}}{j} \right) \\
& \stackrel{\forall j \in L_{\ell}: \, j \leq \ell}{\leq}
\Weight(I_{\ell}) \cdot 
\frac{\max_{\ell < \ell' \leq k} \{\ell' \lenUtil(\ell')\}}{\lenUtil(\ell)} +
\left(\lenUtil(\ell) - \frac{\max_{\ell < \ell' \leq k} \{\ell' \lenUtil(\ell')\}}{\ell} \right) \cdot
\sum_{j \in L_{\ell}} 
j \cdot |O_{\ell}^{j}| \\
& \stackrel{\eqref{equation:iterative-opt-3}, \neg \eqref{equation:iterative-opt-*}}{\leq} 
\Weight(I_{\ell}) \cdot 
\frac{\max_{\ell < \ell' \leq k} \{\ell' \lenUtil(\ell')\}}{\lenUtil(\ell)} +
\left(\lenUtil(\ell) - \frac{\max_{\ell < \ell' \leq k} \{\ell' \lenUtil(\ell')\}}{\ell} \right) \cdot
\sum_{j \in L_{\ell}} 
j \cdot |I_{\ell}^{j}|  \\
& \stackrel{\eqref{equation:iterative-opt-1}}{=} 
\Weight(I_{\ell}) \cdot 
\frac{\max_{\ell < \ell' \leq k} \{\ell' \lenUtil(\ell')\}}{\lenUtil(\ell)} +
\left(\lenUtil(\ell) - \frac{\max_{\ell < \ell' \leq k} \{\ell' \lenUtil(\ell')\}}{\ell} \right) \cdot
\frac{\Weight(I_{\ell})}{\lenUtil(\ell)} \\
& = 
\Weight(I_{\ell}) \cdot \left(
\frac{\max_{\ell < \ell' \leq k} \{\ell' \lenUtil(\ell')\}}{\lenUtil(\ell)} +
1 -
\frac{\max_{\ell < \ell' \leq k} \{\ell' \lenUtil(\ell')\}}{\ell \lenUtil(\ell)}
\right) \\
& = 
\Weight(I_{\ell}) \cdot \left(
1 + 
\frac{\ell - 1}{\ell} \cdot 
\frac{\max_{\ell < \ell' \leq k} \{\ell' \lenUtil(\ell')\}}{\lenUtil(\ell)}
\right)
\, .
\end{align*}

Finally, by taking the maximum of the bounds on $\Weight(O_{\ell})$ from both cases we conclude that 
\[
\Weight(O_{\ell}) \leq
\Weight(I_{\ell}) \cdot
\max \left\{ 
\frac{\max_{\ell < \ell' \leq k} \{\ell' \lenUtil(\ell')\}}{\lenUtil(\ell)} ,
1 + 
\frac{\ell - 1}{\ell} \cdot 
\frac{\max_{\ell < \ell' \leq k} \{\ell' \lenUtil(\ell')\}}{\lenUtil(\ell)}
\right\} \, ,
\]
thus completing the proof of
\Lem{}~\ref{lemma:non-efficient-component-approximation}.
\end{proof}

\section{Lower Bound for Non-Uniform Length Functions} 
\label{section:lower-bound-non-uniform}
Consider a length bound
$k \geq 3$
and a non-uniform length function $\lenUtil$.
In this section, we establish the lower bound for the
$(k, \lenUtil)$-BE
problem promised in \Thm{}~\ref{theorem:non-uniform-lower-bound}.
To do so, we first introduce a
$(k, \lenUtil)$-BE
instance, called
\emph{$(h, v)$-comb},
and use this instance to establish a
$\max_{%
2 \leq \ell < \ell' \leq k \, : \, \lenUtil(\ell) > \lenUtil(\ell')%
}
\left\{
\tfrac{\ell' \cdot \lenUtil(\ell')}{\lenUtil(\ell)}
\right\}$
lower bound on the approximation ratio of any truthful mechanism;
this is done in \Sect{}~\ref{section:weak-lower-bound}.
Subsequently, in \Sect{}~\ref{section:strong-lower-bound}, we introduce a
slightly more sophisticated instance, called
\emph{double $(h, v)$-comb},
and establish the improved
$\max_{%
2 \leq \ell < \ell' \leq k \, : \, \lenUtil(\ell) > \lenUtil(\ell')%
}
\left\{
\max \left\{
\frac{\ell' \cdot \lenUtil(\ell')}{\lenUtil(\ell)}, \,
\frac{\ell - 1}{\ell}
\cdot
\frac{\ell' \cdot \lenUtil(\ell')}{\lenUtil(\ell)}
+ 1
\right\}
\right\}$
lower bound promised in \Thm{}~\ref{theorem:non-uniform-lower-bound} (building
on some arguments presented in \Sect{}~\ref{section:weak-lower-bound}).
We start by introducing the equivalent interpretation of a wish list vector
as an agent digraph, which will make our life easier in the representation and
discussion about specific instances of the
$(k, \lenUtil)$-BE
problem and shall serve us in the current section and in
\Sect{}~\ref{section:lower-bound-local-search}.

Consider a wish list vector
$\Wish \in \WishVecSpace_{n}$,
$n \in \Integers_{> 0}$.
Throughout this section we interpret $\Wish$ as a digraph
$\Wish = ([n], E_{\Wish})$,
referred to as the \emph{agent digraph}, that includes a (directed) arc from
vertex
$i \in [n]$
to vertex
$j \in [n]$
if and only if
$j \in \Wish_{i}$.
Frequently, it will be helpful to think of the strategy space of an agent as
the collection of all subsets of outgoing arcs she might reveal to the
mechanism.
Observe that the set $\Cycles_{\Wish}^{k}$ identifies with the set of
(directed) cycles in the agent digraph $\Wish$ of length at most $k$ and
following the discussion in \Sect{}~\ref{section:model}, notice that
$\Cycles_{\pi}$ is a set of pairwise (vertex-)disjoint cycles in
$\Cycles_{\Wish}^{k}$ for each exchange
$\pi \in \ExcSet_{\Wish}^{k}$.
This facilitates a graph theoretic view of the (restricted version of the)
$(k, \lenUtil)$-BE
problem:
an instance of the problem is characterized by an agent digraph
$\Wish^{*} = ([n], E_{\Wish^{*}})$;
a mechanism $f$ for the problem maps a given agent digraph
$\Wish = ([n], E_{\Wish})$,
where
$E_{\Wish} \subseteq E_{\Wish^{*}}$,
to an exchange
$f(\Wish) \in \ExcSet_{\Wish}^{k}$.

\subsection{A Simplified Construction} 
\label{section:weak-lower-bound}
In this section, we present a simple-to-analyze instance of the 
$(k, \lenUtil)$-BE
problem, and show that it yields a (weak) lower bound to this problem.
Throughout this section we fix a non-uniform length function $\lenUtil$ and
two integers
$2 \leq h < v \leq k$
such that
$\lenUtil(v) < \lenUtil(h)$.

\begin{definition*}[comb instance]
The
\emph{$(h, v)$-comb} instance
is defined to be the agent digraph
$\Wish \in \WishVecSpace_{h v}$
depicted in Figure~\ref{figure:comb-instance},
consisting of $h$ 'black agents' and 
$h(v - 1)$
'white agents'.\footnote{%
To avoid cumbersome expressions,
we do not mention the parameters
$h, v$
in the notation of the 
$(h, v)$-comb
instance, recalling that we fixed
$h, v$
for this section.
}
The cycles in the set 
$\Cycles_{\Wish}^{k}$,
which identifies with the set of (directed) cycles in the agent digraph
$\Wish$,
are highlighted using stretched ellipses with gray color.
We denote the only horizontal cycle by 
$c_{H} = (1, \dots, h)$
and we define $\pi_{H}$ as the exchange for whom
$\Cycles_{\pi_{H}} = \{c_{H}\}$.
Additionally, for every
$i \in [h]$,
we denote by 
$\Wish'_{i} = \{i + 1\}$
(except for 
$\Wish'_{h} = \{1\}$)
the wish list for whom agent $i$ reveals only her outgoing arc to a black
agent.
This notation is naturally extended to agent subsets $[i]$,
$1 \leq i \leq h$, defining 
$\Wish'_{[i]} = (\Wish'_{1}, \Wish'_{2}, \dots, \Wish'_{i})$.
\end{definition*}

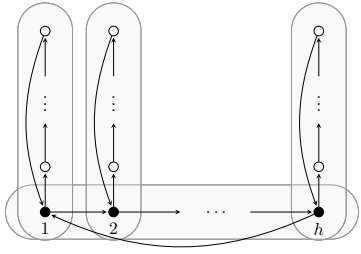
\begin{figure}
\centering
\scalebox{0.6}{
\begin{tikzpicture}
\node[cycrec = 8, fill = lightgray!10] at (3,0) {};
\node[cycrec = 5.5, rotate=90, fill = lightgray!10] at (0,2) {};
\node[cycrec = 5.5, rotate=90, fill = lightgray!10] at (1.5,2) {};
\node[cycrec = 5.5, rotate=90, fill = lightgray!10] at (6,2) {};
\node[cycrec = 8] at (3,0) {};
\node[cycrec = 5.5, rotate=90] at (0,2) {};
\node[cycrec = 5.5, rotate=90] at (1.5,2) {};
\node[cycrec = 5.5, rotate=90] at (6,2) {};

\black (p1) at (0,0) [label=below:$1$] {};
\black (p2) at (1.5,0) [label=below:$2$] {};
\black (pn) at (6,0) [label=below:$h$] {};
	
\white (v11) at (0,1) [] {};
\white (v1k) at (0,4) [] {};
	
\white (v21) at (1.5,1) [] {};
\white (v2k) at (1.5,4) [] {};
		
\white (vn1) at (6,1) [] {};
\white (vnk) at (6,4) [] {};
		
\path 
(p2) -- node[auto=false]{$\ldots$} (pn)
(p1) edge[outarrow] (p2)
(p2) edge[outarrow] (3,0)
(4.5,0) edge[outarrow] (pn)
(pn) edge[outarrow, bend left = 25] (p1)
		
(v11) -- node[auto=true]{$\vdots$} (v1k)
		
(p1) edge[outarrow] (v11)
(v11) edge[outarrow] (0,2)
(0,3) edge[outarrow] (v1k)
(v1k) edge[outarrow, bend right = 20] (p1)
		
(v21) -- node[auto=true]{$\vdots$} (v2k)
		
(p2) edge[outarrow] (v21)
(v21) edge[outarrow] (1.5,2)
(1.5,3) edge[outarrow] (v2k)
(v2k) edge[outarrow, bend right = 20] (p2)
		
(vn1) -- node[auto=true]{$\vdots$} (vnk)
		
(pn) edge[outarrow] (vn1)
(vn1) edge[outarrow] (6,2)
(6,3) edge[outarrow] (vnk)
(vnk) edge[outarrow, bend right = 20] (pn)
;
\end{tikzpicture}
}
\caption{\label{figure:comb-instance}%
The agent digraph of an
$(h, v)$-comb instance,
where
$2 \leq h < v \leq k$
and
$\lenUtil(v) < \lenUtil(h)$,
that consists of $h$ 'black agents' and
$h (v - 1)$
'white agents'.
The (directed) cycles in the agent digraph are depicted by the stretched gray
ellipses.
Any truthful mechanism must output the exchange $\pi_{H}$ defined so that
$\Cycles_{\pi_{H}} = \{ c_{H} \}$,
where $c_{H}$ is the 'horizontal cycle' that consists of the black agents;
the social welfare of this exchange is
$\SW(\pi_{H}) = h \cdot \lenUtil(h)$.
In contrast, the social welfare of the exchange $\pi$, defined so that
$\Cycles_{\pi}$ includes all 'vertical cycles', is
$\SW(\pi) = h \cdot v \cdot \lenUtil(v)$.
Maximizing over $h$ and $v$, we obtain a lower bound of
$\max_{%
2 \leq h < v \leq k \, : \, \lenUtil(h) > \lenUtil(v)%
}
\left\{
\tfrac{v \cdot \lenUtil(v)}{\lenUtil(h)}
\right\}$
on the approximation ratio.
}
\end{figure}

\begin{lemma} \label{lemma:comb-instance-output}
Consider a deterministic truthful mechanism $f$ (that admits a constant approximation ratio) to the 
$(k, \lenUtil)$-BE problem.
Then,
$f(\Wish) = \pi_{H}$.
\end{lemma}

\begin{proof}
For each 
$0 \leq i \leq h$, 
let
$\Wish^{i} = (\Wish_{-[i]}, \Wish'_{[i]})$
be the agent digraph in which all agents report truthfully except for agents
$1, \dots ,i$
who reveal only their outgoing arc to a black agent
(note that
$\Wish^{0} = \Wish$).
Observe that
$\Cycles_{\Wish^{i}}^{k}$,
which identifies with the set of (directed) cycles in the agent digraph
$\Wish^{i}$,
consists of the horizontal cycle $c_{H}$ and the
$h - i$ 
rightmost vertical cycles.
We shall prove by induction that 
$f(\Wish^{i}) = \pi_{H}$ 
for all
$i = h, h - 1, \dots, 0$, 
which will establish the assertion as 
$f(\Wish) = f(\Wish^{0})$.

The induction base concerns the case of 
$i = h$,
and follows immediately as 
$\Cycles_{\Wish^{h}}^{k}$
consists only the horizontal cycle $c_{H}$, which must be included in $\Cycles_{f(\Wish^{h})}$ 
(otherwise, $f$ does not admit a constant approximation ratio).
For the induction step, consider some  
$1 \leq i \leq h$
and assume that 
$f(\Wish^{i}) = \pi_{H}$.
Notice that as
$\Wish^{i} = (\Wish^{i - 1}_{-i}, \Wish'_{i})$,
we know that
$\utility(i, f(\Wish^{i - 1}_{-i}, \Wish'_{i})) = 
\utility(i, f(\Wish^{i})) =
\lenUtil(h)$.
Consequently, the truthfulness of $f$ guarantees that
$\utility(i, f(\Wish^{i - 1})) \geq
\utility(i, f(\Wish^{i - 1}_{-i}, \Wish'_{i})) = 
\lenUtil(h)$,
which can be realized only if 
$f(\Wish^{i - 1}) = \pi_{H}$
as
$\lenUtil(h) > \lenUtil(v)$,
thus completing the induction step.
\end{proof}

Notice that the aforementioned lemma already establishes a (decent) lower bound for the 
$(k, \lenUtil)$-BE
problem, which can be derived as follows. 
Consider the 
$(h, v)$-comb
instance.
\Lem{}~\ref{lemma:comb-instance-output} ensures that a deterministic truthful mechanism obtains a social welfare of
$h \cdot \lenUtil(h)$.
On the other hand, observe that the exchange who possess all vertical cycles in its trading cycle set admits a social welfare of 
$h \cdot v \cdot \lenUtil(v)$.
This yields a lower bound of 
$\frac{h v \lenUtil(v)}{h \lenUtil(h)} = \frac{v \lenUtil(v)}{\lenUtil(h)}$.
Going over all pairs 
$(h, v)$
that satisfy
$\lenUtil(v) < \lenUtil(h)$
derives the following impossibility result. 

\begin{corollary} \label{corollary:weak-lower-bound}
The
$(k, \lenUtil)$-BE
problem does not admit a truthful (deterministic) mechanism with approximation
ratio strictly smaller than 
$
\max_{2 \leq h < v \leq k \, : \, \lenUtil(v) < \lenUtil(h)}
\left\{
\frac{v \lenUtil(v)}{\lenUtil(h)}
\right\}
$ 
for every length bound
$k \geq 3$
and non-uniform length function $\lenUtil$.
\end{corollary}

\subsection{Enhancing the Construction} 
\label{section:strong-lower-bound}
In the current section we develop the lower bound promised in \Thm{}~\ref{theorem:non-uniform-lower-bound},
by taking advantage of an augmented instance to the one presented in \Sect{}~\ref{section:weak-lower-bound}.
Once again, throughout this section we fix a non-uniform length function $\lenUtil$ and two integrals 
$2 \leq h < v \leq k$
such that
$\lenUtil(v) < \lenUtil(h)$.
We note that in this section we shall redefine terms from \Sect{}~\ref{section:weak-lower-bound}, but we will refer to the definitions as presented in the current section.

\begin{definition*}[double comb instance]
The
\emph{double 
$(h, v)$-comb}
instance is defined to be the agent digraph 
$\Wish \in \WishVecSpace_{(2h - 1) v}$
depicted in Figure~\ref{figure:two-sided-comb},
consisting of 
$2h - 1$
'black agents' and
$(2h - 1) (v - 1)$
'white agents'.\footnote{%
Once again, to ease the notations, we do not mention the parameters
$h, v$
in the notation of the double
$(h, v)$-comb
instance, recalling that we fixed
$h, v$
for this section.
}
The cycles in the set 
$\Cycles_{\Wish}^{k}$,
which identifies with the set of (directed) cycles in the agent digraph
$\Wish$,
are highlighted using stretched ellipses with gray color.
We denote the left horizontal cycle by
$c_{L} = (l_{1}, \dots, l_{h})$,
the right horizontal cycle by 
$c_{R} = (l_{h}, r_{1}, \dots, r_{h - 1})$,
and we define 
$\pi_{L}, \pi_{R}$ as the exchanges for whom
$\Cycles_{\pi_{L}} = \{c_{L}\}, \Cycles_{\pi_{R}} = \{c_{R}\}$.
Additionally, for every
$i \in [h - 1]$
we denote by 
$\Wish'_{r_{i}} = \{r_{i + 1}\}$
(except for
$\Wish'_{r_{h - 1}} = \{l_{h}\}$)
the wish list for whom agent $r_{i}$ reveals only her outgoing arc to a black agent.
This notation is naturally extended to agent subsets 
$\{r_{1}, \dots, r_{i}\}$,
$1 \leq i \leq h - 1$, 
defining 
$\Wish'_{\{r_{1}, \dots, r_{i}\}} = (\Wish'_{r_{1}}, \Wish'_{r_{2}}, \dots, \Wish'_{r_{i}})$.
For each 
$0 \leq i \leq h - 1$, 
let
$\Wish^{i} = 
(\Wish_{-\{r_{1}, \dots, r_{i}\}}, \Wish'_{\{r_{1}, \dots, r_{i}\}})$
be the agent digraph in which all agents report truthfully except for agents
$r_{1}, \dots, r_{i}$
who reveal only their outgoing arc to a black agent
(note that
$\Wish^{0} = \Wish$).
\end{definition*}

\begin{figure}
\centering
\scalebox{0.7}{
\begin{tikzpicture}
\node[cycrec = 8, fill = lightgray!10] at (3,0) {};
\node[cycrec = 8, fill = lightgray!10] at (9,0) {};
\node[cycrec = 5.5, rotate=90, fill = lightgray!10] at (0,2) {};
\node[cycrec = 5.5, rotate=90, fill = lightgray!10] at (4.5,2) {};
\node[cycrec = 5.5, rotate=90, fill = lightgray!10] at (6,2) {};
\node[cycrec = 5.5, rotate=90, fill = lightgray!10] at (7.5,2) {};
\node[cycrec = 5.5, rotate=90, fill = lightgray!10] at (12,2) {};
\node[cycrec = 8] at (3,0) {};
\node[cycrec = 8] at (9,0) {};
\node[cycrec = 5.5, rotate=90] at (0,2) {};
\node[cycrec = 5.5, rotate=90] at (4.5,2) {};
\node[cycrec = 5.5, rotate=90] at (6,2) {};
\node[cycrec = 5.5, rotate=90] at (7.5,2) {};
\node[cycrec = 5.5, rotate=90] at (12,2) {};

\black (p1) at (0,0) [label=below:$l_{1}$] {};
\black (p2) at (4.5,0) [label=below:$l_{h - 1}$] {};
\black (pn) at (6,0) [label=below:$l_{h}$] {};
	
\white (v11) at (0,1) [label=right:$ $] {};
\white (v1k) at (0,4) [label=right:$ $] {};
	
\white (v21) at (4.5,1) [label=right:$ $] {};
\white (v2k) at (4.5,4) [label=right:$ $] {};
		
\white (vn1) at (6,1) [label=right:$ $] {};
\white (vnk) at (6,4) [label=right:$ $] {};

\black (r1) at (6,0) {};
\black (r2) at (7.5,0) [label=below:$r_{1}$] {};
\black (rn) at (12,0) [label=below:$r_{h - 1}$] {};

\white (w21) at (7.5,1) [label=right:$ $] {};
\white (w2k) at (7.5,4) [label=right:$ $] {};
		
\white (wn1) at (12,1) [label=right:$ $] {};
\white (wnk) at (12,4) [label=right:$ $] {};
		
\path 
(p1) -- node[auto=false]{$\ldots$} (p2)
(p1) edge[outarrow] (1.5,0)
(3,0) edge[outarrow] (p2)
(p2) edge[outarrow] (pn)
(pn) edge[outarrow, bend left = 25] (p1)
		
(v11) -- node[auto=true]{$\vdots$} (v1k)
		
(p1) edge[outarrow] (v11)
(v11) edge[outarrow] (0,2)
(0,3) edge[outarrow] (v1k)
(v1k) edge[outarrow, bend right = 20] (p1)
		
(v21) -- node[auto=true]{$\vdots$} (v2k)
		
(p2) edge[outarrow] (v21)
(v21) edge[outarrow] (4.5,2)
(4.5,3) edge[outarrow] (v2k)
(v2k) edge[outarrow, bend right = 20] (p2)
		
(vn1) -- node[auto=true]{$\vdots$} (vnk)
		
(pn) edge[outarrow] (vn1)
(vn1) edge[outarrow] (6,2)
(6,3) edge[outarrow] (vnk)
(vnk) edge[outarrow, bend right = 20] (pn)
;

\path 
(r2) -- node[auto=false]{$\ldots$} (rn)
(r1) edge[outarrow] (r2)
(r2) edge[outarrow] (9,0)
(10.5,0) edge[outarrow] (rn)
(rn) edge[outarrow, bend left = 25] (r1)

(w21) -- node[auto=true]{$\vdots$} (w2k)
		
(r2) edge[outarrow] (w21)
(w21) edge[outarrow] (7.5,2)
(7.5,3) edge[outarrow] (w2k)
(w2k) edge[outarrow, bend right = 20] (r2)
		
(wn1) -- node[auto=true]{$\vdots$} (wnk)
		
(rn) edge[outarrow] (wn1)
(wn1) edge[outarrow] (12,2)
(12,3) edge[outarrow] (wnk)
(wnk) edge[outarrow, bend right = 20] (rn)
;
\end{tikzpicture}
}
\caption{\label{figure:two-sided-comb}%
The agent digraph of a double
$(h, v)$-comb
instance that consists of
$2 h - 1$
'black agents' and
$(2 h - 1) (v - 1)$
'white agents'.
The (directed) cycles in the agent digraph are depicted by the stretched gray
ellipses.
}
\end{figure}
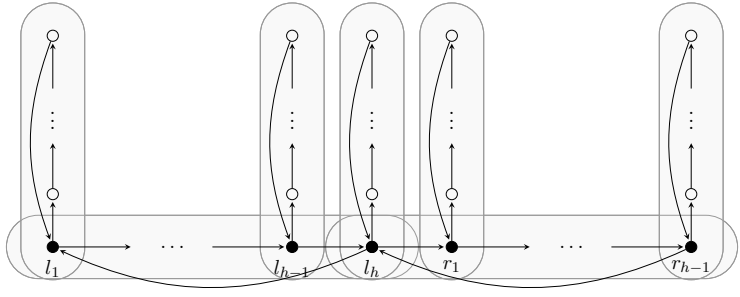

For the rest of this section we consider a deterministic truthful mechanism $f$ for the 
$(k, \lenUtil)$-BE 
problem.
Observe that 
$\Cycles_{f(\Wish)}$
can not include both horizontal cycles 
$c_{L}, c_{R}$ 
as they share a common participant.
Thus, we can assume without loss of generality (by the symmetry of $\Wish$) that
$c_{R} \notin \Cycles_{f(\Wish)}$.
We shall be interested in the output of $f$ to the agent digraph 
$\Wish^{h - 1}$,
depicted in Figure~\ref{figure:subproblem-two-sided-comb}
for illustration.
We will show that the assumption
$c_{R} \notin \Cycles_{f(\Wish)}$
implies that
$c_{R} \notin \Cycles_{f(\Wish^{h - 1})}$,
which will further results in
$f(\Wish^{h - 1}) = \pi_{L}$.
Consequently, the desired lower bound of 
$\criticalRatio(\lenUtil)$
will be directly established.

\begin{figure}
\centering
\scalebox{0.7}{
\begin{tikzpicture}
\node[cycrec = 8, fill = lightgray!10] at (3,0) {};
\node[cycrec = 8, fill = lightgray!10] at (9,0) {};
\node[cycrec = 5.5, rotate=90, fill = lightgray!10] at (0,2) {};
\node[cycrec = 5.5, rotate=90, fill = lightgray!10] at (4.5,2) {};
\node[cycrec = 5.5, rotate=90, fill = lightgray!10] at (6,2) {};
\node[cycrec = 8] at (3,0) {};
\node[cycrec = 8] at (9,0) {};
\node[cycrec = 5.5, rotate=90] at (0,2) {};
\node[cycrec = 5.5, rotate=90] at (4.5,2) {};
\node[cycrec = 5.5, rotate=90] at (6,2) {};

\black (p1) at (0,0) [label=below:$l_{1}$] {};
\black (p2) at (4.5,0) [label=below:$l_{h - 1}$] {};
\black (pn) at (6,0) [label=below:$l_{h}$] {};
	
\white (v11) at (0,1) [label=right:$ $] {};
\white (v1k) at (0,4) [label=right:$ $] {};
	
\white (v21) at (4.5,1) [label=right:$ $] {};
\white (v2k) at (4.5,4) [label=right:$ $] {};
		
\white (vn1) at (6,1) [label=right:$ $] {};
\white (vnk) at (6,4) [label=right:$ $] {};

\black (r1) at (6,0) {};
\black (r2) at (7.5,0) [label=below:$r_{1}$] {};
\black (rn) at (12,0) [label=below:$r_{h - 1}$] {};

\white (w21) at (7.5,1) [label=right:$ $] {};
\white (w2k) at (7.5,4) [label=right:$ $] {};
		
\white (wn1) at (12,1) [label=right:$ $] {};
\white (wnk) at (12,4) [label=right:$ $] {};
		
\path 
(p1) -- node[auto=false]{$\ldots$} (p2)
(p1) edge[outarrow] (1.5,0)
(3,0) edge[outarrow] (p2)
(p2) edge[outarrow] (pn)
(pn) edge[outarrow, bend left = 25] (p1)
		
(v11) -- node[auto=true]{$\vdots$} (v1k)
		
(p1) edge[outarrow] (v11)
(v11) edge[outarrow] (0,2)
(0,3) edge[outarrow] (v1k)
(v1k) edge[outarrow, bend right = 20] (p1)
		
(v21) -- node[auto=true]{$\vdots$} (v2k)
		
(p2) edge[outarrow] (v21)
(v21) edge[outarrow] (4.5,2)
(4.5,3) edge[outarrow] (v2k)
(v2k) edge[outarrow, bend right = 20] (p2)
		
(vn1) -- node[auto=true]{$\vdots$} (vnk)
		
(pn) edge[outarrow] (vn1)
(vn1) edge[outarrow] (6,2)
(6,3) edge[outarrow] (vnk)
(vnk) edge[outarrow, bend right = 20] (pn)
;

\path 
(r2) -- node[auto=false]{$\ldots$} (rn)
(r1) edge[outarrow] (r2)
(r2) edge[outarrow] (9,0)
(10.5,0) edge[outarrow] (rn)
(rn) edge[outarrow, bend left = 25] (r1)

(w21) -- node[auto=true]{$\vdots$} (w2k)
		
(r2) edge[outarrow, dashed] (w21)
(w21) edge[outarrow] (7.5,2)
(7.5,3) edge[outarrow] (w2k)
(w2k) edge[outarrow, bend right = 20] (r2)
		
(wn1) -- node[auto=true]{$\vdots$} (wnk)
		
(rn) edge[outarrow, dashed] (wn1)
(wn1) edge[outarrow] (12,2)
(12,3) edge[outarrow] (wnk)
(wnk) edge[outarrow, bend right = 20] (rn)
;
\end{tikzpicture}
}
\caption{\label{figure:subproblem-two-sided-comb}%
The agent digraph $\Wish^{h - 1}$.
The dashed arcs of agents
$\{r_{1}, \dots, r_{h - 1}\}$
represent the concealed arcs.}
\end{figure}
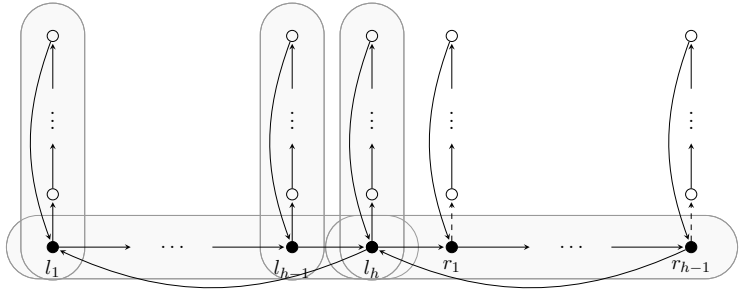

\begin{lemma} \label{lemma:two-sided-comb-not-output}
$c_{R} \notin \Cycles_{f(\Wish^{h - 1})}$.
\end{lemma}

\begin{proof}
We shall prove by induction that
$c_{R} \notin \Cycles_{f(\Wish^{i})}$
for all
$i = 0, 1, \dots, h - 1$
which yields the assertion.

For the induction base
$i = 0$,
we know that
$f(\Wish^{0}) = f(\Wish)$
and the induction claim follows as we assumed that 
$c_{R} \notin \Cycles_{f(\Wish)}$.
For the induction step, consider some 
$1 \leq i \leq h - 1$
and assume that 
$c_{R} \notin \Cycles_{f(\Wish^{i - 1})}$.
Consequently, 
since
$\Wish^{i} = (\Wish_{-r_{i}}^{i - 1}, \Wish'_{r_{i}})$,
the truthfulness of $f$ ensures that
\[
\utility(r_{i}, f(\Wish^{i})) = 
\utility(r_{i}, f(\Wish_{-r_{i}}^{i - 1}, \Wish'_{r_{i}})) \leq
\utility(r_{i}, f(\Wish^{i - 1})) <
\lenUtil(h) \, ,
\]
where the strict inequality follows as
$c_{R} \notin \Cycles_{f(\Wish^{i - 1})}$
and
$\lenUtil(v) < \lenUtil(h)$.
Hence, we deduce that
$c_{R} \notin f(\Wish^{i})$,
thus completing the induction step.
\end{proof}

By inspecting the incentives of agent $l_{h}$, we deduce the following lemma.

\begin{lemma} \label{lemma:two-sided-comb-output}
$f(\Wish^{h - 1}) = \pi_{L}$.
\end{lemma}

\begin{proof}
Let 
$\Wish'_{l_{h}}$
be the wish list of agent $l_{h}$ in which she conceals her outgoing arc to agent $r_{1}$ and reveals her other two.
Notice that the agent digraph
$\Wish^{h} = (\Wish^{h - 1}_{-l_{h}}, \Wish'_{l_{h}})$,
admits the same cycle structure as in the 
$(h, v)$-comb 
instance.
The proof of \Lem{}~\ref{lemma:comb-instance-output} can be applied, verbatim, to the agent digraph $\Wish^{h}$, concluding that 
$f(\Wish^{h}) = \pi_{L}$, thus  
$\utility(l_{h}, f(\Wish^{h})) = \lenUtil(h)$.
Hence, the truthfulness of $f$ ensures that 
\[
\utility(l_{h}, \Wish^{h - 1}) \geq 
\utility(l_{h}, (\Wish^{h - 1}_{-l_{h}}, \Wish'_{l_{h}})) =
\utility(l_{h}, \Wish^{h}) =
\lenUtil(h) \, ,
\]
which can be realized only if 
$\Cycles_{f(\Wish^{h - 1})}$
contains a cycle of length $h$ as
$\lenUtil(h) > \lenUtil(v)$.
Finally, by \Lem{}~\ref{lemma:two-sided-comb-not-output} we know that
$c_{R} \notin \Cycles_{f(\Wish^{h - 1})}$,
thus
$c_{L} \in \Cycles_{f(\Wish^{h - 1})}$,
which means that
$f(\Wish^{h - 1}) = \pi_{L}$.
\end{proof}

We are now ready to establish the lower bound for the non-uniform length function $\lenUtil$, by comparing the social welfare of the output 
$f(\Wish^{h - 1})$
to the social welfare of other exchanges that respect
$\Wish^{h - 1}$. 

\begin{proof}[Proof of \Thm{}~\ref{theorem:non-uniform-lower-bound}]
Let $f$ be a deterministic truthful mechanism to the 
$(k, \lenUtil)$-BE problem.
As we mentioned earlier, 
$\Cycles_{f(\Wish)}$
can not include both horizontal cycles 
$c_{L}, c_{R}$,
so assume that
$c_{R} \notin \Cycles_{f(\Wish)}$.
(If
$c_{R} \in \Cycles_{f(\Wish)}$
and
$c_{L} \notin \Cycles_{f(\Wish)}$
instead, then a symmetric reflection of the arguments in this section will derive the same end result.)
\Lem{}~\ref{lemma:two-sided-comb-output}
ensures that the social welfare of 
$f(\Wish^{h - 1})$ 
is
$h \cdot \lenUtil(h)$.
Consider the following two feasible exchanges that respect the agent digraph
$\Wish^{h - 1}$.\\
(1) The exchange that possess all vertical cycles (of 
$\Wish^{h - 1}$) 
in its trading cycle set, which admits a social welfare of 
$h \cdot v \cdot \lenUtil(v)$, 
and \\
(2)
the exchange for which its trading cycle set consists of the
$h - 1$
leftmost vertical cycles with the addition of the right horizontal cycle $c_{R}$, which admits a social welfare of 
$(h - 1) \cdot v \cdot \lenUtil(v) + h \cdot \lenUtil(h)$.\\
As the ratio between the social welfare of any exchange that respects
$\Wish^{h - 1}$
to the social welfare of the exchange 
$f(\Wish^{h - 1})$ 
establishes a lower bound, we obtain a lower bound of
\[
\max 
\left\{
\frac{h \cdot v \cdot \lenUtil(v)}{h \cdot \lenUtil(h)} ,
\frac{(h - 1) \cdot v \cdot \lenUtil(v) + h \cdot \lenUtil(h)}{h \cdot \lenUtil(h)}
\right\}
=
(h - 1) \frac{v \cdot \lenUtil(v)}{h \cdot \lenUtil(h)}
+
\max
\left\{
\frac{v \cdot \lenUtil(v)}{h \cdot \lenUtil(h)},
1
\right\} \, .
\]
Finally, going over all pairs
$(h, v)$
that satisfy
$\lenUtil(v) < \lenUtil(h)$
yields a lower bound of
\[
\max_{2 \leq \ell < \ell' \leq k \, : \, \lenUtil(\ell) > \lenUtil(\ell')}
\left\{
(\ell - 1) \frac{\ell' \cdot \lenUtil(\ell')}{\ell \cdot \lenUtil(\ell)}
+
\max
\left\{
1,
\frac{\ell' \cdot \lenUtil(\ell')}{\ell \cdot \lenUtil(\ell)}
\right\}
\right\}
\, = \,
\criticalRatio(\lenUtil)
\, ,
\]
thus completing the proof.
\end{proof}

\section{Lower Bound for Local Search Algorithms} 
\label{section:lower-bound-local-search}
In this section we establish an impossibility result for a class of local
search
$(k, \lenUtil)$-maxCGIS
algorithms referred to as \emph{standard local search algorithms}.
This result is based on the construction of $k$-cycle graphs all of whose nodes
admit the same length $k$, hence we obtain the same approximation lower bound
for all length functions $\lenUtil$.
Moreover, all $k$-cycle graphs considered in this section correspond to wish
list vectors, which means that our lower bound applies also to the class of
$(k, \lenUtil)$-BE
mechanisms derived (in the sense of
\Obs{}~\ref{observation:algorithm-implies-mechanism}) from standard local
search algorithms, referred to as \emph{standard local search mechanisms},
thus establishing \Thm{}~\ref{theorem:local-search-lower-bound}.

\sloppy
\begin{definition*}[standard]
A local search algorithm \Alg{} characterized by the improvement rule list
$R = (r_{1}, \dots, r_{|R|})$
is said to be \emph{standard} if it satisfies the following three properties. 
\begin{enumerate}

\item \label{property:expansion}
$r_{1} = r_{E}$. 

\item \label{property:locality}
For every
$n \in \Integers_{> 0}$,
$k$-cycle graph
$G = (V, E) \in \CycleGraphSet_{n}^{k}$,
partition of $V$ into
$V = V_{1} \uplus V_{2}$,
IS $I$ in $G$ and its induced partition
$I_{1} = I \cap V_{1}$,
$I_{2} = I \cap V_{2}$,
and
$1 \leq j \leq |R|$,
if
$r_{j}(G[V_{1}], I_{1}) = I'_{1} \neq \bot$
and
$I_{2} \cup I'_{1}$
is an IS,
then
$r_{j}(G, I) \neq \bot$.

\item \label{property:efficiency}
For every
$\epsilon > 0$,
there exists
$n_{0} = n_{0}(\epsilon) \in \Integers_{> 0}$
such that for every
$n \geq n_{0}$,
$k$-cycle graph
$G = (V, E) \in \CycleGraphSet_{n}^{k}$,
and
ISs $I$ and $I'$ in $G$,
if
$r_{j}(G, I) = I'$
for some
$1 \leq j \leq |R|$,
then
$|I' \triangle I| < \epsilon n$.\footnote{%
The binary operator $\triangle$ denotes the set theoretic symmetric difference.}

\end{enumerate}
\end{definition*}
\par\fussy

On an intuitive level, property \ref{property:locality} requires that if a local
improvement is possible, then so is a global improvement, whereas property
\ref{property:efficiency} requires that the number of nodes who switch their belonging status to the maintained IS in each step is small. 

Recall the agent digraph notion described in \Sect{}~\ref{section:lower-bound-non-uniform} as it will accompany as throughout this section.
The underlying agent digraph that we will take advantage of is the agent
digraph 
$\Wish^{H} \in \WishVecSpace_{k (k - 1) + 1}$
illustrated in Figure~\ref{figure:gadget}, in which all (directed) cycles admit the same length $k$.
In Figure~\ref{figure:gadget}, the agent digraph $\Wish^{H}$ is displayed in the left side while in the right side we present a succinct way to depict this digraph using only dots, who stands for the agents, and long ellipses who stands for the (directed) cycles. 
Let $H$ be the $k$-cycle graph corresponding to $\Wish^{H}$.
We determine the lexicographical order over the nodes in $H$ to be
\[
c_{1} \to 
c_{2} \to
\cdots \to
c_{k}
\, .
\]

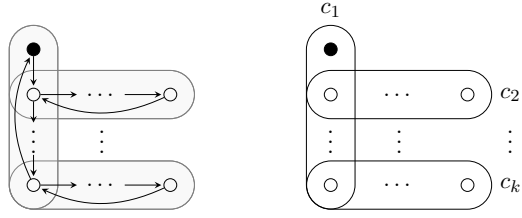
\begin{figure}
\centering
\scalebox{0.8}{
\begin{tikzpicture}
[cycrec/.style={draw=gray, rounded rectangle, minimum width=3.3cm,minimum height=0.8cm}]
\node[cycrec, fill = lightgray!10] at (1.5*0.75,-1*0.75) {};
\node[cycrec, fill = lightgray!10] at (1.5*0.75,-3*0.75) {};
\node[cycrec, rotate=90, fill = lightgray!10] at (0,-1.5*0.75) {};
\node[cycrec] at (1.5*0.75,-1*0.75) {};
\node[cycrec] at (1.5*0.75,-3*0.75) {};
\node[cycrec, rotate=90] at (0,-1.5*0.75) {};

\black (p1) at (0,0) [] {};
\white (c1) at (0,-1*0.75) [] {};
\white (c2) at (3*0.75,-1*0.75) [] {};
\white (d1) at (0,-3*0.75) [] {};
\white (d2) at (3*0.75,-3*0.75) [] {};

\path 
(p1) edge[outarrow] (c1)
(c1) edge[outarrow] (0,-2*0.75+0.25)
(0,-2*0.75-0.25) edge[outarrow] (d1)
(d1) edge[outarrow, bend left = 25] (p1)

(c1) edge[outarrow] (1*0.75,-1*0.75)
(2*0.75,-1*0.75) edge[outarrow] (c2)
(c2) edge[outarrow, bend left = 25] (c1)

(d1) edge[outarrow] (1*0.75,-3*0.75)
(2*0.75,-3*0.75) edge[outarrow] (d2)
(d2) edge[outarrow, bend left = 25] (d1)
;

\node[rotate=90] at (0,-2*0.75) {$\dots$};
\node[rotate=90] at (1.5*0.75,-2*0.75) {$\dots$};
\node at (1.5*0.75,-1*0.75) {$\dots$};
\node at (1.5*0.75,-3*0.75) {$\dots$};
\end{tikzpicture}
}
\hspace*{1cm}
\scalebox{0.8}{
\begin{tikzpicture}
[cycrec/.style={draw=black, rounded rectangle, minimum width=3.3cm,minimum height=0.8cm}]
\node[cycrec] at (1.5*0.75,-0.75) [label=right:$c_{2}$] {};
\node[cycrec] at (1.5*0.75,-3*0.75) [label=right:$c_{k}$] {};
\node[cycrec, rotate=90] at (0,-1.5*0.75) [label=right:$c_{1}$] {};

\black (p1) at (0,0) [] {};
\white (c1) at (0,-1*0.75) [] {};
\white (c2) at (3*0.75,-1*0.75) [] {};
\white (d1) at (0,-3*0.75) [] {};
\white (d2) at (3*0.75,-3*0.75) [] {};

\node[rotate=90] at (0,-2*0.75) {$\dots$};
\node[rotate=90] at (1.5*0.75,-2*0.75) {$\dots$};
\node[rotate=90] at (3.94*0.75,-2*0.75) {$\dots$};
\node at (1.5*0.75,-1*0.75) {$\dots$};
\node at (1.5*0.75,-3*0.75) {$\dots$};

\end{tikzpicture}
}
\caption{\label{figure:gadget}%
The agent digraph $\Wish^{H}$ (left) and its succinct representation (right).}
\end{figure}

For every fixed (large) integer
$N \in \Integers_{> 0}$,
let $G_{N}$ be the $k$-cycle graph corresponding to the agent digraph
$\Wish^{G_{N}} \in \WishVecSpace_{n(N)}$,
depicted (succinctly like $\Wish^{H}$) in
Figure~\ref{figure:local-search-lower-bound},
where
$n(N) = (k - 1)(k + 1 + 2N(k - 1)) + 2$.
We determine the lexicographical order over the nodes in $G_{N}$ to be
\[
c_{1} \to 
c_{2} \to
\cdots \to
c_{k} \to
b \to
A_{1} \to
B_{1} \to
A_{2} \to
B_{2} \to
\cdots \to
A_{N} \to
B_{N} \to
a \, ,
\]
where a set preceding another set implies a priority of all nodes in the first set to those in the second one (and an arbitrary order among nodes in the same set).
Observe that the length function $\lenUtil$ plays no role in the graphs
$H, G_{N}$
as all of their nodes admit the same length $k$.
We shall prove the following lemma, which its negation implies that any truthful standard local search algorithm $\Alg$ for the 
$(k, \lenUtil)$-maxCGIS
problem admits
$\Alg(H) = \{c_{1}\}$.
This will yield the desired lower bound of
$k - 1$
as the optimal IS in $H$ is 
$\{c_{2}, \dots, c_{k}\}$.

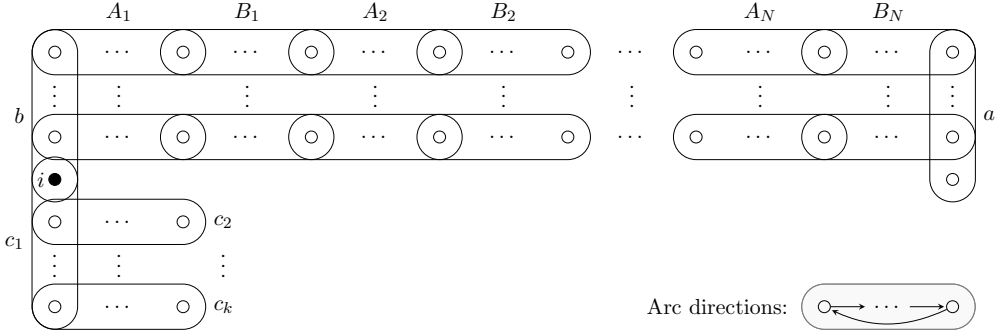
\begin{figure}
\centering
\scalebox{0.75}{
\begin{tikzpicture}
[cycrec/.style={draw=black, rounded rectangle, minimum width=3.3cm,minimum height=0.8cm}]
\node[cycrec] at (1.5*0.75,0.75) {};
\node[cycrec] at (4.5*0.75,0.75) {};
\node[cycrec] at (7.5*0.75,0.75) {};
\node[cycrec] at (10.5*0.75,0.75) {};
\node[cycrec] at (16.5*0.75,0.75) {};
\node[cycrec] at (19.5*0.75,0.75) {};

\node[cycrec] at (1.5*0.75,3*0.75) [label=above:$A_{1}$] {};
\node[cycrec] at (4.5*0.75,3*0.75) [label=above:$B_{1}$] {};
\node[cycrec] at (7.5*0.75,3*0.75) [label=above:$A_{2}$] {};
\node[cycrec] at (10.5*0.75,3*0.75) [label=above:$B_{2}$] {};
\node[cycrec] at (16.5*0.75,3*0.75) [label=above:$A_{N}$] {};
\node[cycrec] at (19.5*0.75,3*0.75) [label=above:$B_{N}$] {};

\node[cycrec] at (1.5*0.75,-0.75) [label=right:$c_{2}$] {};
\node[cycrec] at (1.5*0.75,-3*0.75) [label=right:$c_{k}$] {};

\node[cycrec, rotate=90] at (0,1.5*0.75) [label=above:$b$] {};
\node[cycrec, rotate=90] at (0,-1.5*0.75) [label=above:$c_{1}$] {};
\node[cycrec, rotate=90] at (21*0.75,1.5*0.75) [label=below:$a$] {};

\black (p1) at (0,0) [label={[label distance=-0.7mm]left:$i$}] {};
\white (p2) at (21*0.75,0) [] {};
	
\white (a1) at (0,1*0.75) [] {};
\white (a2) at (3*0.75,1*0.75) [] {};
\white (a3) at (6*0.75,1*0.75) [] {};
\white (a4) at (9*0.75,1*0.75) [] {};

\white (a5) at (12*0.75,1*0.75) [] {};
\white (a6) at (15*0.75,1*0.75) [] {};
\white (a7) at (18*0.75,1*0.75) [] {};
\white (a8) at (21*0.75,1*0.75) [] {};

\white (b1) at (0,3*0.75) [] {};
\white (b2) at (3*0.75,3*0.75) [] {};
\white (b3) at (6*0.75,3*0.75) [] {};
\white (b4) at (9*0.75,3*0.75) [] {};

\white (b5) at (12*0.75,3*0.75) [] {};
\white (b6) at (15*0.75,3*0.75) [] {};
\white (b7) at (18*0.75,3*0.75) [] {};
\white (b8) at (21*0.75,3*0.75) [] {};

\white (c1) at (0,-1*0.75) [] {};
\white (c2) at (3*0.75,-1*0.75) [] {};
\white (d1) at (0,-3*0.75) [] {};
\white (d2) at (3*0.75,-3*0.75) [] {};

\foreach \i in {0,1.5,4.5,7.5,10.5,13.5,16.5,19.5,21}{
\node[rotate=90] at (\i*0.75,2*0.75) {$\dots$};
}

\foreach \i in {1.5,4.5,7.5,10.5,13.5,16.5,19.5}{
\node at (\i*0.75,1*0.75) {$\dots$};
\node at (\i*0.75,3*0.75) {$\dots$};
}

\node[rotate=90] at (0,-2*0.75) {$\dots$};
\node[rotate=90] at (1.5*0.75,-2*0.75) {$\dots$};
\node[rotate=90] at (3.94*0.75,-2*0.75) {$\dots$};
\node at (1.5*0.75,-1*0.75) {$\dots$};
\node at (1.5*0.75,-3*0.75) {$\dots$};

\node[cycrec, draw=gray, fill = lightgray!10] at (19.5*0.75,-3*0.75) {};
\node[cycrec, draw=gray] at (19.5*0.75,-3*0.75) {};

\white (e1) at (18*0.75,-3*0.75) [] {};
\white (e2) at (21*0.75,-3*0.75) [] {};

\path
(e1) edge[outarrow] (19*0.75,-3*0.75)
(20*0.75,-3*0.75) edge[outarrow] (e2)
(e2) edge[outarrow, bend left = 25] (e1)
;

\node at (19.5*0.75,-3*0.75) {$\dots$};
\node at (15.5*0.75,-3*0.75) {Arc directions:};

\end{tikzpicture}
}
\caption{\label{figure:local-search-lower-bound}%
The agent digraph $\Wish^{G_{N}}$ represented in the succinct form.
Lower case letters represent single cycles (except for $i$ who represent the filled agent) and upper case letters represent the set of all horizontal cycles beneath them.
The arc directions for horizontal cycles are described in the bottom right corner, and we set the arc directions for vertical cycles to be a 90 degrees clockwise rotation of them.}
\end{figure}

\begin{lemma}\label{lemma:standard-not-truthful}
Consider a standard local search algorithm \Alg{} for the 
$(k, \lenUtil)$-maxCGIS
problem, characterized by the improvement rule list
$R = (r_{1}, \dots, r_{|R|})$.
If 
$r_{j}(H, \{c_{1}\}) \ne \bot$
for some
$1 \leq j \leq |R|$,
then $f$ is not truthful.
\end{lemma}

\begin{proof}
Let \Alg{} be a standard local search algorithm for the 
$(k, \lenUtil)$-maxCGIS
problem, characterized by the improvement rule list
$R = (r_{1}, \dots, r_{|R|})$,
for which
$r_{j}(H, \{c_{1}\}) \ne \bot$
for some
$1 \leq j \leq |R|$.
Let
$\epsilon \in (0, 1)$,
and let
$n_{0} = n_{0}(\epsilon)$
be the one promised in Property~\ref{property:efficiency} for \Alg{} and $\epsilon$.
Pick
$N = k \cdot n_{0}$,
and consider the $k$-cycle graph
$G_{N} = (V_{N}, E_{N})$.
Observe that the number of nodes in $G_{N}$ satisfies
$n(N) > n_{0}$.
We shall prove that agent $i$ improves her utility by hiding the node $c_{1}$ in $G_{N}$.

\sloppy
Assume first that agent $i$ does not conceal any node in $G_{N}$.
In that case, we know that property \ref{property:expansion} states that 
$r_{1} = r_{E}$,
and recalling the determined lexicographical order over the nodes in $G_{N}$, it follows that
$\Alg^{t}(G_{N}) = \{c_{1}\} \cup A_{1} \cup \dots \cup A_{N} \cup \{a\}$
for iteration 
$t = N (k - 1) + 2$.
By plugging
$G = G_{N}$,
$V_{1} = \{c_{1}, c_{2}, \dots, c_{k}\}$,
$I = \Alg^{t}(G_{N})$
into property~\ref{property:locality},
and by the assumption that
$r_{j}(G_{N}[\{c_{1}, c_{2}, \dots, c_{k}\}], \{c_{1}\}) = 
r_{j}(H, \{c_{1}\}) \ne \bot$,
we deduce that
$r_{j}(G_{N}, \Alg^{t}(G_{N})) \ne \bot$.
Since
$A_{1} \cup \dots \cup A_{N} \cup \{a\}$
is an optimal IS in the upper part of $G_{N}$ (the whole graph without the nodes $c_{1}, \dots ,c_{k}$),
it follows that
$c_{1} \notin \Alg^{t + 1}(G_{N})$
and that there exists  
$c_{j} \in \Alg^{t + 1}(G_{N})$,
for some
$2 \leq j \leq k$,
which implies that
$c_{1} \notin \Alg(G_{N})$.
The node $b$ can be included in $\Alg^{t'}(G_{N})$ for
$t' > t$
only if a complete horizontal row of nodes in the upper part of $G_{N}$ (that corresponds to a complete row of horizontal cycles in $\Wish^{G_{N}}$) alternate their belonging status to the maintained IS from $\Alg^{t}(G_{N})$ to $\Alg^{t'}(G_{N})$.
As each improvement rule improves the weight of the maintained IS by $\lenUtil(k)$, the maximum remaining iterations after iteration $t$ is 
$k - 2$.
Property \ref{property:efficiency} ensures that at most 
$\epsilon \cdot n_{0}$
nodes may switch their belonging status to $\Alg^{t'}(G_{N})$ at every iteration
$t' > 0$.
Hence, at most
$(k - 2) \cdot \epsilon \cdot n_{0} < N$
nodes may switch their belonging status to the maintained IS in the application of \Alg{} to $G_{N}$ after iteration $t$,
implying that
$b \notin \Alg(G_{N})$.
Therefore,
as both
$c_{1} \notin \Alg(G_{N})$
and
$b \notin \Alg(G_{N})$
we conclude that 
$\utility(i, \Alg(G_{N})) = 0$.
\par\fussy

We shall now examine the behavior of \Alg{} on the input 
$G_{N}[V_{N} - \{c_{1}\}]$,
in which agent $i$ hides the node $c_{1}$.
Since 
$r_{1} = r_{E}$,
the lexicographical order over the nodes in $G_{N}[V_{N} - \{c_{1}\}]$ (induced by that of $G_{N}$ by omitting $c_{1}$ from the top of the list) guarantees that
$\Alg^{t}(G_{N}[V_{N} - \{c_{1}\}]) = 
\{c_{2}, \dots, c_{k}\} \cup \{b\} \cup B_{1} \cup \dots \cup B_{N}$
for iteration 
$t = N (k - 1) + k$.
As 
$\Alg^{t}(G_{N}[V_{N} - \{c_{1}\}])$
is an optimal IS in $G_{N}[V_{N} - \{c_{1}\}]$, it follows that
$\Alg(G_{N}[V_{N} - \{c_{1}\}]) = \Alg^{t}(G_{N}[V_{N} - \{c_{1}\}])$.
Hence, 
since 
$b \in \Alg(G_{N}[V_{N} - \{c_{1}\}])$, we know that
$\utility(i, \Alg(G_{N}[V - \{c_{1}\}])) = \lenUtil(k)$.
Therefore, agent $i$ has an incentive to conceal the node $\{c_{1}\}$ that she partakes in, which means that \Alg{} is not truthful.
\end{proof}

\begin{corollary}
The 
$(k, \lenUtil)$-maxCGIS
problem does not admit a truthful standard local search algorithm with approximation ratio strictly smaller than
$k - 1$
for every length bound 
$k \geq 3$
and length function 
$\lenUtil$.
\end{corollary}

\begin{proof}
Let \Alg{} be a standard truthful local search algorithm to the
$(k, \lenUtil)$-maxCGIS
problem.
The negation of \Lem{}~\ref{lemma:standard-not-truthful} ensures that 
$\Alg(H) = \{c_{1}\}$,
thus
$\Weight(\Alg(H)) = \lenUtil(k)$.
However, the optimal IS for $H$ is
$O_{H} = \{c_{2}, \dots, c_{k}\}$
for which 
$\Weight(O_{H}) = (k - 1) \cdot \lenUtil(k)$.
The ratio between the weight of
$O_{H}$
to the weight of 
$\Alg(H)$
establishes the desired lower bound.
\end{proof}

\section{Additional Related Work}
\label{section:related-work}

\subsection*{%
Barter Exchange from a Game Theoretic View Point}
The game theoretic study of barter markets dates back to the seminal
work of Shapley and Scarf \cite{ShapleyS1974cores} that considers a model
where each agent is associated with a preference list over all other agents.
Shapley and Scarf introduced the classic top trading cycle (TTC) mechanism for
this setting whose truthfulness was established by Roth
\cite{roth1982incentive}.

The trend of looking at the kidney exchange market through a rigorous game
theoretic lens was initiated by Roth et al.~\cite{roth2004kidney}.
In contrast to the model considered in the current paper, the one studied in
\cite{roth2004kidney} includes (strict) preferences over the agent wish lists
as well as non-cyclic trades.
Another major difference is that in \cite{roth2004kidney}, Roth et al.\ allow
for trading cycles (and chains) of unbounded length.

As discussed in \Sect{}~\ref{section:introduction}, the model considered in
the current paper was introduced by Roth et al.\ in \cite{roth2005pairwise}.
In the language of our BE problem, each agent in \cite{roth2005pairwise}
corresponds to a pair consisting of a patient that needs a kidney transplant
and an incompatible donor that is willing to donate a healthy kidney. 
The inclusion of agent $i$ in the wish list of agent $j$ corresponds in this
regard to a compatibility between the healthy donor associated with agent $i$
and the patient associated with agent $j$.
Among other results, Roth et al.\ reduce (implicitly) the task of designing a
truthful Pareto efficient mechanism for the case of trading cycles of length
$k = 2$
to that of constructing a maximum size matching in an undirected graph.
They also prove that in the scope of deterministic mechanisms, Pareto
efficiency is equivalent to maximizing the social welfare, which leads to the
$(k, \lenUtil_{k}^{u})$-BE
problem defined in the current paper.
Hatfield \cite{hatfield2005pairwise} extends some of the arguments of
\cite{roth2005pairwise} to trading cycles of length at most
$k \geq 3$
(among other constraints), which enables the development of an optimal, alas
computationally inefficient, truthful mechanism for the
$(k, \lenUtil_{k}^{u})$-BE
problem.

In \cite{roth2007efficient}, Roth et al.\ analyze the social welfare gain
obtained from increasing the trading cycles length bound $k$ assuming that the
donation compatibility is determined by the $4$ possible blood types.
They prove that the increase from
$k = 2$
to
$k = 3$
leads to a dramatic social welfare gain and that this gain is substantially
larger than the one obtained by further increasing the trading cycle length
bound from
$k = 3$
to
$k = 4$.
Moreover, any further increase in $k$ does not improve the optimal social
welfare.

We emphasize that despite the extensive treatment that the
$(k, \lenUtil_{k}^{u})$-BE
problem received in the last two decades, the current paper is the first to
develop efficient truthful mechanisms for this problem with a non-trivial
approximation ratio.
To the best of our knowledge, the extension to non-uniform length functions
$\lenUtil \neq \lenUtil_{k}^{u}$
has not been studied so far.

Regarding the
$(k, \lenUtil_{k}^{u})$-maxCGIS
problem, the notion of INPA algorithms, which is pivotal to the truthfulness
proofs in the current paper, is closely related to the notion of
``consistent'' mechanisms, introduced in \cite{hatfield2005pairwise}.
Consistency states that if an exchange is chosen from some set of feasible
exchanges, then it is also chosen from any subset of that set.
From the viewpoint of the 
$(k, \lenUtil_{k}^{u})$-maxCGIS
problem, INPA is a weak version of consistency as it merely suppresses
manipulations of non-partaking agents.
Other related concepts that limit the effect of individual manipulations on
the group outcome (alas in the context of binary outcomes) include the notions
of ``bitonic'' mechanisms \cite{mu2008truthful} and ``XBONE'' and
``non-bossy'' mechanisms \cite{akbarpour2020investment}.

Abbassi et al.~\cite{abbassi2015exchange} introduced a generalization of our
$(k, \lenUtil_{k}^{u})$-BE
problem, where each agent holds a set of items that she wishes to exchange for
certain items held by the other agents.
As in our BE problem, the exchanges are carried out via trading cycles,
however, an agent may now participate in multiple trading cycles and her
utility is the number of wished items she obtains.
The main positive results of Abbassi et al.\ are
(1)
an efficient optimal truthful mechanism for the problem version with trading
cycles of unbounded length;
and
(2)
an efficient truthful mechanism that approximates the optimal solution within
ratio $8$ for pairwise exchanges, i.e., for trading cycles of length
$k = 2$.
The version with trading cycles of length at most
$k \geq 3$,
that is closer to the focus of the current paper, is treated in
\cite{abbassi2015exchange} only on the negative side, establishing
approximability lower bounds both for the combinatorial optimization problem
and for the task of designing truthful mechanisms. 

In a related problem, studied by Ashlagi et al.\
\cite{ashlagi2015mix, ashlagi2014free},
the strategic entities are hospitals, each responsible for a subset of
patient-donor pairs.
The objective of a hospital is to maximize the number of patients under its
responsibility that receive a donation through pairwise exchanges
(i.e.,
$k = 2$)
and to this end, the hospital may choose to match some of their patient-donor
pairs internally, rather than revealing them to the mechanism so that they
become available for inter-hospital matches.
Another variant of the barter exchange setting studied recently by Dickerson et
al.~\cite{dickerson2019failure} is the one where each arc in the agent digraph
is associated with a failure probability.

Also related to the
$(k, \lenUtil_{k}^{u})$-maxCGIS
problem are the \emph{hedonic coalition formation} games in which a mechanism
partitions the players into disjoint sets (coalitions) and each player's
utility is fully determined by the other members of her
coalition~\cite{woeginger2013core}.
Specifically, our maxCGIS problem resembles the case of \emph{anonymous
preferences}, where the players are indifferent about coalitions of the same
size~\cite{woeginger2013core}.
The (informally defined)
$(\infty, \lenUtil_{\infty}^{u})$-maxCGIS
problem falls under the family of \emph{group activity selection problems
(GASPs)}~\cite{darmann2012group}.
The goal in GASP is to assign agents to activities in a ``good way'', where
agents have preferences over
$\langle \text{activity}, \text{group size} \rangle$
pairs.
Specifically, the
$(\infty, \lenUtil_{\infty}^{u})$-maxCGIS
problem corresponds to the special case in which agents are indifferent
between ``accepted" coalitions'' and there is a single $\infty$-copyable
activity.

The main difference between the maxCGIS problem and the ones mentioned in the
previous paragraph is our hard bound on the coalition size (or, using the
maxCGIS terminology, the number of agents included in the label of each node).
A setting that includes such a hard bound was recently investigated in
\cite{DarmannDDS2022} in the context of \emph{simplified GASPs}, where
agents only express their preferences over the set of activities with
restrictions on the number of participants in each activity.
Notice that the simplified GASP setting differs from our maxCGIS problem as
the agents in the latter problem admit (cardinal) preferences over the
coalitions (based on their size), rather than the activities.

\subsection*{%
Barter Exchange from a Combinatorial Optimization View Point }
The family of BE problems turns out to admit an interesting combinatorial
structure even when one focuses merely on the construction of a social welfare
maximizing exchange, putting truthfulness considerations aside.
While the
$k = 2$
case reduces to maximum size matching, and hence admits an efficient
algorithm, Abraham et al.~\cite{abraham2007clearing} proved that the
$(k, \lenUtil)$-BE
problem is NP-hard for every length function $\lenUtil$ already for
$k = 3$.
Biro et al.~\cite{biro2009maximum} enhanced this hardness result and proved
that the problem is also APX-hard.
An explicit approximation lower bound of
$698 / 697$
was established by Luo et al.~\cite{LuoTWZ2016approximation}.
We note that the original constructions in
\cite{abraham2007clearing, biro2009maximum, LuoTWZ2016approximation}
are designed for the case of uniform length functions
$\lenUtil = \lenUtil_{k}^{u}$,
however the hard instances they use can be adjusted so that all trading cycles
in $\Cycles_{\Wish}^{k}$ are of length (exactly) $k$, which means that
the problem becomes oblivious to the length function $\lenUtil$.

The instances of the
$(k, \lenUtil)$-BE
problem where all trading cycles in $\Cycles_{\Wish}^{k}$ are of length $k$ is
a special case of the classic $k$-set packing ($k$-SP) problem.
As proved by Hazan et al.\
\cite{hazan2003hardness, hazan2006complexity},
this problem is hard to approximate within ratio better then
$\Omega(\frac{k}{\ln k})$.
It is not clear if this lower bound can be extended also to our
$(k, \lenUtil)$-BE
problem.

On the positive side, it is well known that the greedy approach leads to a
$k$-approximation of $k$-SP.
A folklore result (see, e.g., \cite{halldorsson1998approximations}) states that
the local search algorithm relying on the $2$-for-$1$ policy provides a
$\frac{k + 1}{2}$-approximation.
As proved in \cite{hurkens1989size}, the approximation ratio improves to
$\frac{k}{2} + \varepsilon$
if the $2$-for-$1$ policy is replaced by the $t$-local search policy.
More recently, Sviridenko and Ward \cite{sviridenko2013large} and Cygan
\cite{cygan2013improved} used more sophisticated local search algorithms to
further improve the approximation ratio upper bound to
$\frac{k + 2}{3}$
and
$\frac{k + 1}{3} + \varepsilon$,
respectively.

As discussed in \Sect{}~\ref{section:preliminaries} (see \Obs{}\
\ref{observation:algorithm-implies-mechanism} and
\ref{observation:cycle-graph-structure}), the combinatorial optimization
aspect of the
$(k, \lenUtil)$-BE
problem reduces to the problem of constructing a maximum weight IS in
$(k + 1)$-claw
free graphs (this is true also for the $k$-SP problem).
This problem has received considerable attention and several approximation
algorithms were developed for it.
Specifically, Arkin and Hassin \cite{arkin1998local} analyze the $t$-local
search policy in
$(k + 1)$-claw
free graphs and prove that it is guaranteed to approximate the maximum weight
IS within a
$k - 1 + \varepsilon$
ratio.
Biro et al.~\cite{biro2009maximum}, who investigated the maximum weight IS
problem in
$(k + 1)$-claw
free graphs specifically in the context of the
$(k, \lenUtil_{k}^{u})$-BE
problem, obtain the same approximation guarantee with another local search
algorithm which is simpler to analyze.
The approximation ratio upper bound was improved to
$\frac{2 (k + 1)}{3} + \varepsilon$
by Chandra and Halld\'{o}rsson
\cite{chandra2001greedy}
and to
$\frac{k + 1}{2} + \varepsilon$
by Berman \cite{berman2000d};
these two works also use local search algorithm alas with slightly more
advanced policies.

It is important to point out that while there is a resemblance between the
improvement rules used by our local search algorithm \LocalSearchAlg{} and some
of the existing local search policies, our improvement rules must take into
account the strings associated with each node, a feature that becomes
effective only when one aims for a truthful algorithm.
Consequently, the approximation ratio analysis of \LocalSearchAlg{} (see
\Thm{}~\ref{theorem:uniform-main-algorithm-approximation}) follows a
different path in comparison to the analyses used in the existing literature.
Moreover, the algorithms developed in
\cite{arkin1998local, berman2000d, chandra2001greedy, biro2009maximum}
belong to the family of standard local search algorithms, hence, by
\Thm{}~\ref{theorem:local-search-lower-bound}, there is no hope in relying on
them to obtain a truthful algorithm with approximation ratio better than
$k - 1$.

\clearpage
\bibliographystyle{alpha}
\bibliography{references}

\clearpage
\appendix
\begin{figure}[!t]
{\centering
\Large{\textbf{APPENDIX}}
\par}
\end{figure}

\section{Lifting the Subset Wish List Assumption}
\label{appendix:subset-assumption}
Recall that in \Sect{}~\ref{section:model}, we make the subset wish list
assumption stating that the wish list $\Wish_{i}$ reported by an agent
$i \in [n]$
must be a subset of her true wish list $\Wish^{*}_{i}$.
In this section, we present a simple method for lifting this assumption by
developing a reduction that transforms any ``reasonable'' (in a sense revealed
soon) truthful
$(k, \lenUtil)$-BE
mechanism $f$ working under the subset wish list assumption into a randomized
truthful
$(k, \lenUtil)$-BE
mechanism $f'$ that does not rely on the subset wish list assumption, while
increasing the guaranteed approximation ratio by no more than a
$(1 + \zeta)$-factor
for an arbitrarily small parameter
$\zeta > 0$.
This transformation is designed so that $f'$ is ``almost deterministic'' in
the sense that the amount of entropy in the image of $f'$ is arbitrarily close
to $0$.
Moreover, if $f$ is computationally efficient, than so is $f'$.

\subsection*{The Assumption's Critical Role}
Before presenting the aforementioned reduction, let us explain the critical
role that the subset wish list assumption plays in the mechanisms developed in
the current paper.
Recall that our 
$(k, \lenUtil)$-BE
mechanisms are developed as 
$(k, \lenUtil)$-maxCGIS
algorithms $\Alg$ and that a key feature of these algorithms is that they obey
the INPA property.
This property ensures that given a $k$-cycle graph
$G = (V, E) \in \CycleGraphSet_{n}^{k}$
and an agent
$i \in [n]$,
if
$i \notin \Agents(\Alg(G))$,
then
$\Alg(G[V - S]) = \Alg(G)$
for any node subset
$S \subseteq \Agents^{-1}(i)$.
In other words, a non-partaking agent $i$ cannot affect the outcome of the
algorithm by removing some of its nodes.

The INPA property on the other hand, does not give us any guarantees regarding
agents
$i \in [n]$
that do partake in $\Alg(G)$.
In particular, if
$i \in \Agents(\Alg(G))$,
then we cannot rule out that
$i \notin \Alg(G[V - S])$
for some node subsets
$S \subseteq \Agents^{-1}(i)$
even if
$S \cap \Alg(G) = \emptyset$;
in fact, simple examples demonstrate that this phenomenon exists in our main
algorithm $\LocalSearchAlg$ and its derivatives.
Looking at it from the converse direction, agent $i$, that does not partake in
the outcome
$\Alg(G[V - S])$,
may manipulate the algorithm, adding the node subset $S$ to the graph
$G[V - S]$,
so that it does partake in
$\Alg(G)$
with
$i \in \Agents(\Alg(G) - S)$,
thus gaining from the manipulation.

This is exactly where the subset wish list assumption comes into play:
It ``guards'' our maxCGIS algorithms from such manipulations,
allowing us to concentrate on the removal of node subsets when trying to
establish truthfulness.

\subsection*{The Reduction}
Consider a truthful mechanism $f$ for the
$(k, \lenUtil)$-BE
problem, working under the subset wish list assumption.
Suppose that $f$ is ``reasonable'' in the sense that it is not affected by the
omission of agent $j$ from wish list $\Wish_{i}$ if there does not exist any
exchange
$\pi \in \ExcSet_{\Wish}^{k}$
such that
$\pi(i) = j$.\footnote{%
We emphasize that all mechanisms presented in the current paper satisfy this
property.}
Let
$\zeta > 0$ 
be an arbitrarily small real number
and consider an input wish list vector 
$\Wish \in \WishVecSpace_{n}$,
$n \in \Integers_{> 0}$.
For each
$c \in \Cycles_{\Wish}^{k}$,
let $\pi_{c}$ denote the exchange satisfying
$\Cycles_{\pi_{c}} = \{ c \}$.

The mechanism $f'$ is constructed as follows:
Given a wish list vector
$\Wish \in \WishVecSpace_{n}$,
$n \in \Integers_{> 0}$,
$f'$ returns the exchange
$\pi = f'(\Wish)$
chosen randomly via the following process:
\begin{enumerate}

\item
With probability
$1 - \zeta$
set
$\pi = f(\Wish)$.

\item 
With probability $\zeta$,
pick
$\ell \in \{2, 3, \dots, k\}$
uniformly at random,
then choose a cycle from
$\{c \in \Cycles_{\Wish}^{k} \mid \length(c) = \ell\}$
uniformly at random and set
$\pi = \pi_{c}$.
If $\{c \in \Cycles_{\Wish}^{k} \mid \length(c) = \ell\} = \emptyset$, then set $\pi$ to be the identity bijection over $[n]$.

\end{enumerate}

Let us now show that $f'$ admits an approximation ratio of 
$\frac{r}{1 - \zeta}$,
where $r$ is the approximation ratio that $f$ admits,
and afterwards we shall prove that $f'$ is truthful.
First, consider some 
$\Wish^{*} \in \WishVecSpace_{n}$.
As each 
$c \in \Cycles_{\Wish^{*}}^{k}$
grants a non-negative utility to all of its partaking agents, combined with the approximation ratio guarantee of $f$, we deduce that
\[
\Ex \left( \SW(f'(\Wish^{*})) \right)
\, \geq \,
(1 - \zeta) \cdot \Ex \left( \SW(f(\Wish^{*})) \right)
\, \geq \,
(1 - \zeta) \cdot
\frac{1}{r}
\cdot
\max \left\{
\SW(\pi) \, : \, \pi \in \ExcSet_{\Wish^{*}}^{k}
\right\}
\, ,
\]
thus establishing the desired approximation ratio of $f'$.

To see that $f'$ is truthful, consider a wish list vector
$\Wish \in \WishVecSpace_{n}$,
$n \in \Integers_{> 0}$,
and an agent 
$i \in [n]$.
We shall prove that
\begin{equation} \label{equation:appendix-A}
\Ex \left( \utility(i, f'(\Wish)) \right)
\, \leq \,
\Ex \left( \utility \left(
i,
f' \left( \Wish_{-i}, \Wish^{*}_{i} \right)
\right) \right)
\, .
\end{equation}

If 
$\Cycles_{\Wish}^{k} \not \subseteq \Cycles_{(\Wish_{-i}, \Wish^{*}_{i})}^{k}$
then there exists a cycle
$c' \in \Cycles_{\Wish}^{k} - \Cycles_{(\Wish_{-i}, \Wish^{*}_{i})}^{k}$
in which agent $i$ is assigned to a non-wished agent.
We know that
$\utility(i, \pi_{c'}) = -\infty$,
which implies that
$\Ex \left( \utility(i, f'(\Wish)) \right) = -\infty$
as $\pi_{c'}$ has a positive utility to be the realization of $f'(\Wish)$.
This yields inequality~\eqref{equation:appendix-A} as $f$ is individually rational and, by definition, there is no cycle in 
$\Cycles_{(\Wish_{-i}, \Wish^{*}_{i})}^{k}$
that assigns $i$ to a non-wished agent.

If  
$\Cycles_{\Wish}^{k} \subseteq \Cycles_{(\Wish_{-i}, \Wish^{*}_{i})}^{k}$,
then there is no exchange 
$\pi \in \ExcSet_{\Wish}^{k}$
such that
$\pi(i) = j$
for any agent
$j \in \Wish_{i} - \Wish^{*}_{i}$.
Therefore,
as $f$ is ``reasonable'',
we deduce that
$
f(\Wish)
=
f \left( \Wish_{-i}, \Wish_{i} \cap \Wish^{*}_{i} \right)
$.
Hence, the truthfulness of $f$ guarantees that 
\[
\Ex \left( \utility(i, f(\Wish)) \right)
=
\Ex \left( \utility \left(
i,
f \left( \Wish_{-i}, \Wish_{i} \cap \Wish^{*}_{i} \right)
\right) \right)
\leq
\Ex \left( \utility \left(
i,
f \left( \Wish_{-i}, \Wish^{*}_{i} \right)
\right) \right)
\, .
\]
Finally, the fact that 
$\{c \in \Cycles_{\Wish}^{k} \mid \length(c) = \ell\} \subseteq
\{c \in \Cycles_{(\Wish_{-i}, \Wish^{*}_{i})}^{k} \mid \length(c) = \ell\}$
for every
$\ell = 2, 3, \dots, k$
derives~\eqref{equation:appendix-A}.

\section{Elaborating on $\criticalRatio(\lenUtil)$}
\label{appendix:critical-ratio}
Consider a non-uniform length function
$\lenUtil : \{ 2, 3, \dots, k \} \rightarrow (0, 1]$
and recall the definition of 
$\criticalRatio(\lambda)$
as presented in \Sect{}~\ref{section:contribution}:
\[
\criticalRatio(\lenUtil)
\, = \,
\max_{%
2 \leq \ell < \ell' \leq k \, : \, \lenUtil(\ell) > \lenUtil(\ell')%
}
\left\{
\max \left\{
\frac{\ell' \cdot \lenUtil(\ell')}{\lenUtil(\ell)}, \,
\frac{\ell - 1}{\ell}
\cdot
\frac{\ell' \cdot \lenUtil(\ell')}{\lenUtil(\ell)}
+ 1
\right\}
\right\} 
\, = \,
\max
\left\{
\criticalRatio_{1}(\lenUtil), 
\criticalRatio_{2}(\lenUtil)
\right\}
\, ,
\]
where
\[
\criticalRatio_{1}(\lenUtil)
\, = \,
\max_{%
2 \leq \ell < \ell' \leq k \, : \, \lenUtil(\ell) > \lenUtil(\ell')%
}
\left\{
\frac{\ell' \cdot \lenUtil(\ell')}{\lenUtil(\ell)}
\right\}
\]
and
\[
\criticalRatio_{2}(\lenUtil)
\, = \,
\max_{%
2 \leq \ell < \ell' \leq k \, : \, \lenUtil(\ell) > \lenUtil(\ell')%
}
\left\{
\frac{\ell - 1}{\ell}
\cdot
\frac{\ell' \cdot \lenUtil(\ell')}{\lenUtil(\ell)}
+ 1
\right\} \, .
\]
We begin our discussion of $\criticalRatio(\lenUtil)$ by establishing the
following (easy) lemma that provides a sanity check that \Thm{}
\ref{theorem:greedy-mechanism} and \ref{theorem:non-uniform-lower-bound} do
not contradict.

\begin{lemma} \label{lemma:critical-ratio:sanity-check}
The expression $\criticalRatio(\lenUtil)$ is (strictly) smaller than $k$.
\end{lemma}
\begin{proof}
Since the component $\criticalRatio_{1}(\lenUtil)$ is clearly smaller than
$k$, it remains to show that
$\criticalRatio_{2}(\lenUtil) < k$
or, equivalently, that
\[
\max_{%
2 \leq \ell < \ell' \leq k \, : \, \lenUtil(\ell) > \lenUtil(\ell')%
}
\left\{
\frac{\ell - 1}{\ell}
\cdot
\frac{\ell' \cdot \lenUtil(\ell')}{\lenUtil(\ell)}
\right\}
\, < \,
k - 1
\, .
\]
The left-hand side of the last inequality is up-bounded by
\[
\frac{k - 2}{k - 1} \cdot \criticalRatio_{1}(\lenUtil)
\, < \,
\frac{k - 2}{k - 1} \cdot k
\]
which is indeed smaller than
$k - 1$.
\end{proof}

Next, we turn our attention to an intuitive explanation of the
$\criticalRatio(\lenUtil)$-lower bound promised in
\Thm{}~\ref{theorem:non-uniform-lower-bound}.
To this end, it is advisable to take a look at
Figure~\ref{figure:comb-instance};
this figure demonstrates that the approximation ratio of any truthful
$(k, \lenUtil)$-BE
mechanism is low-bounded by the component
$\criticalRatio_{1}(\lenUtil)$
that captures the ratio of the social welfare obtained from $\ell$ trading
cycles of length $\ell'$ to to that obtained from a single trading cycle of
length $\ell$.
The component
$\criticalRatio_{2}(\lenUtil)$
stems from the enhanced construction presented in
Figure~\ref{figure:subproblem-two-sided-comb}, which involves a bit more
manipulations with the design of the agent digraph.

\subsection*{Matching Lower and Upper Bounds}
Since $\criticalRatio(\lenUtil)$ appears in both the lower bound of
\Thm{}~\ref{theorem:non-uniform-lower-bound} and the upper bound of
\Thm{}~\ref{theorem:non-uniform-upper-bound}, in the remainder of this
section, we wish to reveal the
properties of the length function $\lenUtil$ that ensure that the two bounds
match or, equivalently, the properties that yield
\[
\criticalRatio(\lenUtil)
\, > \,
k - 1
\, .
\]

First, let us understand when each of the components $\criticalRatio_{1}$ and
$\criticalRatio_{2}$ is dominant.
As the intuition from Figure~\ref{figure:subproblem-two-sided-comb} and the
corresponding lower bound proof in Section~\ref{section:strong-lower-bound}
suggest, the trade-off between $\criticalRatio_{1}$ and $\criticalRatio_{2}$ is
the trade-off between the social welfare obtained from the rightmost vertical
trading cycle in Figure~\ref{figure:subproblem-two-sided-comb} and that
obtained from the right horizontal trading cycle in
Figure~\ref{figure:subproblem-two-sided-comb};
that is,
\[
\criticalRatio_{1}(\lenUtil) \ge \criticalRatio_{2}(\lenUtil)
\iff
\ell' \cdot \lenUtil(\ell') \ge \ell \cdot \lenUtil(\ell)
\]
for the choice of
$2 \leq \ell < \ell' \leq k$,
$\lenUtil(\ell) > \lenUtil(\ell')$,
for which
\[
\criticalRatio(\lambda) 
=
\max \left\{
\frac{\ell' \cdot \lenUtil(\ell')}{\lenUtil(\ell)}, \,
\frac{\ell - 1}{\ell}
\cdot
\frac{\ell' \cdot \lenUtil(\ell')}{\lenUtil(\ell)}
+ 1
\right\} \, .
\]
Hence, if
$\criticalRatio_{1}(\lenUtil) < \criticalRatio_{2}(\lenUtil)$,
then
\begin{align*}
\criticalRatio(\lenUtil) 
\, = \,
\criticalRatio_{2}(\lenUtil)
\, = \, &
\max_{%
2 \leq \ell < \ell' \leq k \, : \, \lenUtil(\ell) > \lenUtil(\ell')%
}
\left\{
\frac{\ell - 1}{\ell}
\cdot
\frac{\ell' \cdot \lenUtil(\ell')}{\lenUtil(\ell)}
+ 1
\right\}
\\
< \, &
\max_{%
2 \leq \ell < \ell' \leq k \, : \, \lenUtil(\ell) > \lenUtil(\ell')%
}
\left\{
\frac{\ell - 1}{\ell}
\cdot
\frac{\ell \cdot \lenUtil(\ell)}{\lenUtil(\ell)}
+ 1
\right\}
\\
= \, &
\max_{%
2 \leq \ell < \ell' \leq k \, : \, \lenUtil(\ell) > \lenUtil(\ell')%
}
\{ \ell -1 + 1 \}
\, = \,
k - 1
\, .
\end{align*}
Consequently, the inequality
$\criticalRatio(\lenUtil) > k - 1$
can hold only if
$\criticalRatio(\lambda)
=
\criticalRatio_{1}(\lambda)
\geq
\criticalRatio_{2}(\lambda)$.
Moreover, it is necessary that 
$\ell' = k$
since otherwise 
\[
\criticalRatio_{1}(\lenUtil)
\, = \,
\max_{%
2 \leq \ell < \ell' \leq k - 1 \, : \, \lenUtil(\ell) > \lenUtil(\ell')%
}
\left\{
\frac{\ell' \cdot \lenUtil(\ell')}{\lenUtil(\ell)}
\right\}
\, < \,
\max_{%
2 \leq \ell < \ell' \leq k - 1 \, : \, \lenUtil(\ell) > \lenUtil(\ell')%
}
\ell'
\, \leq \,
k - 1
\, ,
\]
thus establishing the following lemma.

\begin{lemma} \label{lemma:critical-ratio:matching-bounds-condition}
The expression $\criticalRatio(\lenUtil)$ satisfies
$\criticalRatio(\lenUtil) > k - 1$
if and only if
\[
\frac{\lenUtil(k)}{\lenUtil(\ell^{*})}
\, > \,
\frac{k - 1}{k}
\, = \,
1 - \frac{1}{k}
\, ,
\]
recalling that $\ell^{*}$ is the largest
$2 \leq \ell \leq k - 1$
that satisfies
$\lenUtil(\ell) > \lenUtil(k)$.
\end{lemma}

In other words, the ratio of the smallest possible agent utility, i.e., the
one obtained from a trading cycle of length $k$, to the second smallest agent
utility (obtained from a  trading cycle of length $\ell^{*}$) must be
strictly larger than
$1 - \frac{1}{k}$.
So, for the condition
$\frac{\lenUtil(k)}{\lenUtil(\ell^{*})}
>
1 - \frac{1}{k}$
presented in \Lem{}~\ref{lemma:critical-ratio:matching-bounds-condition} to
hold, the length function $\lenUtil$ must ``tail (sufficiently) flatly''.

Any ``near-uniform'' length function $\lenUtil$ in particular tails flatly and
thus, satisfies the condition of
\Lem{}~\ref{lemma:critical-ratio:matching-bounds-condition}.
An example for a length function $\lenUtil$ that tails sufficiently flatly
although it is far from being ``near-uniform'' is the exponential length
function defined as
$\lenUtil(\ell) = C \cdot (1 - \epsilon)^{\ell}$,
where
$0 < \epsilon < \frac{1}{k}$
and 
$0 < C \leq 1$.
We note that although this exponential length function is convex, in general,
convex functions (let alone concave functions) do not necessarily satisfy the
condition of \Lem{}~\ref{lemma:critical-ratio:matching-bounds-condition} as
they may fail to tail sufficiently flatly.



\section{Cycle Graphs are Richer than Wish List Vectors}
\label{appendix:cycle-graph-with-no-wish-list-vector}
In this section we present a $k$-cycle graph that does not correspond to any
wish list vector.
Let
$k = 3$,
$n = 6$,
and consider the $3$-cycle graph
$G = (V, E) \in \CycleGraphSet_{6}^{3}$
characterized by the node set
$V = \{ v_{1}, v_{2}, v_{3}, v_{4} \}$
and the following agent assignments:
\[
\Agents(v_{1}) = \{ 1, 2 \}
\, , \quad
\Agents(v_{2}) = \{ 2, 3 \}
\, , \quad
\Agents(v_{3}) = \{ 1, 3 \}
\, , \quad \text{and} \quad
\Agents(v_{4}) = \{ 4, 5, 6 \}
\]
(in any cyclic order).

Assume towards contradiction that $G$ corresponds to the wish list
$\Wish \in \WishVecSpace_{4}$.
By definition,
the agent assignment of $v_{1}$ implies that
$2 \in \Wish_{1}$,
the agent assignment of $v_{2}$ implies that
$3 \in \Wish_{2}$,
and
the agent assignment of $v_{3}$ implies that
$1 \in \Wish_{3}$.
Hence, we deduce that the cycle corresponding to the cyclic order
$( 1, 2, 3 )$
is included in $\Cycles_{\Wish}^{k}$, deriving a contradiction as there is no
node
$v \in V$
for whom
$\Agents(v) = \{ 1, 2, 3 \}$.

\end{document}